\documentclass{amsart}
\usepackage[utf8]{inputenc}
\usepackage{amsthm}
\usepackage{amssymb}
\usepackage{bbold}
\usepackage{thmtools}
\usepackage{amsmath}
\usepackage{graphicx}

\usepackage{algorithm}
\usepackage{algpseudocode} 

\theoremstyle{definition}
\newtheorem{remark}{Remark}

\usepackage{longtable}

\usepackage{cases}

\usepackage{natbib}

\newtheorem{lemma}{Lemma}

\declaretheorem[style=definition]{example}

\newtheorem{theorem}{Theorem}
\newtheorem{proposition}{Proposition}
\newtheorem{definition}{Definition}
\newtheorem{assumption}{Assumption}

\usepackage{enumitem}

\usepackage[mathscr]{euscript}
\usepackage{mathrsfs}
\usepackage{amssymb}
\usepackage{pifont}
\usepackage{array}
\usepackage{tcolorbox}
\usepackage{multirow}

\usepackage{mathrsfs}

\newtheorem{result}{Result}

\usepackage{prodint}

\def\Rplus{{\mathbb R_+}}
\def\F{{\mathcal F}}

\newcommand{\indep}{\perp \!\!\! \perp}

\def\Nastop{{\tau^a}}

\newcommand{\mbbNa}[1]{\mathbb N^a_{#1}}
\newcommand{\mbbMa}[1]{\mathbb M^a_{#1}}
\newcommand{\mbbLa}[1]{\mathbb \Lambda^a_{#1}}

\newcommand{\na}[1]{\mathfrak{n}^a_{#1}}

\newcommand{\Np}[1]{\mathcal N^p_{#1}}
\newcommand{\Hp}[1]{\mathcal H^p_{#1}}
\newcommand{\Nd}[1]{\mathcal N^d_{#1}}
\newcommand{\Hd}[1]{\mathcal H^d_{#1}}

\newcommand{\cfproc}[1]{\tilde N^{#1}}

\def\A{{\mathcal A}}
\def\B{{\mathcal B}}
\def\N{{\mathcal N}}
\def\H{{\mathcal H}}

\def\S{{\mathcal S}}

\def\G{{\mathcal G}}

\def\I{{\mathcal I}}

\def\X{{\mathcal X}}
\def\E{{\mathcal E}}
\def\Z{{\mathcal Z}}
\def\L{{\mathcal L}}

\def\KK{{\mathbb K^a}}

\def\T{{[0,T]}}

\def\Idminaset{\I_d \setminus \{ a \}}
\def\Idmina{ \setminus a }

\renewcommand\thmcontinues[1]{Continued}

\title{On causal inference with marked point process data}
\date{\today}

\author{Pål Christie Ryalen$^{1,2,}$}
\author{Mats Julius Stensrud$^{3}$}
\author{Kjetil Røysland$^{2}$}

\begin{document}

\maketitle

\begin{center}
$^{1}$Research Support Services, Oslo University Hospital,  Norway \\
$^{2}$Department of Biostatistics, University of Oslo,  Norway \\
$^{3}$Department of Mathematics, École Polytechnique Fédérale de Lausanne, Switzerland

\vspace{0.2cm}
\texttt{Email: p.c.ryalen@medisin.uio.no}
\end{center}

\keywords{}
\begin{abstract}
   We define dynamic treatment regimes and associated potential outcomes for data described by marked point processes (MPPs). These definitions motivate MPP analogues of the commonly used consistency, exchangeability, and positivity conditions that are sufficient for identifying effects in MPP data structures. The conditions are formulated based on martingale theory, which allows us to derive explicit identifying assumptions for data described by stochastic processes. The definitions and conditions align with well-established discrete-time results in important special cases. Thus, this work bridges the large literatures on survival (event history) analysis with counting processes in continuous time and causal inference with variables in discrete-time. After formulating a set of identification conditions, we derive and characterize marginal $g$-formulas. The $g$-formulas are generally different from those studied in related works, though they coincide in important special cases. We relate our findings to previous work on causal inference with (counting) processes, the classical survival literature, and the discrete-time causal inference literature.
\end{abstract}


\section{Introduction}
\label{section: introduction}

The literature on causal inference is vast. In fields like medicine and economics, causal frameworks are now routinely applied, also in longitudinal settings with (time-varying) treatment regimes  \cite{rubin1974estimating, robins1986parametric_g, robins1987addendum,robins1997complex,Hernanrobins2021causal,richardson_single_2013}. Many causal models can be represented graphically \cite{pearl, richardson_single_2013}, providing formal tools for reasoning about conditional independencies, which are important for identification. Broadly, these models require that random variables satisfy a particular ordering, thereby defining a discrete-time stochastic process.

However, the foundational 'discrete-time' causal inference literature does not immediately apply to data structures that are often considered in the classical statistical literature on (continuous time) processes. For example, there is a large literature on survival and event history analysis formulated with continuous-time counting processes, a special case of marked point processes (MPPs), which is commonly applied in similar fields \cite{aalen2008survivalandevent,Andersen,martinussen2006dynamic,FlemingHarrington2005,CookLawless2007}. 
Many works have attempted to unify these literatures, e.g., by articulating counting process analogues of the discrete-time conditions for identifying effects of time-varying treatment regimes. Gill and Robins \cite{gill2004continuous}, henceforth referred to as GR04, conjectured that a $g$-formula (aka $g$-computation formula) could be derived under analogues of these conditions in an MPP setting. However, after GR04 stated their conjectured conditions, the authors expressed uncertainty when attempting to prove the formula, writing that
\begin{center}
    "How to proceed from here, is not so clear." \cite[p. 5]{gill2004continuous}.
\end{center}
To the best of our knowledge, the gap identified by GR04 has not since been rigorously addressed. As far as we are aware, analogues of the discrete-time identifying conditions of consistency, exchangeability, and positivity  \cite{Hernanrobins2021causal,robins1986parametric_g,robins1987addendum,robins1997complex}, have not been established for MPP data. There are no proofs of identification results in the literature that validate a set of conditions as sufficient.\footnote{Although there are conjectured identification conditions, and applications based on these conjectures, see Section \ref{section: previous work on continuous-time causal inference} for a comparison of this work with previous work on continuous-time causal inference.} 
Because MPP analogues of these conditions have not been operationalized, theories for causal identification that connect these statistical literatures are lacking. In particular, it is unclear how $g$-formulas for MPP data can be derived.

The gap between results in continuous and discrete time is unsatisfactory and inconvenient. It complicates the comparison of different approaches, and constrains the knowledge transfer across statistical disciplines. Discrete-time researchers seeking to compare their methods with classical survival analysis are forced to work across data structures that do not immediately align. This misalignment can lead to imprecision or errors in methodological comparisons. Moreover, in longitudinal studies where participants are tracked over time, the state of being under follow-up, i.e., not being censored, can be formalized as a time-varying exposure \cite{robins1986parametric_g,Hernanrobins2021causal}. Although the counting process methods were designed to tackle the problem of censoring, they are arguably ill-equipped to cover identification strategies of effects of practical interest, because conditions for identifying effects of time-varying treatment regimes are lacking.

Discrete-time data structures can be subsumed by MPP theory; adding certain restrictions to an MPP's compensator gives discrete-time data structures as a special case. Conceptually, we would therefore expect that identifying conditions in the MPP theory reduce to canonical discrete conditions when the compensator is suitably constrained. This inclusion is also strict. In particular, the continuous-time theory contains a richer collection of optional times (aka stopping times), such as \textit{totally inaccessible optional times} describing events that, informally, cannot be foreseen based on earlier recorded events. In practice, many event times may be of this type, e.g. the time of death or recurrence of a disease. Totally inaccessible times are also essential to articulate statistical estimands in the survival literature, as the existence of hazards and intensities relies on such times. Because any optional time can be uniquely decomposed into the minimum of an \textit{accessible time} and a totally inaccessible time, the discrete-time theory excludes some structure of event times of interest. This exclusion is rarely discussed, and the potential impact of the exclusion is not well understood. On the other hand, the MPP theory is well suited for the statistical analysis of event times because it accommodates a rich collection of optional times.


This paper introduces theory inspired by existing work on potential outcome variables \cite{rubin1974estimating, robins1986parametric_g, robins1987addendum, robins1997complex}, designed for MPP data structures. We give a comprehensive description of deterministic time-varying treatment regimes in the point process setting, and formally define the associated potential outcomes. Building on this foundation, we develop MPP analogues of existing established identifying conditions \cite{Hernanrobins2021causal} that are sufficient to identify associated effects. We illustrate our theory with several examples to highlight its practical consequences.

To establish identification formulas, we develop a new identifying likelihood-ratio process derived directly from the intervention rules. 
Using this likelihood-ratio process, we derive and characterize identification formulas, in particular the $g$-formula \cite{robins1986parametric_g}, whose discrete time analogues are foundational in causal inference. Our $g$-formula representations are new, as our identifying likelihood ratio process has not, to our knowledge, been emphasized in related causal inference works. To ensure backward compatibility, we demonstrate that our conditions reduce to existing discrete-time conditions \cite{robins1986parametric_g,robins1997complex,richardson_single_2013} when compensator processes induce discrete-time data structures.



When developing identification conditions and proving identification results in the MPP setting, we find that stochastic process concepts such as optional times, filtrations, and compensators (aka dual predictable projections) play central roles. As a result, some tools from stochastic processes, and specifically point processes, are required to make precise statements and to appropriately deal with measurability issues. For a comprehensive study of these concepts and results, we refer readers to, e.g., \cite{JacodShiryaev, protter, cohen2015stochastic,He1992Semimartingale}. For literature that focuses specifically on MPPs, we recommend \cite{jacobsen2006point, LastBrandt1995marked}. 

The article is organized as follows. In Section \ref{section: previous work on continuous-time causal inference}, we relate this work to existing work on continuous-time causal inference. In Section \ref{section: set-up and notation}, we introduce basic notation. In Section \ref{section: interventions and potential outcomes}, we define interventions, introduce potential outcome processes, and describe our identification criteria. In Sections \ref{section: the main identification result} and \ref{section: construction of potential outcomes}, we present our main identification result and then provide an explicit construction demonstrating that observed data, potential outcomes, and the identifying conditions (consistency and exchangeability) can be realized simultaneously on a single probability space. In Section \ref{section: g-formula for MPP data} we express classes of identifying functionals, characterize the identifying likelihood-ratio process, and propose solutions to open unsolved problems posed by GR04. In Section \ref{section: relation to discrete-time theories}, we relate our results to identifying conditions in the discrete-time literature. In Section \ref{section: identification with multiple interventions in the general MPP setting}, we extend the results to general MPP settings, before we consider estimation in Section \ref{section: estimation}. A discussion is found in Section \ref{section: discussion}. The appendices contain basic definitions, equivalent representations of MPPs central to articulate our results, proofs of main results, supporting lemmas, and calculations not found in the main text. An overview of the notation is given in the final part of the appendix. 

\section{Previous work on continuous-time causal inference}
\label{section: previous work on continuous-time causal inference}
Most of the work on causal inference concerns events in discrete time. An early contribution to continuous time causal inference is the work of Lok \cite{Lok2001statistical,Lok2004estimating,Lok2008statistical}, who, e.g., developed structural nested models \cite{robins1989analysis,Robins2004optimal} in continuous time. Work by Røysland and colleagues introduced methods based on change of measures and Girsanov's theorem \cite{roysland2011,roysland2012counterfactual}, with related contributions \cite{ryalen2018pcancer,ryalen2019additive} further developing these ideas. More recently, Røysland, Ryalen, Nygaard, and Didelez \cite{roysland2022graphical} introduced graphical criteria to aid in identification for this theory.  We will refer to these works collectively as RRND. 

However, these existing works consider structural models that are different from those we will consider. For example, structural nested models are designed with the motivation of imposing parametric restrictions on the effect of a sequence of treatments on, e.g., the conditional mean of an outcome. In contrast, we focus on the identification of marginal estimands with minimal parametric assumptions. RRND study a structural model based on changing treatment counting process intensities anchored in the more abstract invariance principle of "causal validity,"  which does not generally correspond to specifying explicit interventions. 
We, in contrast, specify explicit interventions that can be implemented by decision-makers based on patients' history.

Others, such as \cite{commenges2009dynamical, sun2022causal}, assume an underlying dynamical system, for example, a stochastic differential equation (SDE), where the path space of complex (e.g. infinite variation) processes is functionally constrained by a system that is indexed by a low-dimensional parameter. Such approaches may be too restrictive, e.g., in medical applications, where we often do not know how data are generated. Our work, in contrast, concerns causal identification of explicit interventions on specific outcomes of interest while operating under minimal assumptions regarding how the data is generated other than that it satisfies an MPP data structure.




Others have targeted estimands that are similar ours, including  Rytgaard, Gerds, and van der Laan \cite{Rytgaard2022}, hereafter referred to as RGvdL. They developed targeted minimum loss-based estimation tailored to a specific point process setting. We discuss the relationship between our work and RGvdL in Remark \ref{remark: rgvdl} and compare our examples with theirs in Section \ref{subsubsection: evaluating identification claims by RGvdL}. Another related contribution is Ying's recent framework for identifying causal effects in functional longitudinal data settings \cite{ying2024functional,ying2024functionaldynamic}. Ying aims  to impose minimal restrictions on the path space of the stochastic processes involved, and he obtains general results that can be applied in many different (functional) data settings. 

In contrast to both these related approaches, and inspired by \cite{robins1986parametric_g,richardson_single_2013,Hernanrobins2021causal}, our work is centered around single-world exchangeability conditions. In particular, we aim to specify independence conditions for identification, whose logical implications, in principle, can be tested in ideal experiments. 

Other relevant counting process applications in the literature include \cite{Hu2019causal,johnson2005semiparametric}, which explore certain dynamic regimes in similar contexts to ours, focusing  on CD4-based dynamic treatment regimes for HIV-infected adolescents and treatment duration policies where treatment can be discontinued at any point in time, respectively. 
These works specify identifying assumptions that resemble those conjectured by GR04. Similar exchangeability conditions have also been described in works on structural nested models, see e.g. 
\cite{Zhang2011,Lok2008statistical, Yang2021, Robins2000marginalvsstructural}. Our work complements existing studies that target 'estimand-based' parameters, by detailing the identification arguments, thereby clarifying how identifying conditions can effectively be applied. Some relevant details and caveats are contained in Sections \ref{section: interventions and potential outcomes}-\ref{section: g-formula for MPP data}.

In brief, our work offers the following contributions. We give definitions of dynamic treatment regimes and the associated potential outcomes, establish MPP analogues of common identifying conditions, and derive and characterize marginal $g$-formulas for MPP data. In particular, we prove and give precise meaning to GR04's conjectured $g$-formula in the MPP setting, where we also address GR04's concerns regarding the so-called validity of $g$-formulas, which has been an open problem in the literature. 
We bridge strands of the statistical literature by specifying the precise restrictions on compensators, identifying conditions, and data-generating laws in the MPP setting to align with those commonly studied in discrete-time theories. Furthermore, we develop identification results for the effects of interventions on a wide range of outcomes that include both variables and processes, rather than on a single outcome variable measured at the end of follow-up. Throughout the text, we relate our work and existing work on causal inference with (counting) processes, the classical survival and event history analysis literature, and the discrete causal inference literature. 



\begin{remark}[RGvdL's identification claims and the gap in the literature]
\label{remark: rgvdl}
    RGvdL develop results on estimation starting from postulated identification formulas. Their focus is on developing estimators; they assume that certain observed data functionals equal causal effects, based on existing identification results by Gill and Robins \cite{gill2001complex} (hereafter referred to as GR01, not to be confused with GR04).

Specifically, RGvdL assert that "the 'traditional' causal assumptions as stated by Gill and Robins (2001) can be applied at the random times" \cite[p. 2476]{Rytgaard2022}.

However, a formal argument supporting this claim is not provided by RGvdL or GR01. Furthermore, RGvdL do not address Gill and Robins' statement that extending the discrete-time identification theory to counting process data is an open problem. In GR01, Gill and Robins write: "Lok (2001) develops a counting process framework, within which she is able to formalize parts of the theory and prove many of the key results. \emph{It is an open problem to complete this project with a continuous time version of the $g$-computation formula and the theorems centered around it}" \cite[p. 1791]{gill2001complex} (emphasis added). In GR04, Gill and Robins reiterate: "It is an open problem to complete that project with a continuous time version of the $g$-computation formula and the theorems centered around it. ... Below we do not succeed in proving the formula, nor establishing the wished-for results which should follow from it" \cite[p. 1-2]{gill2004continuous}. As Gill and Robins further note, without a properly established $g$-formula, "the statistical methodology lacks motivation" \cite[p. 1]{gill2004continuous}.

In Section \ref{subsec: illustrative examples}, we examine some of RGvdL's identification formulas, based on our derived results. In particular, we identify a decision-making narrative implicit in RGvdL's examples. Based on substantive considerations, we propose an alternative decision-making narrative under which we develop our examples. Our approach aligns with causal inference works that ground assumptions in substantive stories \cite{richardson_single_2013,robins1986parametric_g,robins2011alternative, Young2024story}.
\end{remark}

\section{Set-up and notation}
\label{section: set-up and notation}

To simplify the presentation, we initially consider the special case of a $d$-dimensional multivariate counting process $N=(N^1, \dots, N^d)$ representing recordings of observed data. This aligns with data structures that are commonly studied in survival and event history analysis \cite{aalen2008survivalandevent,Andersen,martinussen2006dynamic,FlemingHarrington2005,CookLawless2007}. After establishing the main results in this setting, the generalization to general MPP data, given in Section \ref{section: identification with multiple interventions in the general MPP setting}, follows without much effort. 

A given component $N_t^i$ counts the number of times an event of type $i$, such as hospitalization or death, has occurred from time zero until $t$. In particular, there is a treatment counting process $N^a$ on which we want to intervene.

$N$ is defined on a measurable space $(\Omega, \F)$. Throughout the paper, we focus on processes on a fixed time interval $\T$, and write $\F_\T$ for the filtration generated by $N$, i.e. the natural filtration.


\subsection{Representations of point process trajectories}
\label{subsection: representation of point process trajectories}

Point processes can be represented in various ways, and we will leverage these different representations throughout the paper. For instance, a $d$-dimensional counting process $N = (N^1, \dots, \allowbreak N^d)$ can also be represented as
\begin{itemize}
    \item A double sequence $N = (T_k, X_k)_{k \geq 1}$ of event times $T_k \in \T \cup \{\infty\}$ and event types, or 'marks', $X_k$, where $T_k < T_{k+1}$ and $X_k \in \I_d := \{1,\dots,d\}$ whenever $T_k \leq T$, and $T_{k} = T_{k+1}$ and $X_k = \nabla$ whenever $T_k = \infty$, where $\nabla$ is the 'irrelevant mark' signifying 'no event'. 
    \item A random counting measure $N(dt \times dx) = \sum_{k \geq 1} \delta_{(T_k, X_k)}(dt \times dx)$ on $\T \times \I_d$, where $\delta_{(t,x)}$ denotes the Dirac measure at $(t,x)$.
\end{itemize}
These different representations are related via $N^i_t = N((0,t] \times \{i\}) = \sum_{k \geq 1} I(T_k \leq t, X_k = i)$. A more precise description of the relationship between the representations can be found in \cite{LastBrandt1995marked,jacobsen2006point} or Appendix \ref{appendix: the canonical space of point process realizations}. 

For each $\omega \in \Omega$, $N(\omega)$ is a point process trajectory. Following \cite{LastBrandt1995marked}, we reserve the notation $\varphi$ for a generic point process trajectory and move freely between its representations as needed. When treating $\varphi$ as a double sequence of ordered event times $t_k$ and associated marks $x_k$, we write $\varphi = (t_k, x_k)_{k \geq 1}$. Alternatively, when treating $\varphi$ as a $d$-dimensional counting process trajectory, we write $\varphi_t = (\varphi^1_t, \dots, \varphi^d_t)$. When treating $\varphi$ as a counting measure, we write $\varphi(dt \times dx)$, where these representations are related via $\varphi^i_t = \varphi((0,t] \times \{ i \}) = \sum_{k \geq 1} I(t_k \leq t, x_k = i)$.

We let $\N_T^{d}$ denote the set of all such trajectories $\varphi$ (represented as double sequences) on $\T$ with mark set $\I_d$. Equipped with the projection $\sigma$-algebra $\Hd{T}$, $(\Nd{T}, \Hd{T})$ becomes a measurable space, often called the canonical space of point process realizations (see Appendix \ref{appendix: the canonical space of point process realizations}, or \cite{LastBrandt1995marked,jacobsen2006point}).

Under the above identifications, the observed data process $N$ takes values in $\N_T^d$, and the treatment process $N^a$ on which we intervene takes values in $\N_T^1$. This will be referenced when defining interventions in Section \ref{subsection: interventions}.

\subsection{Compensators}

There exists a non-decreasing predictable cadlag process $\Lambda = (\Lambda^1, \dots, \Lambda^d)$ that satisfies
\begin{align}
    E_P[N^i_\tau] = E_P[\Lambda^i_\tau] \text{ for each } \F_\T\text{-optional time } \tau \text{ and each }i  ,\label{eq: compensator introduction}
\end{align}
see, e.g., \cite{jacobsen2006point,LastBrandt1995marked,Bremaud1981point}. This is called the compensator of $N$, and it is uniquely specified (up to indistinguishability) by \eqref{eq: compensator introduction}. The role of compensators in our MPP theory is similar to the role of conditional distributions in discrete causal inference theories \cite{pearl,robins1986parametric_g,Dawid2021decision}; they parametrize statistical models \cite{jacobsen2006point,LastBrandt1995marked, jacod1975}. The compensator depends on both the probability measure and the filtration. Because we will consider different probability measures and filtrations, we will emphasize this dependence; for example, $\Lambda$ is the $(P, \F_\T)$-compensator of $N$.

Similarly to the point process trajectories, we will exploit equivalences of process and random measure representations of the compensator $\Lambda$ \cite{LastBrandt1995marked,jacobsen2006point}, and move freely between these representations as needed.

For technical reasons we assume throughout that
\begin{assumption}\label{assumption: no explosions}
    $N$ is without explosions, i.e. $P\big(\sum_{i=1}^d N_T^i = \infty \big) = 0$.
\end{assumption}
Assumption \ref{assumption: no explosions} says that an infinite number of events (or jumps) occur with probability zero in the interval $\T$. In particular, Assumption \ref{assumption: no explosions} ensures that $M := N - \Lambda$ is a local martingale with respect to $P$ and $\F_\T$ \cite{Andersen}.

\section{Interventions, Potential outcome processes, and identifying conditions}
\label{section: interventions and potential outcomes}

\subsection{Interventions} 
\label{subsection: interventions}
To derive the consequences of interventions we must be clear about what these interventions --- also called actions, treatment strategies, plans, policies, protocols, or regimes \cite{Hernanrobins2021causal,pearl,Dawid2021decision} --- are. In particular, our results are sensitive to the kind of information the interventions depend upon. 
To motivate our intervention definition, we begin with an example to illustrate what this dependence looks like in practice.

Consider a patient visiting a medical clinic for a consultation with a physician. Before deciding on treatment, the physician first reviews test results and consults guidelines. Thus, the clinical information recorded at assessment is available before the time treatment is subsequently administered. 
Finally, just before acting, doctors have, in principle, access to the timing of their action (they could simply check what time it is).

In this simple illustration, the decision-relevant information at each time $t$ (clinical measurements and action timing) is available to the doctor strictly prior to $t$. 
Assuming that all data for a given patient is described by an MPP trajectory $\varphi$, we write $\varphi|_{t-}$ for the restriction of $\varphi$ to the interval $[0, t)$, i.e., the history up to right before $t$. In this example, any hypothetical intervention specifying assigned treatment for the given patient can, at each time $t$, depend only on $\varphi|_{t-}$.

The conclusion from our motivating example invites a more general question: can any real-world decision, whether made by a clinician, an automated system, or any other agent, depend on truly temporally simultaneous events? We believe the answer to this question is 'no'. Any plausible decision-making scenario involves a delay between information receipt and action; during that interval, the decision-maker has access to the information that will guide the action, including (in principle) the timing of the action itself. Because this property appears fundamental to how decisions are actually made, we encode it directly into our definition of an intervention.

To formally state our definition, we must introduce some notation. Following \cite{LastBrandt1995marked}, we say that a function $f$ on $\T \times \Nd{T}$ is \emph{predictable} if $f(t, \varphi) = f(t, \varphi|_{t-})$ for each $t$ and point process trajectory $\varphi \in \Nd{T}$. This definition of predictability coincides with the conventional meaning of predictability when understood with respect to the filtration generated by the identity on $(\Nd{T}, \Hd{T})$ \cite[Theorem 2.2.6, 2.2.8]{LastBrandt1995marked}.
\begin{definition}[Intervention]\label{definition: intervention}
    We say that a map $\na{}: \Nd{T} \rightarrow \N_T^1$ is an \emph{intervention}, or a \emph{predictable intervention}, if the mapping $(t, \varphi) \mapsto \na{t}(\varphi) := \na{}(\varphi)_t$ is predictable, i.e., if $\na{t}(\varphi) = \na{t}(\varphi|_{t-})$.
\end{definition}
That is, an intervention $\mathfrak n^a$ is a predictable one-dimensional counting process on the canonical space: for each $\varphi \in \N_T^d$, $\na{t}(\varphi)$ is the value at time $t$ of the one-dimensional counting process trajectory $\na{}(\varphi) \in \N_T^1$ (recall the notational conventions in Section \ref{subsection: representation of point process trajectories}). Since the observed treatment process $N^a$ takes values in $\N_T^1$, Definition \ref{definition: intervention} ensures that $\na{}$ produces feasible treatment trajectories. Thus, $\na{}$ formalizes a hypothetical decision-making process, where at any time $t$, the assigned treatment may depend adaptively on the history of events strictly prior to $t$.

The perspective on interventions encoded in Definition \ref{definition: intervention} is not new: foundational causal inference works define dynamic treatment regimes as functions of variables strictly earlier in the observed data ordering \cite{Hernanrobins2021causal, richardson_single_2013, robins1986parametric_g, pearl, Dawid2010identifying}, an ordering that is typically understood, implicitly and explicitly, to be temporal. These definitions align with our observation that real-world interventions --- actions that can actually be carried out --- depend only on information available strictly before the action. Consequently, our definition captures the interventions that can be specified as part of an ideal experiment or protocol.

Definition \ref{definition: intervention} focuses on deterministic interventions on a one-dimensional counting process component. The extension to interventions on multiple components in the general MPP setting is given in Section \ref{section: identification with multiple interventions in the general MPP setting}.

Deterministic interventions are often of main interest in practice; for example, medical doctors rarely give treatments to patients by flipping coins (or by drawing random numbers). The generalization to random treatment strategies is nevertheless not difficult in principle if we include in our filtration a collection of 'randomizers' that can figure as input in $\na{}$; see, e.g., \cite[equation (57)]{richardson_single_2013} for an analogous definition in discrete time. 

In our examples (Sections \ref{subsection: examples of interventions and potential optcomes processes} and \ref{section: identification with multiple interventions in the general MPP setting}), we formalize stories about how decisions are made in practice. The interventions described in these examples are captured by Definition \ref{definition: intervention}.

\subsection{Potential outcome processes}
\label{subsection: potential outcome processes}
We introduce processes representing the outcomes that would have been realized if an intervention as in Section \ref{subsection: interventions} had been carried out with perfect compliance. This takes the form of another $d$-dimensional counting process $\cfproc{}$, which plays a role similar to counterfactual variables, aka potential outcomes, in existing theories \cite{rubin1974estimating,robins1986parametric_g,robins1997complex}. We will also consider the filtration $\tilde \F_\T$ generated by $\cfproc{}$, which gives a formalization of the history that emerges when the intervention is imposed. Though not explicit in the notation, $\cfproc{}$ is defined in terms of an intervention as given in Section \ref{subsection: interventions}.

In contrast to $N$, the process $\cfproc{}$ is not fully 'observed.' A central part of this work is to present conditions that ensure that marginal expectations of potential outcomes of interest are identified from the information in $\F_T$. The outcomes we consider, with associated identifying conditions, are detailed in Sections \ref{subsection: outcomes of interest}–\ref{subsection: identifying conditions}.

We characterize $\cfproc{}$ via its compensator. Because the distribution of an MPP is determined by the compensator, \cite[Section 8]{LastBrandt1995marked}, the upcoming definition also determines $\cfproc{}$'s law, $P(\tilde N \in \cdot)$. Consequently, $P$-expectations of measurable functionals of $\tilde N$ are unambiguously defined. 

Our definition makes use of the unique "canonical compensator," $\alpha = (\alpha^1, \dots, \allowbreak \alpha^d)$, of $N$ with respect to $P$ and the natural filtration $\F_\T$. $\alpha$ has the property that when evaluated in the observed trajectories, it defines a compensator of the observed process, i.e.
\begin{align}
    \Lambda^j = \alpha^j(N)
    \quad P\text{-a.s.} \label{eq: alpha j compensator} 
\end{align}
for each $j$. See Appendix \ref{appendix: canonical compensators and optional times} or \cite[Theorem 4.2.2]{LastBrandt1995marked} for a precise result on canonical compensators. 

To ensure a well-defined canonical compensator (see \cite[Theorem 4.2.2 and Section 4.3]{LastBrandt1995marked}) suitable for constructing potential outcomes, we impose the following regularity conditions:
\begin{align}
    \big\{ \sum_{j \in \I_d \setminus \{ a \}} d\alpha^j_t(\varphi) +  d\mathfrak{n}^a_t(\varphi)  \big\} I(\pi_\infty'(\varphi) < t) &= 0, \label{eq: compensator explosion regularity} \\
    \sum_{j \in \I_d \setminus \{a\}} \Delta \alpha_t^j(\varphi) + \Delta \mathfrak{n}_t^a(\varphi) &\leq 1, \label{eq: compensator jump regularity}
\end{align}
for each $\varphi \in \N_T^d$ and $t \in \T$, where $\pi_\infty'(\varphi) := \inf\{ s > 0 | \sum_{j \in \I_d \setminus \{ a \}} \alpha^j_s(\varphi) +  \mathfrak{n}^a_s(\varphi) = \infty \}$ and $\Delta Z_t := Z_t - Z_{t-}$ for a cadlag process $Z$. The condition 
\eqref{eq: compensator explosion regularity} can be made to hold by redefining $\mathfrak n^a$ and the $\alpha^j$'s if necessary. \eqref{eq: compensator jump regularity} is a substantive condition which depends on the intervention under study. The condition is specific to the MPP setting, and it is connected to the assumption of strict monotonicity of MPP jump times. The condition may be violated if $\na{}$ imposes treatment times which have a positive probability of coinciding with event times of the other processes.

The conditions \eqref{eq: compensator explosion regularity}-\eqref{eq: compensator jump regularity} hold for many of the interventions considered in this article, such as the treatment-prevention intervention in Example \ref{exa:determ} later in this section. To simplify the presentation, we will assume that \eqref{eq: compensator explosion regularity}-\eqref{eq: compensator jump regularity} hold throughout, recognizing that these conditions must be verified on a case-by-case basis in practice.

\begin{definition}[Potential outcome process]\label{definition: potential outcome process}
    Fix a version of the canonical compensator $\alpha = (\alpha^1, \dots, \alpha^d)$ in \eqref{eq: alpha j compensator} and an intervention $\mathfrak n^a$ satisfying \eqref{eq: compensator explosion regularity}-\eqref{eq: compensator jump regularity}. We say that $\cfproc{}$ is a \emph{potential outcome process} with respect to the intervention $\na{}$ if
\begin{enumerate}[label=\textnormal{(\alph*)}]
        \item \label{item: invariance a} $\mathfrak{n}^a(\cfproc{})$ defines a $(P,\tilde \F_\T)$-compensator of $\cfproc{a}$, 
        \item \label{item: invariance} $\alpha^j ( \cfproc{} )$ defines a $(P,\tilde \F_\T)$-compensator of $\cfproc{j}$ for $j \neq a$.
\end{enumerate}
\end{definition}
Definition \ref{definition: potential outcome process} \ref{item: invariance a} is equivalent to
\begin{align}
    \cfproc{a} &= \na{}(\cfproc{}) \quad P\text{-a.s.} \label{eq: predictable intervention tilde a compensator}
\end{align}
see Appendix \ref{subsection: potential outcome process for predictable regime} for details. Thus, Definition \ref{definition: potential outcome process} states that $\cfproc{a}$ is fixed to $\na{}(\cfproc{})$, while the canonical compensators of all processes other than the treatment process are the same as for the observed processes. In particular, $\alpha^j ( \cfproc{} )$ then coincides with
\begin{align}
    \alpha^j \big( (\cfproc{1}, \dots, \cfproc{a-1},\na{}(\cfproc{}),\cfproc{a+1}, \dots, \cfproc{d}) \big). \label{eq: predictable intervention tilde j compensator}
\end{align}
Definition \ref{definition: potential outcome process} and \eqref{eq: predictable intervention tilde a compensator}-\eqref{eq: predictable intervention tilde j compensator} are analogous to definitions of potential outcomes starting from structural equations \cite{pearl,richardson_single_2013}. In particular, Definition \ref{definition: potential outcome process} becomes equivalent to common definitions of data-generating laws found in the discrete-time literature \cite{robins1986parametric_g,richardson_single_2013}, if we impose the restrictions on compensators that induce discrete-time data structures, see Proposition \ref{proposition: relation to existing assumptions} in Section \ref{section: relation to discrete-time theories}. 

To limit the scope of this work we have not included the natural value process in Definition \ref{definition: potential outcome process}, that is, the treatments a subject would take under the regime if we had not (adaptively) implemented the regime \cite{robins2004effects,Young2014natural,richardson_single_2013,sarvet2025natural}. We do therefore not study interventions that depend on the natural value process, nor do we identify effects on the natural value process. Such extensions will be considered in future research.

\subsection{Examples of interventions and associated potential outcome processes}
\label{subsection: examples of interventions and potential optcomes processes}

The following examples illustrate interventions and their associated potential outcome processes. We give further examples in Section \ref{section: identification with multiple interventions in the general MPP setting}.

\begin{example}[label=exa:determ][Atomic intervention]  \label{example: atomic}
Consider an intervention that fixes each trajectory to some $\hat \varphi \in \N^1_T$, i.e. where the map $\varphi \mapsto \na{}(\varphi) = \hat \varphi$ is constant. An example of such a $\hat \varphi$ is the trajectory that is identically zero on $\T$, corresponding to the intervention which prevents treatment. Because constant processes are predictable, we know from \eqref{eq: predictable intervention tilde a compensator} that
 $$ \cfproc{a} =  \na{}(\cfproc{}) = \hat \varphi \quad P \text{-a.s.} $$
 Furthermore, \eqref{eq: predictable intervention tilde j compensator} states that the compensator of $\cfproc{j}$ coincides with 
 $$\alpha^j\big( (\tilde N^1, \dots, \tilde N^{a-1}, \hat \varphi,\tilde N^{a+1}, \dots, \tilde N^{d} ) \big),$$ 
 i.e. with the $a$'th component substituted with the intervened trajectory $\hat \varphi$. Interventions that fix treatment trajectories are often called "deterministic static regimes" \cite{Hernanrobins2021causal}, or "actions" \cite{pearl, Dawid2021decision}. 
\end{example}

\begin{example}[label=exa:determ_dynaimc][Starting treatment after a covariate spike] \label{example: deterministic dynamic}
    Consider a setting where individuals are followed starting from diagnosis (at $t=0$), with the possibility of joining a treatment program at a later point in time. Such treatment decisions are often influenced by changes in certain covariates. Suppose there is a covariate that can 'spike' during the study period, and let $\varphi^l$ in $\varphi = (\varphi^1, \dots, \varphi^d)$ denote the component that counts the occurrences of these covariate spikes.
    
    The question of when to start a treatment program is often of clinical interest. Consider, for example, an investigator who is interested in initiating the program a time increment $\delta > 0$ after a covariate spike is reported. This action corresponds to
    \begin{align*}
        \na{t}(\varphi) = \varphi_{t - \delta}^l,
    \end{align*}
    using the convention $Z_t := 0$ for $t < 0$ when $Z$ is a function or process on $\T$. The delay $\delta$ gives clinicians time to review information about the spike and decide whether to initiate treatment. Here we assume $\delta$ is deterministic for simplicity, although our results allow for random delays. \footnote{We briefly discuss random delays in Section \ref{subsec: illustrative examples}.} 
    
    Because $\delta > 0$, the intervention is predictable, i.e. $ \na{t}(\varphi) = \na{t}(\varphi|_{t-})$, and from \eqref{eq: predictable intervention tilde a compensator} it follows that
    \begin{align*}
        \cfproc{a}_t = \na{t}(\cfproc{} ) = \cfproc{l}_{t - \delta} \quad P \text{-a.s.},
    \end{align*}
    which can further be substituted into the compensator of $\cfproc{j}$ as in \eqref{eq: predictable intervention tilde j compensator}. This $\na{}$ is an example of a "deterministic dynamic regime" \cite{Hernanrobins2021causal}, or a "conditional action" \cite{pearl}.
\end{example}

\begin{example}[label=exa:determ_dynaimc_review_period][Treatment assignment after review period] \label{example: deterministic dynamic predictable refinement}
    Consider a doctor who makes a binary treatment decision (treat or not treat) when a patient visits. More specifically, after a patient arrives at a clinic, the doctor reviews available information (e.g., medical history), and \textit{then} makes a treatment decision. This review takes time, inevitably creating a delay between information gathering and treatment allocation.
    
    In the observed data, there is a short delay $\delta > 0$ between the time information is gathered and the treatment is allocated. As in Example \ref{example: deterministic dynamic}, we take $\delta$ to be deterministic for simplicity; the specific nature of the delay is not important for this story. The delay reflects both the processing time needed for clinical decisions and that the doctor is aware of the decision timing before acting, consistent with our discussion in Section \ref{subsection: interventions}. 

    Suppose that the doctor's decision may be informed by whether the patient experienced a covariate spike in the previous week. Let $\varphi^l$ count covariate spikes, $\varphi^v$ count when information is gathered (assessment), and $\varphi^a$ count recorded treatments. The intervention "administer treatment if a covariate spike occurred in the previous week" assigns treatment at $t$ if a spike occurred in $[t-1, t)$ (assuming time is measured in weeks). Formally, this intervention is
    \begin{align}
        \mathfrak n_t^a(\varphi) = \int_0^t I(\varphi_{s-}^l - \varphi_{s-1}^l > 0) d\varphi_s^{v,\delta}, \label{eq: predictable dynamic with review}
    \end{align}
    where the counting process trajectory $\varphi^{v,\delta}_t := \varphi^v_{t - \delta}$ shifts the assessment times forward by $\delta$, so that treatment decisions occur $\delta$ time units after the corresponding assessment. Because $\delta > 0$, the decision at time $t$ depends only on information available strictly before $t$, making $\mathfrak n^a$ predictable. By \eqref{eq: predictable intervention tilde a compensator} we get that
    \begin{align*}
        \tilde N^a_t = \int_0^t I(\tilde N^l_{s-} - \tilde N^l_{s-1} > 0) d\tilde N_s^{v,\delta} \quad P\text{-a.s.},
    \end{align*}
    where $\tilde N_s^{v,\delta} := \tilde N^v_{s-\delta}$, and $\tilde N^a$ can then be substituted into the compensator of $\cfproc{j}$ as in \eqref{eq: predictable intervention tilde j compensator}.
\end{example}

\subsection{Outcome functionals of interest, observed and potential outcomes of interest}
\label{subsection: outcomes of interest}
We will identify estimands that are marginal expectations of specific functionals of $\cfproc{}$. We study functionals $\dot Y = \{ \dot Y_t \}_{t \in \T}$ which are optional stochastic processes on  $(\N_T^{d},\H_T^{d})$; that is, $\dot Y$ satisfies
\begin{align*}
    \dot Y_t(\varphi) &= \dot Y_t(\varphi|_t).
\end{align*}
From $\dot Y$ we define the potential outcome of interest as  
\begin{align*}
    \tilde{Y} &:= \dot{Y}( \tilde N ),
\end{align*}
while the associated observed outcome is 
\begin{align*}
    Y &:= \dot{Y}( N ).
\end{align*}
In plain English, these definitions state that at any time $t$, the observed (or potential) outcome of interest is a functional of events occurring up to time $t$ in the observed process $N$ (or the potential outcome process $\tilde{N}$, respectively).

An outcome functional $\dot Y$ as above can e.g. be a process with piecewise constant cadlag or caglad paths. A simple example is survival: $\dot Y_t(\varphi) = I(\varphi^d_t = 0)$ defines the survival indicator at each time $t$, where the $d$-component counts the occurrence of death. The corresponding observed and potential outcomes are $Y_t = I(N_t^d = 0)$ and $\tilde Y_t = I(\tilde N_t^d = 0)$.  Outcomes of interest that are random variables at specific time points, such as survival at the end of follow-up $T$, can easily be desribed by such a functional. 

\subsection{Identifying conditions}
\label{subsection: identifying conditions}
Having defined the outcomes of interest, we turn to the conditions under which expectations of such outcomes can be identified. The upcoming definition focuses on an irrelevance condition which ensures that a wide range of outcomes are identified.  

To formulate our identification conditions, consider the optional time
\begin{align}
    \Nastop := \inf\{ s > 0 | N^a_s \neq \na{s}(N) \}, \label{eq: tau a}
\end{align}
which plays a central role in this work; this is the first time the observed treatment process differs from the treatment assigned by the regime based on the observed trajectories. We define the process $\mbbNa{} := I(\Nastop \leq \cdot)$, which counts whether $\Nastop$ has occurred up to a given point in time. Consider also $\mbbLa{}$, the $(P,\F_\T)$-compensator of $\mbbNa{}$, and define $\mbbMa{} := \mbbNa{} - \mbbLa{}$, which is a martingale with respect to $P$ and $\F_\T$.

Throughout we adopt the assumption
\begin{align}
    P\Big( \int_0^T \frac{d\mbbLa{s}}{1 - \Delta \mbbLa{s}} < \infty \Big) = 1. \quad\quad\quad \text{(Positivity)} \label{eq: K jump to zero condition}
\end{align}
We refer to \eqref{eq: K jump to zero condition} as \textit{positivity}, as it is equivalent to positivity conditions described in the discrete-time literature \cite{Hernanrobins2021causal} when compensators are restricted. Further details on this connection are given in Section \ref{section: relation to discrete-time theories}. In our work, \eqref{eq: K jump to zero condition} ensures that the upcoming key equations are well-defined. However, we will suppose additional integrability conditions beyond \eqref{eq: K jump to zero condition} to establish identification formulas in our setting. These additional integrability conditions hold automatically in discrete-time settings provided that \eqref{eq: K jump to zero condition} holds, see Proposition \ref{proposition: relation to existing assumptions} in Section \ref{section: relation to discrete-time theories}.

Under \eqref{eq: K jump to zero condition}, the process
\begin{align}
    \KK := -\int_0^\cdot \frac{d\mbbMa{s}}{1 - \Delta \mbbLa{s}} \label{eq: K semimart}
\end{align}
is well-defined, and also a local martingale with respect to $P$ and $\F_\T$. Consequently, the SDE 
\begin{align}
    W &= 1 + \int_0^\cdot W_{s-} d\mathbb K^a_s, 
 \label{eq: W dolean}
\end{align}
i.e. the stochastic exponential (aka Doléans-Dade exponential) of $\KK$, is also well-defined, and admits a unique solution $W = \E(\KK)$ \cite[II Theorem 37]{protter}.

We require that $W$ is a likelihood ratio process, i.e. that $W$ is a nonnegative, mean one uniformly integrable martingale on $\T$. The stochastic exponential provides the canonical method to generate likelihood ratio candidates and verifying integrability conditions; see, e.g., \cite{cohen2015stochastic,protter,JacodShiryaev,Sokol2015,shiryaevkallsen2002cumulant} and the references therein. In particular, it follows from the preceding construction that $W$ is a nonnegative local martingale (see Lemma \ref{lemma: W finite variation} in Appendix \ref{appendix: supporting calculations} for details). Thus, $W$ is a valid likelihood ratio process provided that the integrability conditions in the following definition are met.

\begin{definition}[Identifying conditions]\label{def:strong prime}
For an intervention $\na{}$ and outcome functional $\dot Y$ (Section \ref{subsection: outcomes of interest}), the following conditions are said to be the \emph{identifying conditions} for the potential outcome process $\tilde Y$ of interest:
    \begin{enumerate}[label=\textnormal{(\roman*)}]
        \item \label{item: strong cf consistency} Consistency:
        \begin{align}
            Y_t I(\tau^a > t) = \tilde Y_t I(\tau^a > t) \quad P  \text{-a.s. for each } t \in \T.
            \label{eq: consistency}
        \end{align}
        \item \label{item: strong cf exch} Exchangeability:
        \footnote{A filtration $ \big\{ \F_t \vee \sigma(\tilde Y) \big\}_{t \in \T}$, which treats $\tilde Y$ as a random element that is realized at baseline is sometimes called the \textit{initial enlargement} of $\F_\T$ by the random element $\tilde Y$. We could equivalently have phrased \eqref{eq: exchan disc} in terms of the filtration $\big\{ \mathcal F_{\tau^a \wedge t} \vee \sigma(\tilde Y) \big\}_{t \in \T}$, where  $\mathcal F_{\tau^a \wedge t}$ is the $\sigma$-algebra associated with the optional time $\Nastop \wedge t$. } 
        \begin{align}
        \begin{split}
            \mbbLa{} \text{ is the compensator of } \mbbNa{} \text{ with respect to both } \F_\T \text{ and } \F_\T^{\tilde Y} \text{ under } P,
        \end{split}
            \label{eq: exchan disc}
        \end{align}
        where $\F_\T^{\tilde Y} := \{ \F_t \vee \sigma(\tilde Y) \}_{t \in \T}$.
        \item \label{item: likelihood ratio regularity} Likelihood ratio regularity:
            \begin{align}
            E_P[W_t]=1 & \text{ for each } t \in \T. \label{eq: W likelihood ratio regularity} 
            \end{align}
    \end{enumerate}
\end{definition}

The consistency condition \eqref{eq: consistency} states that almost all trajectories of $Y$ and $\tilde Y$ coincide as long as the trajectories of $\na{}(N)$ are equal to the trajectories of the observed treatment process $N^a$, and is an extension of the discrete time consistency condition for time-varying treatment regimes. The condition connects observed outcomes with potential outcomes through the treatment actually received; the realizations coincide so long as an individual's observed treatment aligns with the treatment regime of interest.

The exchangeability condition \eqref{eq: exchan disc} states informally that, for an individual who, in the observed data, has followed the regime up to any point in time, knowledge of their future potential outcome $\tilde Y$ does not offer any additional information to improve predictions about whether they will continue following the regime in the future. The condition is often violated when unmeasured covariates, not recorded in $\F_\T$, are predictive of subsequent treatments and outcomes of interest; conceptually, this can be interpreted as unmeasured confounding. In the case of a sequentially randomized experiment with perfect compliance, we have that  \eqref{eq: exchan disc} holds by design for $\tilde Y = \tilde N$. We discuss the testability of our type of exchangeability assumptions in Appendix \ref{subsection: testability of exchangeability}.

The exchangeability condition \eqref{eq: exchan disc} shares similarities with conditions previously described in the causal inference literature \cite{Zhang2011,Lok2008statistical,Yang2021,Robins2000marginalvsstructural,gill2004continuous}. Yet, to our knowledge, and as explicitly highlighted in \cite[p. 1791]{gill2001complex} and \cite[p. 1]{gill2004continuous}, such conditions have not before been successfully adapted to derive a $g$-formula for MPP data.  
To establish the formula we develop the new likelihood ratio construction in \eqref{eq: K semimart}-\eqref{eq: W dolean}, along with integrability conditions for valid measure changes which ensure that our $g$-formulas are well-behaved.

Discrete-time theories express exchangeability as conditional independencies, \allowbreak which are symmetric. Symmetry is not generally defined for our exchangeability condition because $\tilde Y$ need not have compensators with respect to, nor even be adapted to, the filtrations $\F_\T$ and $\{\F_t \vee \sigma(\mbbNa{}) \}_{t \in \T}$. Still, even if $\tilde Y$ has compensators with respect to both these filtrations, it does not follow from \eqref{eq: exchan disc} that these compensators coincide.

The condition \eqref{eq: exchan disc} is also closely related to conditions found in the literature on missing data with stochastic processes \cite{robins1992recovery,Sattenrobins2001estimating,Robins2000correcting}. A more detailed study of the connection between our conditions and those in the missing data literature will be addressed in subsequent work.

Since \eqref{eq: K jump to zero condition} is assumed to hold, the likelihood ratio regularity condition \eqref{eq: W likelihood ratio regularity} is equivalent to stating that $W$ defines a likelihood ratio process on $\T$ with respect to $P$ and the filtration $\F_\T$. This is shown in Lemma \ref{lemma: positivity and likelihood ratio equivalence} in Appendix \ref{appendix: general MPP potential outcomes}. \footnote{Given \eqref{eq: K jump to zero condition}, \eqref{eq: W likelihood ratio regularity} is equivalent to that $W$ is a martingale on $\T$. $W$ is then automatically uniformly integrable, as shown in the proof of Lemma \ref{lemma: positivity and likelihood ratio equivalence}. For non-compact time domains the matter is more delicate; see e.g. \cite[I 2]{protter}.} The same lemma shows that, under the exchangeability condition \eqref{eq: exchan disc}, $W$ is also a likelihood ratio process with respect to the filtration $\F_\T^{\tilde Y}$, and in particular a martingale with respect to this filtration. This innocent-looking implication is the driving force behind our main identification result, Theorem \ref{theorem: ipw} in Section \ref{section: the main identification result}.

The process $W$ has further representations with alternative interpretations. For instance, equations \eqref{eq: W unstabilized weight}-\eqref{eq: W discrete and cont expression} in Section \ref{subsection: characterizing identifying likelihood ratio} represent $W$ as a regime adherence indicator divided by a cumulative product of conditional (infinitesimal) regime adherence probabilities. These identities establish $W$ as a continuous-time analogue of the "unstabilized weights" involved in discrete inverse probability weighting (IPW) identification formulas for time-varying treatment regimes \cite{hernan2000marginal, Hernanrobins2021causal}. The representations also highlight how IPW processes can be thought of as "censoring weights" by choosing a suitable reparametrization. We discuss these aspects further in Section \ref{subsection: characterizing identifying likelihood ratio}.

Section \ref{section: relation to discrete-time theories} establishes a detailed connection between our conditions and conditions which are widely used in the discrete-time causal inference literature. We refer readers to \cite{Hernanrobins2021causal} for a comprehensive and practically oriented discussion of these conditions.

\subsection{Differences from identification conditions in related work}
\label{subsection: differences between positivity conditions}
In contrast to likelihood ratios that have been posed in related work on causal inference with stochastic processes \cite{ryalen2019additive,roysland2011,roysland2012counterfactual,roysland2022graphical,Rytgaard2022, ying2024functional,ying2024functionaldynamic}, we do not \textit{define} $W$ as a likelihood ratio between a specified 'target' distribution and the observational distribution, where the target distribution is obtained by replacing one component of the observational distribution with an 'interventional' component. In particular, our definition of $W$ does not reference a parametrization of the model, and interpreting $W$ in terms of how the model is parametrized does not immediately seem helpful. Nonetheless, our results are derived from formal definitions of interventions and align with existing discrete-time work \cite{robins1986parametric_g,robins1997complex,Hernanrobins2021causal,richardson_single_2013}, where $g$-formulas are derived under single world exchangeability conditions.

The commonly used informal description of positivity, that there must be some subjects who follow the hypothetical intervention within all levels of the past, \cite{Hernanrobins2021causal}, aligns more naturally with the evaluability conditions \eqref{eq: K jump to zero condition} and \eqref{eq: W likelihood ratio regularity}. This connection is captured by the representation of $W$ found in \eqref{eq: W unstabilized weight} in Proposition \ref{prop: characterizations of W} in Section \ref{subsection: characterizing identifying likelihood ratio}. In that same section, we further characterize $W$ and provide a specific unique relationship between $W$ and our exchangeability and consistency conditions. 

To distinguish between our work and related works it is helpful to introduce a new term for the distribution induced under the conditions in Definition \ref{def:strong prime}. This distribution is integrally linked to our definitions of interventions and our identifying conditions.
\begin{definition}[$g$-formula distribution] \label{defn: g formula distribution}
    We call the probability measure
    \begin{align*}
        dQ := W_T dP 
    \end{align*}
    on $\F_T$ the \emph{$g$-formula distribution} associated with a given intervention $\na{}$ under the conditions in Definition \ref{def:strong prime}.
\end{definition}

\section{Main Identification Result}
\label{section: the main identification result}
Under the conditions in Definition \ref{def:strong prime}, marginal expectations of a bounded \footnote{Boundedness can be substituted with other context-specific integrability assumptions, and further relaxed by imposing more stringent integrability conditions on $W$ or $\mbbMa{}$; see e.g. \cite[Theorem 2.1.42]{LastBrandt1995marked}.} potential outcome process of interest $\tilde Y_t$, as described in Section \ref{subsection: outcomes of interest}, associated with a realization $\cfproc{}$ that emerges under an intervention as in Definition \ref{definition: intervention}, are identified. This is established in the following theorem. 
\begin{theorem}\label{theorem: ipw}
    Suppose that the conditions of Definition \ref{def:strong prime} are satisfied for a bounded outcome functional of interest (Section \ref{subsection: outcomes of interest}). Then, for each $t \in \T$,
    \begin{align}
        E_P[\tilde  Y_t ] = E_P[ W_t Y_t ]. \label{eq: exchangeability ipw formula}
    \end{align}
\end{theorem}
\begin{proof}
    Consistency implies via Lemma \ref{lemma: W finite variation} \ref{enum: W jump to zero identity} (Appendix \ref{appendix: supporting calculations}) that 
    \begin{align*}
        E_P[W_t Y_t] = E_P[W_t \tilde Y_t].
    \end{align*}
    Moreover, Lemma \ref{lemma: positivity and likelihood ratio equivalence} (Appendix \ref{appendix: supporting calculations}) gives that $W$ is a uniformly integrable $(P,\F_\T^{\tilde Y})$-martingale under the conditions \eqref{eq: K jump to zero condition}, \eqref{eq: exchan disc}, and  \eqref{eq: W likelihood ratio regularity} --- i.e. under our positivity, exchangeability, and likelihood ratio regularity conditions. Since $\tilde Y_t$ is bounded and evidently $\sigma(\tilde Y)$-measurable, we get $E_P[W_t \tilde Y_t] = E_P[ W_0 \tilde Y_t]$ by the martingale property. The desired result \eqref{eq: exchangeability ipw formula} follows since $W_0=1$.
\end{proof}
The formula \eqref{eq: exchangeability ipw formula} is an example of an IPW formula \cite{ROSENBAUM1983propensity}, and it is an identification formula since $W$ is $\F_\T$-adapted, making the right-hand side a functional of observed quantities. Our work focuses on marginal expectations of bounded potential outcome processes $\tilde Y$, but the techniques in Theorem \ref{theorem: ipw} can readily be used to derive identification results for many other quantities of interest, such as conditional effects given baseline covariates or distributional effects.

While we focus on point processes in this paper, the conditions in Definition \ref{def:strong prime} and the associated proof strategy refer to concepts from the general theory of stochastic processes such as filtrations, optional times, and compensators. The ideas presented here apply more generally to a wide range of stochastic processes, including discrete time processes, and our presentation thus unifies identification results across vastly different data structures. Extending our work to other stochastic process settings would require definitions of potential outcomes for such settings, which is beyond the scope of this paper and will not be discussed further here.

Other works have conjectured or claimed that a sequential exchangeability condition --- a sequence of conditional independence conditions, one condition connected to all or some of the jump times $\{ T_k \}_{k \geq 1}$ of $N$ --- is sufficient to derive formulas analogous to \eqref{eq: exchangeability ipw formula}  \cite{gill2004continuous,Rytgaard2022}. These conjectured conditions, however, are not equivalent. The question of whether such conditions are valid is not just a theoretical puzzle, because investigators interested in testing exchangeability need to know if they are testing a sufficient condition. In Proposition \ref{proposition: sequential exchangeability predictable times} in Appendix \ref{appendix: seq exchangeability proof sketch} we sketch an argument using a sequential condition that ensures \eqref{eq: exchangeability ipw formula} under the additional assumption that each $T_k$ is predictable. The result may have limited practical value, since the assumption that the jump times of $N$ are predictable with respect to the observed filtration is very hard to justify in practical settings. We include the proposition and the connected lemmas because they may help other researchers derive $g$‑formulas, or clarify relationships between different identification conditions in MPP settings. 



\section{A joint construction ensuring identifying conditions are satisfied}
\label{section: construction of potential outcomes}
Theorem \ref{theorem: ipw} shows that if the consistency, exchangeability, and likelihood ratio regularity conditions in Definition \ref{def:strong prime} hold, then marginal means of potential outcomes of interest are identified. A concern when imposing identifying assumptions is whether they place hidden restrictions on the observed data law. This problem was examined in discrete time by GR01. While exchangeability conditions similar to ours have been described in the literature, the issue of whether such conditions impose hidden restrictions has not, to our knowledge, been addressed.

The following proposition shows that, under a certain martingale orthogonality condition, we can construct, on a new probability space if necessary, observed and potential outcomes with the correct laws such that our consistency and exchangeability conditions hold simultaneously. Specifically, if the intervention respects \eqref{eq: compensator explosion regularity}-\eqref{eq: compensator jump regularity}, and the canonical compensator $\alpha = (\alpha^1, \dots, \alpha^d)$ in \eqref{eq: alpha j compensator} satisfies \eqref{eq: mcp can comp orthogonality}, then our identifying conditions impose no hidden restrictions on the observed data law. The result is generalized to multiple interventions in Theorem \ref{theorem: multiple marked interventions} in Section \ref{section: identification with multiple interventions in the general MPP setting}.
\begin{proposition}[Construction of observed and potential outcomes satisfying Definition \ref{def:strong prime}] \label{proposition: mcp construction of observed and potential outcomes}
    Suppose that the canonical compensator $\alpha = (\alpha^1, \dots, \alpha^d)$ of $N$ with respect to $P$ and $\F_\T$ satisfies the orthogonality condition
    \begin{align}
        \sum_{j \in \mathcal I_d \setminus\{a\} }\Delta {\alpha}_t^a(\varphi) \Delta {\alpha}_t^j(\varphi) = 0 \label{eq: mcp can comp orthogonality}
    \end{align}
    for each $t \in \T$ and $\varphi \in \N_T^d$, and that \eqref{eq: compensator explosion regularity}-\eqref{eq: compensator jump regularity} holds. Then, there exists a probability space $(\Omega', \mathcal{F}', P')$ supporting observed and potential outcomes $N'$ and $\tilde N'$, both $d$-dimensional multivariate counting processes, such that
\begin{enumerate}[label=(\roman*)]
    \item  \label{enum: mcp cons} $N'|_t I(\tau'^a > t) = \tilde N'|_t I(\tau'^a > t)$ for each $t \in \T$.
    \item \label{enum: mcp exch} $\mathbb \Lambda^{\prime a}$ is a compensator of $\mathbb N^{\prime a} = I(\tau^{\prime a} \leq \cdot)$ with respect to both $\{\mathcal{F}_t'\}_t$ and $\{\mathcal{F}_t' \vee \sigma(\tilde N')\}_t$ under $P'$, where $\tau^{\prime a} = \inf\{ s > 0 | N_s^{\prime a} \neq \mathfrak{n}_s^a(N') \}$ (analogous to \eqref{eq: tau a}), and $\{\mathcal{F}_t'\}_t$ is the filtration generated by $N'$.
    \item \label{enum: mcp law} $P'(N' \in \cdot) = P(N \in \cdot)$, and $P'(\tilde N' \in \cdot) = P(\tilde N \in \cdot)$.
\end{enumerate}
    In particular, under these regularity conditions, there exists a probability space on which the consistency and exchangeability conditions in Definition \ref{def:strong prime} hold for any optional outcome functional $\dot Y$, with $Y = \dot Y(N')$ and $\tilde Y = \dot Y(\tilde N')$. Thus, these conditions impose no hidden restrictions on the observed data law.
\end{proposition}
The assumption \eqref{eq: mcp can comp orthogonality} is closely related to the assumption that the observed treatment martingale is orthogonal to the other observed martingale components (see Appendix \ref{appendix: orthogonal counting process martingales}). Analogous assumptions appear implicitly in related work, either through continuity assumptions on the observed data compensator \cite{roysland2011,roysland2012counterfactual,ryalen2018pcancer,roysland2022graphical} or by the way the likelihood is specified \cite{Rytgaard2022}. 
Many existing statistical methods in survival analysis implicitly assume orthogonal martingales \cite{Andersen,aalen2008survivalandevent,martinussen2006dynamic,FlemingHarrington2005,CookLawless2007}, but this literature generally does not discuss how the assumption might be wrong. We show an extreme example of violation of orthogonal martingales in Appendix \ref{appendix: failure of orthogonality assumption}. 

If the consistency condition \eqref{eq: consistency} holds, one might expect the stronger condition
\begin{align}
    Y I(\tau^a \geq \cdot) = \tilde Y I(\tau^a \geq \cdot) \quad P\text{-a.s.}\label{eq: stronger consistency}
\end{align}
to also hold. This strengthened condition, which will be explored further in subsequent work, requires that the observed and potential outcomes of interest agree up to \textit{and including} the regime deviation time.

We show through an example in Appendix \ref{appendix: incompatibility of exchangeability and stronger consistency} that the stronger consistency condition \eqref{eq: stronger consistency} can be incompatible with the exchangeability condition \eqref{eq: exchan disc}. The example violates the orthogonality condition \eqref{eq: mcp can comp orthogonality}.

\section{Deriving the $g$-formula for marked point process data}
\label{section: g-formula for MPP data}

\subsection{A foundation for $g$-formulas in MPP settings}
\label{subsection: derivation via compensator theory}

It is well known that compensators characterize both MPP distributions and likelihood (ratios) \cite{jacobsen2006point,LastBrandt1995marked, jacod1975}. Related work has expressed $g$-formulas where compensators and even MPP likelihoods figure symbolically as 'integrators' \cite{Rytgaard2022,gill2004continuous, ying2024functionaldynamic}, even though MPP distributions and likelihoods are very different objects. In particular, RGvdL present formulas it appears hard to connect with established theory on compensators and integration theory; see \cite[Equation (8) and numerous subsequent expressions]{Rytgaard2022}. These connections are crucial for giving theoretical guarantees, but it is not immediately clear how to make them.

We here give precise meaning to existing formulas by using that the compensator defines a kernel from the underlying measurable space to the product of the positive reals and the mark space. A benefit of making clear sense of the integral involved is the immediate conclusion that the $g$-formula does not depend on choices of the conditional laws involved in it; see Section \ref{subsection: gill and the continuous time g formula} for details, and also GR01 and GR04. To express the $g$-formula we first give required results for its formulation in Lemmas \ref{lemma: P fidi dist} and \ref{lemma: Q compensator}, before we state the formula in Theorem \ref{theorem: representation of functionals}. 

In the "canonical setting" where the underlying measurable space $(\Omega, \F)$ is the canonical space and $N = Id$, we obtain a particularly nice expression (see Definition \ref{definition: canonical setting} in Appendix \ref{appendix: the canonical space of point process realizations}, or \cite[Remark 2.2.5]{LastBrandt1995marked} for details on the canonical setting).   In the case of a general abstract measurable space, the upcoming equation \eqref{eq: P regular conditional distribution new} takes a more complicated form which references the interplay between $N$, its compensator, and the canonical compensator. \footnote{The resulting form is indicated in the proof of the lemma; see equation \eqref{eq: canonical compensator distribution} and the surrounding text.} \footnote{In practical terms, adopting the convention in Lemma \ref{lemma: P fidi dist} amounts to identifying each subject $\omega$ in a study with that individual's point process realization $(T_k(\omega), X_k(\omega))_k$, where $T_k$ is the $k$'th event time and $X_k$ is the $k$'th event type, which is not a restrictive convention.}

\begin{lemma}[MPP distribution (canonical setting)]\label{lemma: P fidi dist}
        Consider, on the measurable space $(\Np{}, \Hp{})$ of $p$-dimensional point process trajectories on $\mathbb R_+$, the point process $N = Id$ (canonical setting). \footnote{See Definition \ref{definition: canonical setting} in Appendix \ref{appendix: the canonical space of point process realizations} or \cite{LastBrandt1995marked} for more details on the canonical setting.} Let $\beta$ be the compensator of $N$ with respect to its natural filtration $\{\Hp{t} \}_{t}$ and $P$, and suppose that $N$ is nonexplosive under $P$. Writing $\{ T_n\}_{n\geq 1}$ for the ordered jump times of $N$, the distribution of $N$ with respect to $P$ when restricted to the optional time $\sigma$-algebra $\Hp{T_n}$ takes the form 
        \begin{align}
            \begin{split}
                dF_n(t_1, x_1, \dots, t_n, x_n) &= I(t_1 < \dots < t_n) \cdot Z^{(0)}(dt_1 \times dx_1) \\
                &\cdot Z^{(1)}((t_k, x_k)_k |_{t_1}, dt_2 \times dx_2 ) \cdot \\
                & \ldots \cdot Z^{(n-1)}((t_k, x_k)_k |_{t_{n-1}},dt_n \times dx_n), \label{eq: P density n new}
    \end{split}
    \end{align} 
where $Z^{(0)}$ is a probability measure on $[0,\infty)\times \I_p$,
and for $i \geq 1$, $Z^{(i)}$ is a kernel from $\mathcal{N}_{t_i}^p$ (the space of trajectories restricted to $[0, t_i]$) to $(t_i, \infty) \times \I_p$, 
where $\I_p = \{1, \dots, p \}$ is the mark space of a $p$-dimensional multivariate counting process. 
The relationship between the terms in \eqref{eq: P density n new} and the compensator is in the canonical setting  
\begin{equation}
\resizebox{\textwidth}{!}{$\displaystyle
\begin{aligned}
Z^{(i)}((t_k, x_k)_k |_{t_i},dt \times dx) &= \prodi_{t_{i} < u < t} \Big( 1 - \beta((t_k, x_k)_k |_{t_i}, du \times \I_p ) \Big) \beta((t_k, x_k)_k |_{t_i}, dt \times dx), \\
Z^{(0)}(dt \times dx) &= \prodi_{0 < u < t} \Big( 1 - \beta((t_k, x_k)_k|_0, du \times \I_p) \Big) \beta((t_k, x_k)_k|_0, dt \times dx),
\end{aligned}
$}
\label{eq: P regular conditional distribution new}
\end{equation}
where $\prodi$ denotes the product-integral. \footnote{For details on product integration see e.g. \cite{Andersen}, or \cite{gill1990survey} for an extensive survey.} Moreover, for every integrable and $\bigvee_{t \geq 0} \H_t^p$-measurable random variable $H$, we have 
\begin{align}
        E_P[H] = \int_{(\Rplus \times \I_p)^{\mathbb N}} H\big( (t_k, x_k)_k \big) dF_\infty(t_1, x_1, \dots), \label{eq: P g-formula new}
\end{align}
    where $\mathbb N = \{1, 2, 3, \dots \}$. Here, $F_\infty$ is the unique probability distribution on $(\Rplus \times \I_p)^{\mathbb N}$ which satisfies $F_n = F_\infty \circ \pi_n^{-1}$ for each $n \in \mathbb N$, where $\pi_n$ is the projection of the first $n$ components of $(\Rplus \times \I_p)^{\mathbb N}$, and $\Rplus$ is the non-negative real numbers.
\end{lemma}


\subsection{The compensator under the $g$-formula distribution $Q$}
\label{subsection: the compensator under the g-formula distribution}

To apply Lemma \ref{lemma: P fidi dist} to the $g$-formula distribution $dQ = W_T dP$, we need the $(Q,\F_\T)$-compensator of $N$. This compensator can be obtained via Girsanov's theorem, a well-known result from stochastic analysis. 

In general, the Girsanov transform determined by Definition \ref{defn: g formula distribution} shows a complicated relationship between the intervention and the derived distribution (see Lemma \ref{lemma: Q compensator general} in Appendix \ref{appendix: proofs}). However, under our regularity condition \eqref{eq: compensator jump regularity} on the intervention, and under the sufficient existence condition \eqref{eq: mcp can comp orthogonality} in Proposition \ref{proposition: mcp construction of observed and potential outcomes}, we obtain the much simpler form shown in Lemma \ref{lemma: Q compensator}.
\begin{lemma}[The $(Q, \F_\T)$-compensator]\label{lemma: Q compensator}
        Assume Definition \ref{def:strong prime} \ref{item: likelihood ratio regularity} holds together with conditions \eqref{eq: compensator jump regularity} and \eqref{eq: mcp can comp orthogonality}. Then,
    \begin{enumerate}[label=\textnormal{(\alph*)}]
        \item \label{eq: N a Q compensator} $\mathfrak {n}^a(N)$ is a $(Q, \F_\T)$-compensator of $N^a$,
        \item \label{eq: N j Q compensator} $\Lambda^j$ is a $(Q, \F_\T)$-compensator of $N^j$ for $j \neq a$.
    \end{enumerate}
\end{lemma}
\begin{definition}[Invariance property (IP)]\label{def: invariance property}
Suppose that for a probability measure $Q^*$ on $\F_T$ with $Q^* \ll P$, there exists a non‑decreasing, predictable process $\Lambda^{a,*}$ such that
\begin{align*}
    &\text{(i) } \Lambda^{a,*} \text{ is a } (Q^*,\F_\T)\text{-compensator of } N^a, \\
    &\text{(ii) } \Lambda^j \text{ is a } (Q^*,\F_\T)\text{-compensator of } N^j \text{ for each } j \neq a.
\end{align*}
We then say that $Q^*$ satisfies the \emph{invariance property (IP)} with respect to $\Lambda^{a,*}$.
\end{definition}
Recall that $\Lambda^j$ is the $(P,\F_\T)$-compensator of $N^j$ (see Section \ref{section: set-up and notation}). Lemma \ref{lemma: Q compensator} shows that, under our regularity conditions, $Q$ satisfies the invariance property (IP) with respect to $\mathfrak n^a(N)$; that is, the treatment compensator is modified according to the intervention while the compensators of the non‑treatment processes remain unchanged relative to $P$. The invariance property in Definition \ref{def: invariance property} appears in many existing works on causal inference with counting process data (see, e.g., \cite{roysland2011}\footnote{In fact, the distribution specified by Lemma \ref{lemma: Q compensator} matches 'counterfactual distributions' as derived in \cite[Section 3]{roysland2011}.} for an early example, and many other works discussed in Section \ref{section: previous work on continuous-time causal inference}). In particular, related works term this invariance property "causal validity" \cite{ryalen2018pcancer,ryalen2019additive,roysland2022graphical}. In our work, (IP) is a consequence of 
our definitions and regularity conditions, not an assumption.


Thus, Lemma \ref{lemma: Q compensator} shows that $g$-formula distributions satisfy (IP), but distributions satisfying (IP) are not necessarily $g$-formula distributions: many distributions satisfying (IP) do not correspond to explicit interventions as in Definition \ref{definition: intervention}. We elaborate on this distinction in Remark \ref{remark: identifying vs invariance distributions}.

\begin{remark}[Distinguishing $g$-formula distributions from general distributions satisfying (IP)]
\label{remark: identifying vs invariance distributions}
The family of distributions $Q^*$ satisfying (IP) with respect to some $\Lambda^{a,*}$ is large: it is indexed by a class of non‑decreasing and predictable processes $\Lambda^{a,*}$ that need only satisfy mild regularity conditions \cite[Corollary 10.1.6 and Theorem 10.2.1]{LastBrandt1995marked}.

In Appendix \ref{appendix: the likelihood ratio process which gives invariance}, we characterize this class through their likelihood‑ratio processes $W_t^* = \frac{dQ^*}{dP}|_{\F_t}$. These $W^*$ processes are typically more complex than our likelihood‑ratio process $W$ (compare \eqref{eq: (IP) lik rat proc} in Appendix \ref{appendix: the likelihood ratio process which gives invariance} with \eqref{eq: K semimart}-\eqref{eq: W dolean}). Moreover, they generally fail to satisfy the conditions \ref{enum: W dolean unique}–\ref{enum: W greater than 0} in Proposition \ref{prop: characterizations of W} --- conditions that uniquely characterize our identifying likelihood ratio process.

$g$-formula distributions (Definition \ref{defn: g formula distribution}) are a strict subclass, corresponding to the choice $\Lambda^{a,*} = \mathfrak n^a(N)$ for some intervention $\mathfrak n^a$. We give in Appendix \ref{appendix: non g formula distribution} an example of a distribution $Q^*$ satisfying (IP) that is not a $g$-formula distribution.
\end{remark}


\subsection{The $g$-formula}
\label{subsection: the g formula}
By Lemma \ref{lemma: Q compensator} we can express the $(Q,\F_\T)$-compensating measure of $N$ as
\begin{align}
    \L(dt \times dx ) =  \sum_{j \in \I_d} \delta_j(dx) d\L^j_t = \delta_a(dx)d\mathfrak{n}_t^a(N) + \sum_{j \in \I_d \setminus \{ a \}} \delta_j(dx) d\Lambda^j_t,
    \label{eq: kappa P compensating measure}
\end{align}
where $\delta_j$ is the Dirac measure located at $j$, or equivalently $\L(dt \times \{ j \}) = d\L_t^j$ for each $j \in \I_d$, where we associate a finite variation process with a random measure in the usual way. 
We obtain a (non-extended) $g$-formula, \cite{robins1986parametric_g}, by substituting this $\mathcal{L}$ into the distributional representation of Lemma \ref{lemma: P fidi dist}. 
\begin{theorem}[$g$-formulas]\label{theorem: representation of functionals}
    If the conditions of Lemma \ref{lemma: P fidi dist}-\ref{lemma: Q compensator} and Theorem \ref{theorem: ipw} hold, then
    \begin{align*}
        E_P[\tilde Y_t] = E_P \big[ W_t Y_t \big] = E_Q[ Y_t ],
    \end{align*}
    for each $t \in \T$, where $W$ is given in \eqref{eq: W dolean}. 
    The rightmost term can be expressed via the $g$-formula distribution, and we get in the canonical setting that \footnote{Since $N = Id$ in the canonical setting, we have by the definitions in Section \ref{subsection: outcomes of interest} that $Y(\omega) = \dot Y(N(\omega)) = \dot Y(\omega)$ for each $\omega$ in the canonical space. That is, $\dot Y = Y$, and the notations $\dot Y$ and $Y$ can be used interchangeably.} 
    \begin{align}
        E_{Q}[Y_t] = \int_{(\T \times \I_p)^{\mathbb N}} Y_t\big( (t_k, x_k)_k \big) dF_\infty^Q(t_1, x_1, \dots).
        \label{eq: g-formula new} 
    \end{align}
    Here, $F_\infty^Q$ is the unique distribution that satisfies $F_n^Q = F_\infty^Q \circ \pi_n^{-1}$, where $F_n^Q$ takes the form \eqref{eq: P density n new} but with $\L$ in \eqref{eq: kappa P compensating measure} replacing  $\beta$ in  \eqref{eq: P regular conditional distribution new}.
\end{theorem}

\begin{example}[continues=exa:determ]
 We present the $g$-formula for the atomic intervention. 
 The $(Q, \F_\T)$-compensator equals
 \begin{align*}
        \L(dt \times dx) &= \delta_a(dx)d\hat \varphi_t + \sum_{j \in \I_d \setminus \{ a \}} \delta_j(dx) d\Lambda^j_t,
    \end{align*} 
 and we can define the finite-dimensional distribution on the canonical space as: 
    \begin{align*}
    \resizebox{\textwidth}{!}{$dF_n^Q(t_1, x_1, \dots, t_n, x_n) =  \prod\limits_{i=1}^n \Big\{ I(t_{i-1}< t_i) 
    \prodi\limits_{t_{i-1}<u < t_i} \big(1 - d\bar \L_{u} \big)\sum_{j \in \I_d} \delta_{j}(dx_i) d\L_{t_i}^{j}  \Big\}$}
    \end{align*}
    where we have defined $t_0 := 0$, and $d\bar \L_t = \L(dt \times \I_d) = \sum_{j \in \I_d} d\L^j_t$. Each $\L^j$ depends predictably on the history, i.e.  $\L_t^j\big( (t_k, x_k)_k \big) = \L_t^j\big( (t_k, x_k)_k|_{t_i} \big)$ on $\{ t_i < t \leq t_{i+1} \}$, but this dependence is not explicitly shown for ease of notation. This leads to the $g$-formula
    \begin{align}
        \begin{split}
            E_P[\tilde Y_t] &= \int_{(\T \times \I_d)^{\mathbb N}}  Y_t\big( (t_k, x_k)_k \big)  dF^Q_\infty(t_1, x_1, \dots )  
        \end{split}
         \label{eq: g-form ex atomic}
    \end{align}
    as in Theorem \ref{theorem: representation of functionals}. We have not seen the formula \eqref{eq: g-form ex atomic} been studied in its general form in the literature. If the outcome functional of interest is survival $\dot Y_t(\varphi) = I(\varphi^d_t = 0)$, then the integrand on the right-hand side of \eqref{eq: g-form ex atomic} is $I\big( \sum_{i \geq 1}I(t_i \leq t, x_i = d) = 0 \big)$. If the $(P,\F_\T)$-compensator of $N^j$ for $j \neq a$
    is absolutely continuous with respect to the Lebesgue measure, i.e. $\Lambda^j = \int_0^\cdot \lambda_s^j ds$, we get the following somewhat simpler finite-dimensional distribution under $Q$ for this intervention:
    \begin{align*}
        dF_n^Q(t_1, x_1, \dots, t_n, x_n) &= e^{- \sum_{j \neq a } \int_0^{t_n}  \lambda_s^j ds} \\
        &\cdot \prod\limits_{i=1}^n  I(t_{i-1}< t_i) \Big( \sum_{j \neq a} \delta_{j}(dx_i) \lambda_{t_i}^{j} dt_i + \delta_{a}(dx_i) d\hat \varphi_{t_i} \Big). 
    \end{align*}
    If additionally $\hat \varphi = 0$ (prevention of treatment), the distribution coincides with the (finite-dimensional) distribution RRND associate with the same intervention. 
\end{example}

\subsection{A characterization of the identifying likelihood ratio process}
\label{subsection: characterizing identifying likelihood ratio}

We give in Proposition \ref{prop: characterizations of W} different characterizations of the process $W$ in \eqref{eq: W dolean}. 
In particular, we show a specific unique connection between $W$ and the consistency and  exchangeability conditions \eqref{eq: consistency}-\eqref{eq: exchan disc}, which we elaborate on immediately following the proposition. 
\begin{proposition} \label{prop: characterizations of W}
    $W$ in \eqref{eq: W dolean} is the unique (up to indistinguishability) process $W^h$ which satisfies conditions \ref{enum: W dolean unique}-\ref{enum: W greater than 0}, where
    \begin{enumerate}[label=(\Alph*)]
        \item \label{enum: W dolean unique} $W^h = \E(\mathbb K^h)$, where $\mathbb K^h_t = \int_0^t h_s d \mbbMa{s}$,
        and $h$ is an $\F_\T$-predictable process such that $\mathbb K^h$ is a local $(P,\F_\T)$-martingale of finite variation and $W^h \geq 0$, and 
        \item \label{enum: W greater than 0} $I(W_\cdot^h > 0) = I(\cdot < \Nastop)$ $P$-a.s. 
\end{enumerate}     
It can furthermore be represented as 
    \begin{align}
        W_t &= \frac{I(\Nastop > t)}{\prodi\limits_{0 < s \leq t} (1 - d\mbbLa{s})} = \frac{\E(-\mbbNa{})_t}{\E(-\mbbLa{})_t} , \label{eq: W unstabilized weight} \\
        W_t &= \frac{I(\Nastop > t)}{e^{- \mathbb \Lambda^{a, c}_t} \prod_{\boldsymbol \sigma_k \leq t}P(\Nastop > \boldsymbol \sigma_k | \F_{\boldsymbol \sigma_k-}) }, \label{eq: W discrete and cont expression} 
    \end{align}  
    where $\{ \boldsymbol \sigma_k \}_k$ are the ordered jump times of $\mbbLa{}$ and $\mathbb \Lambda^{a, c} := \int_0^\cdot I(\Delta \mbbLa{s}=0)d\mbbLa{s}$ is the continuous part of $\mbbLa{}$. 
    
    With the set-up and notation in Section \ref{section: interventions and potential outcomes}, $W$ has the alternative representation 
    \begin{align} 
        W_t &= \frac{I(\Nastop > t)}{\prodi_{0 < s \leq t} (d\Lambda_s^a)^{\Delta N_s^a}(1 - d\Lambda_s^a )^{1 - \Delta N_s^a}}.  \label{eq: Gill likelihood}
    \end{align}
    In this case, $\mathbb K_t^a = -\int_0^{t\wedge \Nastop} \frac{1}{1 - \Delta \na{s}(N) - \Delta \Lambda_s^a} dM_s^a,$
    and $\mbbLa{}$ coincides with $~^{\Nastop}\Lambda^a + ~^{\Nastop}\na{}(N) - 2 \int_0^{\cdot \wedge \Nastop}\Delta \na{s}(N) d\Lambda_s^a$, where we use the notation $~^{\tau}Z_t := Z_{\tau \wedge t}$ for a process $Z$ and random time $\tau$. 
\end{proposition}
Any MPP likelihood ratio is a stochastic exponential of a predictable process integrated against a basic martingale \cite{jacod1975}. Of all exponentials driven by the 'exchangeability martingale' $\mbbMa{}$, $W$ is the only such process that is nonzero exactly as dictated by the consistency condition \eqref{eq: consistency}. Since $W$ is the unique process satisfying \ref{enum: W dolean unique}-\ref{enum: W greater than 0} above, it is also the unique such process that is non-decreasing up to right before $\Nastop$, i.e. progressively upweighting subjects as long as they follow the regime and assigning a weight of zero to the subjects who no longer follow the regime.

The evaluability conditions \eqref{eq: K jump to zero condition} and \eqref{eq: W likelihood ratio regularity} can be understood informally through the expressions \eqref{eq: W unstabilized weight}-\eqref{eq: W discrete and cont expression} as follows. For the subjects who have followed the regime up to any point in time, the conditional probability of following the regime in the next infinitesimal time increment given the strict observed history should not be too small, so that the cumulative product of these probabilities (over infinitesimal time increments) is not too small. \footnote{Note that we in \eqref{eq: W discrete and cont expression} can replace $P(\Nastop > \boldsymbol \sigma_k | \F_{\boldsymbol \sigma_k-})$ with $P(\Nastop > \boldsymbol \sigma_k | \F_{\boldsymbol \sigma_k-}, \Nastop \geq \boldsymbol \sigma_k)$, where the latter can be understood as $P_{ \{ \Nastop \geq \boldsymbol \sigma_k \} }(\Nastop > \boldsymbol \sigma_k | \F_{\boldsymbol \sigma_k-})$, the conditional expectation with respect to the event probability measure $P_{ \{ \Nastop \geq \boldsymbol \sigma_k \} } (\cdot) := P(\cdot \cap \{ \Nastop \geq \boldsymbol \sigma_k \})/P(\Nastop \geq \boldsymbol \sigma_k)$.} In particular, the product in \eqref{eq: W unstabilized weight} factorizes over the increments where there is a strictly positive probability of deviating from the regime, which admits an explicit representation as a product of conditional probabilities, and the increments where there is only an infinitesimal probability of deviating, which is a product-integral over the continuous part of $\mbbLa{}$ and reduces to the exponential of minus the continuous part, which leads to \eqref{eq: W discrete and cont expression}.

\subsection{Gill and Robins, and the validity of $g$-formulas}
\label{subsection: gill and the continuous time g formula}
    GR04 conjectured that a $g$-formula similar to ours (equation \eqref{eq: g-formula new}) could be derived under versions of the consistency, exchangeability, and positivity assumptions \cite[p. 4]{gill2004continuous}. \footnote{In particular, their 'Treatment plans,' \cite[Section 3]{gill2004continuous}, which "prescribes subsequent action timepoints, ... so long as no further longitudinal data timepoint intervenes" can be precisely formulated as predictable interventions as we define them.} They also mention three possible complications associated with formally establishing the validity of the formula, which they refer to as "correctness," "uniqueness," and the "no-explosions condition." \footnote{GR01 tackle comparable uniqueness and correctness issues when extending Robins's theory to settings where the random variables involved are continuous, as opposed to discrete. They address the correctness issue by imposing continuity restrictions on the conditional laws involved in the formula, thus taking a different approach than we follow here. Interestingly, we have not seen the issues raised in \cite{gill2001complex,gill2004continuous,yu2002construction} mentioned in the discrete-time literature beyond these works.} The developments in this paper provide a way of addressing and resolving their concerns, which we elaborate on here. We can resolve these issues because we have identified the compensating measure in \eqref{eq: kappa P compensating measure}, given a rigorous formulation of the integral involved (Lemma \ref{lemma: P fidi dist}), and because we characterize the potential outcomes law (Definition \ref{definition: potential outcome process}). 
     \begin{itemize}
         \item \underline{Correctness:} GR04 raise the question of whether and when the right-hand side of \eqref{eq: g-formula new} does, in fact, determine the law of the potential outcome variable of interest.  In theories involving a finite number of random variables, if treatment variables are continuous, it is known that one can construct potential outcomes satisfying standard consistency and sequential exchangeability assumptions having a distribution that is not given by the $g$-formula; see \cite{gill2001complex,yu2002construction} for examples and discussions. 
         
         In our work, as long as we characterize $\cfproc{}$ as in Definition \ref{definition: potential outcome process} and $Y$ and $ \tilde Y$ as in Section \ref{subsection: outcomes of interest}, Theorem \ref{theorem: representation of functionals} says that the $P$-expectation $\tilde Y_t$ is given by the right-hand side of \eqref{eq: g-formula new} under the conditions of Definition \ref{def:strong prime}. That is, we have made the law of $\cfproc{}$ and thus $\tilde Y$ explicit, and our formulas are true when the conditions of our theorems are true. Since the $g$-formula is equivalent to the IPW formula in Theorem \ref{theorem: ipw}, the IPW formulas presented here are correct in the same sense.  
         \item \underline{Uniqueness:} GR04's second concern is the formula's possible dependence on the selected versions of the conditional probability laws involved. Allowing versions to be chosen haphazardly could conceivably lead to an ill-defined expression.
        
        In this case, uniqueness results for the distribution of MPPs offer a solution. For instance, \cite[Theorem 8.2.3]{LastBrandt1995marked} says that the MPP distribution is uniquely determined by the canonical compensator. 
        \footnote{GR04 anticipate that the conditional distributions might be chosen in such a canonical manner.} Moreover, owing to the uniqueness of compensators, any other choice of compensator will give the same answer. This is because any two compensators agree except on a null set, i.e. a set which does not contribute to the numerical value of the integral. 
        \item \underline{No explosions:} Finally, GR04 recognize the possibility that explosions can be introduced when a given regime is imposed, and the question is whether the $g$-formula has total probability 1. That our $g$-formulas have total probability 1 is immediate under our assumptions \eqref{eq: K jump to zero condition} and \eqref{eq: W likelihood ratio regularity}, when $W$ is a likelihood ratio process. 
     \end{itemize}


\section{Relation to discrete-time theories}
\label{section: relation to discrete-time theories}

We can embed commonly studied discrete-time data structures in the point process data structure by adding specific restrictions on the compensator processes. This allows us to establish connections between our definitions and definitions that are used in existing causal inference theories, and draw parallels between the identification results contained in our paper and well-known existing results from the discrete-time literature \cite{robins1986parametric_g,robins1997complex,richardson_single_2013}. The result is stated as a proposition, followed by some comments. More detailed comparisons with discrete-time approaches will be made in future work.
\begin{proposition}\label{proposition: relation to existing assumptions}
     Suppose that the indices of $N$ are ordered such that $N=(N^1, \allowbreak \dots, \allowbreak N^{d-1}, \allowbreak N^a)$, and
     consider real numbers $\{ \theta_{k}\}_k$ and $\{ \ell_{k}^j\}_{k,j}$ for $k = 1,\dots,K+1$ and $ j=1,\dots, d-1$ such that
    \begin{align}
        0 < \ell_1^1 < \dots < \ell_1^{d-1} < \theta_{1} < \ell_2^1 < \dots < \ell_2^{d-1} < \theta_{2} < \dots < \theta_{K+1} < T. \label{eq: ffr times}
    \end{align}
Suppose that the $N^a$-compensator is constant  except at the times $\{ \theta_{k} \}_{k}$ where it may jump, and that each $N^j$-compensator for $j \in \Idminaset$ is constant except at the times $\{ \ell_k^j \}_{k}$ where it may jump. Consider a static regime $\na{}$ as in Example \ref{example: atomic} such that the possible action times, i.e. jump times of $\na{}$, are contained in $\{ \theta_{k} \}_k$. Put $A_k = \Delta N_{\theta_{k}}^a$, $a_k= \Delta \na{\theta_k}$, $L_k^j = \Delta N_{\ell_k^j}^j,$ $\tilde L_k^j = \Delta \cfproc{j}_{\ell_k^j},$ and $\tilde A_k =  \Delta \cfproc{a}_{\theta_k}$ for $j = 1, \dots, d-1$ and $k=1, \dots, K+1$, where $\cfproc{}$ is a potential outcome process associated with the intervention $\na{}$ as in Definition \ref{definition: potential outcome process}. Suppose that $P(\bar A_{K+1} = \bar a_{K+1}) > 0$, where we write $\bar V_k = (V_1, \dots, V_k)$ for a discrete process $\{ V_k \}_k$. The sets of observed covariates and treatments are then temporally ordered according to \eqref{eq: ffr times}, i.e.
    \begin{align}
       \langle L_1^1,\dots,  L_1^{d-1}, A_1, L_2^1, \dots, L_2^{d-1}, A_2, \dots, A_{K+1} \rangle. \label{eq: RCISTG ordering}
    \end{align}
    Moreover, defining $L_k = \{ L_k^j | j=1,\dots, d-1\}$, and $L_{k,<j}$ as the 'past' of the observed non-treatment variables up to before $L_k^j$;
    \begin{align}
        L_{k,<j} = \big\{ L_i^l | ( 1 \leq i < k \text{ and } 1 \leq l \leq d-1 ) \text{ or } (i=k  \text{ and } 1 \leq l < j ) \big\}, \label{eq: non treatment covariate past}
    \end{align}
    we have that 
    \begin{itemize}
    \item \textbf{Equivalence of exchangeability and consistency conditions.} \label{item: reduction of exchangeability}  
    The exchangeability condition \eqref{eq: exchan disc} with $\tilde Y = (\tilde N^{1}, \dots, \tilde N^{a-1},\tilde N^{a+1},\dots, \tilde N^{d})$ 
    is equivalent to the independencies 
    \begin{align}
        \underline{\tilde L}_{m+1} \indep A_m | \bar L_{m}, \bar A_{m-1} = \bar a_{m-1} \quad \text{ for each  } m = 1,\dots,K, \label{eq: RCISTG indep}
    \end{align}
    where $\tilde L_k = \{ \tilde L_k^j | j=1,\dots, d-1 \}$, and $\underline{V}_k = (V_k, V_{k+1}, \dots)$. That is, setting $\mathbf{ A} = \{ A_k| k = 1, \dots, K+1 \}$, $\mathbf{ L} = \{ L_k| k = 1, \dots, K+1 \}$, and $\mathbf{ V} = \mathbf{ A} \cup \mathbf{L}$, \eqref{eq: exchan disc} is with the given $\tilde Y$ equivalent to the independence conditions in the RCISTG $(\mathbf{A}, \mathbf{ V})$ model for this intervention; see \cite[Definition 62]{richardson_single_2013}. 

    Furthermore, our consistency condition is equivalent to
    \begin{align}
        L_k^j = \tilde L_k^j \quad P\text{-a.s. on } \{ \bar A_{k-1} = \bar a_{k-1} \} \label{eq: consistency discrete}
    \end{align}
    for $j = 1, \dots, d-1$ and $k=1, \dots, K$.
    
    \item \textbf{Equivalence of positivity conditions.} \label{item: reduction of positivity} 
    Our positivity condition \eqref{eq: K jump to zero condition} is equivalent to 
    \begin{align}
        P( A_k = a_k | \bar L_k, \bar A_k) > 0 \quad  P\text{-a.s. on }  \{ \bar A_{k-1} = \bar a_{k-1} \} \text{ for each }k. \label{eq: hr condition}
    \end{align}
    Moreover, under \eqref{eq: K jump to zero condition} and the premises of this proposition, the likelihood ratio regularity condition in Definition \ref{def:strong prime} \ref{item: likelihood ratio regularity} holds. 

    \item \textbf{Reduction of the data-generating law.} \label{item: reduction of data generating}
    Each canonical compensator $\alpha^j$ fixed in Definition \ref{definition: potential outcome process} can be associated with functions $\{ f_k^j\}_k$, where each $f_k^j(L_{k,<j}, \bar A_{k-1})$ is some version of the conditional expectation of $L_k^j$ given $\sigma(L_{k,<j}, \bar A_{k-1})$ under $P$. Specifically, the following holds surely: 
    \begin{align*}
        \Delta \alpha_{\ell_k^j}^j(N) = f_k^j(L_{k,<j}, \bar A_{k-1}) \text{ for each $k$ and $j \neq a$, }
    \end{align*}
    where $\alpha^j$ in $\alpha = (\alpha^1, \dots, \alpha^d)$ is the selected version of the canonical compensator in Definition \ref{definition: potential outcome process}. 
    Moreover, Definition \ref{definition: potential outcome process} is equivalent to: 
    \begin{align}
        \begin{split}
            \tilde A_k &= a_k, \\
            P(\tilde L_k^j = 1 | \tilde L_{k,<j}, \bar {\tilde A}_{k-1} = \bar a_{k-1}) &= f_k^j(\tilde L_{k,<j}, \bar a_{k-1} ),
        \end{split} \label{eq: data-generating static}
    \end{align}
    $P$-a.s. for each $k=1, \dots, K+1$ and $j=1, \dots, d-1$, where $\tilde L_{k,<j}$ is the 'past' of the non-treatment variables of the potential outcome process up to before $\tilde L_k^j$, defined similarly to $L_{k,<j}$ in \eqref{eq: non treatment covariate past}.
\end{itemize}
\end{proposition}

A larger set of outcomes are identified under the FRCISTG independence assumptions than the RCISTG independence assumptions. The former set of assumptions allows for the identification of expected outcomes under a regime that is a function of the natural value of treatment; see \cite[Appendix C.2]{richardson_single_2013}. We do not consider such regimes, and thus the required exchangeability conditions are therefore not developed here. The consistency condition \eqref{eq: consistency discrete} coincides with \cite[Technical point 19.2]{Hernanrobins2021causal} (see also \cite[Definition 60 (ii)]{richardson_single_2013}). \eqref{eq: hr condition} is an 'almost surely' formulation of the condition in \cite[Technical point 19.2]{Hernanrobins2021causal}. The lower equality in \eqref{eq: data-generating static} is what Richardson and Robins call "modularity" \cite[Definition 15]{richardson_single_2013}. \footnote{To make the connection explicit, note that $P(\tilde L_k^j = 1 | \tilde L_{k,<j}, \bar {\tilde A}_{k-1} = \bar a_{k-1}) $ coincides with $P(\tilde L_k^j = 1 | \tilde L_{k,<j})$ because $\bar {\tilde A}_{k-1} = \bar a_{k-1}$ $P$-a.s., and \eqref{eq: data-generating static} thus states, as \cite[Equation (30)]{richardson_single_2013}, that 
$$ P\big(\tilde L_k^j = 1 \big| \tilde L_{k,<j} = l_{k,<j}\big) = f_k^j(l_{k,<j}, \bar a_{k-1} ) = P\big(L_k^j=1\big|L_{k,<j} = l_{k,<j}, \bar A_{k-1} = \bar a_{k-1}\big), $$
when the conditioning sets involved have positive probability.}

Instead of characterizing $\cfproc{}$ via its law as in Definition \ref{definition: potential outcome process}, one could alternatively consider $\cfproc{}$ as being explicitly realized from its law starting with a collection of specified "error terms." Such explicit constructions could lead to further connections with definitions which recursively define potential outcomes based on earlier variables in the ordering \cite[Definition 1]{richardson_single_2013}. However, we do not follow this path here.

\section{Identification with multiple interventions in the general MPP setting}
\label{section: identification with multiple interventions in the general MPP setting}

This section extends the results presented in the preceding sections to interventions on multiple components in general MPP settings. Assuming familiarity with the earlier sections, we first introduce the necessary notation for multiple interventions, then present the main identification theorem (Theorem \ref{theorem: multiple marked interventions}), before examining some of its consequences through examples.

\subsection{Setup, notation, and a general identification result for multiple interventions}
\label{subsec: setup multiple interventions}

We consider a non-explosive observed data MPP $N$ on the time domain $\T$, with mark space $(X, \mathcal{X})$ assumed to be a Borel space. We assume $N$ consists of components $(N^i)_{i \in \mathcal{I}_d}$, where each $N^i$ is itself an MPP operating on its own mark space $(X^i, \mathcal{X}^i)$, where the $X^i$'s are pairwise disjoint subsets of $X$. In addition, we include here an observed baseline random element $L$ taking values in $(S, \mathcal{S})$.

$N$ takes values in the canonical space $(\N_T^X, \H_T^X)$, and $(L, N)$ takes values in the canonical space $(S \times \N_T^X, \S \otimes \H_T^X)$. \footnote{Appendix \ref{appendix: the canonical space of point process realizations} or \cite{LastBrandt1995marked, jacobsen2006point} provide details on the canonical space of MPPs.} Following the definitions of \cite{LastBrandt1995marked}, which match those used in Section \ref{subsection: interventions}, we say that a process $Z$ on $(S \times \N_T^X, \S \otimes \H_T^X)$ is \textit{optional} if $Z_t(l,\varphi) = Z_t(l,\varphi|_t)$ and \textit{predictable} if $Z_t(l,\varphi) = Z_t(l,\varphi|_{t-})$ for every $(l,\varphi) \in S \times \N_T^X$, and similarly for processes defined on $(\N_T^X, \H_T^X)$. Aligning with \cite{LastBrandt1995marked}, we adopt analogous definitions for kernels and random measures defined on these spaces.


We consider interventions on components of $N = (N^i)_{i \in \mathcal{I}_d}$ indexed by $J \subseteq \mathcal{I}_d$. Analogously to Definition \ref{definition: intervention}, we specify an intervention $\mathfrak{n}^j$ on component $j \in J$ as a predictable counting measure on the canonical space, i.e. satisfying $\mathfrak{n}^j(l, \varphi, dt \times dx) = \mathfrak{n}^j(l, \varphi|_{t-}, dt \times dx)$ for each $(l,\varphi) \in S \times \N_T^X$ and $t \in \T$. 

Each $\mathfrak{n}^j$ is a treatment rule that fixes values of the component $j$. More precisely, the mapping $(l, \varphi) \mapsto \mathfrak n^j(l,\varphi)$ is a mapping from $S \times \N_T^X$ to $\N_T^{X^j}$, the canonical space of $N^j$. Thus, $\mathfrak n^j$ produces a feasible realization of the component $j$, which may depend on the strict past of $l, \varphi$, analogously to Definition \ref{definition: intervention}.

By \cite[Theorem 4.2.2]{LastBrandt1995marked}, there exists a canonical compensator $\alpha$ of $N$ with respect to $P$ and $\F_\T$, the filtration generated by $(L, N)$ on $\T$. This is a uniquely defined predictable kernel which satisfies
\begin{align}
\begin{split}
    \alpha(L, N, dt \times dx) & \text{ defines a compensator of }  N(dt \times dx) \\
    & \text{with respect to } P \text{ and } \F_\T.
\end{split} 
    \label{eq: alpha canonical compensator}
\end{align}
The potential outcome process $\tilde{N}$ can be characterized by this canonical compensator and the interventions $\mathfrak n^j$ under study, analogously to Definition \ref{definition: potential outcome process} in Section \ref{subsection: potential outcome processes}. Appendix \ref{appendix: general MPP potential outcomes} gives details on this characterization in the general MPP setting.

Analogously to Section \ref{subsection: outcomes of interest}, the outcome functional of interest $\dot{Y}$ is a bounded optional stochastic process on $(\N_T^X, \H_T^X)$, i.e. satisfying $\dot Y_t(\varphi) = \dot Y_t(\varphi|_t)$ for each $(t, \varphi) \in \T \times \N_T^X$. We define the observed outcome as $Y := \dot{Y}(N)$ and the potential outcome as $\tilde{Y} := \dot{Y}(\tilde{N})$, where $\tilde{N}$ is the potential outcome process arising from these interventions.

From the interventions and the observed data, we recover the regime‑specific deviation times
\begin{align}
    \tau^j := \inf\big\{t > 0 | N^j((0,t] \times D) \neq \mathfrak{n}^j(L, N, (0,t] \times D) \text{ for some } D \in \mathcal{X}^j\big\}, \label{eq: tau j deviation time}
\end{align}
and the overall deviation time $\tau^J := \wedge_{j \in J} \tau^j$, analogously to \eqref{eq: tau a}. To focus on the important things, we assume in the following that $X^j$ is finite for each $j \in J$; then $\tau^J$ is an optional time with respect to $\F_\T$, and $\mathbb{N}^J_t := I(\tau^J \leq t)$ has a compensator $\mathbb{\Lambda}^J$ with respect to $P$ and $\F_\T$.

Analogous to the construction in Section \ref{subsection: identifying conditions}, we assume 
that
\begin{align}
    P\Big(\int_0^T \frac{d\mathbb{\Lambda}^J_s}{1 - \Delta\mathbb{\Lambda}^J_s} < \infty\Big) = 1, \label{eq: J int positivity}
\end{align}
so that
\begin{align}
    \mathbb{K}^J &:= -\int_0^\cdot \frac{d\mathbb{N}^J_s - d\mathbb{\Lambda}^J_s}{1 - \Delta\mathbb{\Lambda}^J_s}, \label{eq: mbbk J well defined}
\end{align}
is well-defined, and both $\mathbb{K}^J$ and $\E(\mathbb{K}^J)$ are local $(P, \F_\T)$-martingales.
\begin{theorem}[Identification with multiple interventions in the general MPP setting]\label{theorem: multiple marked interventions}
With the setup and notation thus far in Section \ref{section: identification with multiple interventions in the general MPP setting}, 
assume that the following regularity conditions hold on the canonical compensator $\alpha$ in \eqref{eq: alpha canonical compensator}, the interventions $\mathfrak n^j$, and each $j,h \in J$ such that $j \neq h$: 
\begin{align}
        \alpha\big(l,\varphi, \{t\} \times X^j\big)  \alpha\big(l,\varphi, \{t\} \times X^h\big) &= 0, \label{eq: can comp j h} \\
        \alpha\big(l,\varphi, \{t\} \times X^j\big)  \alpha\big(l,\varphi, \{t\} \times X^{\setminus J}\big) &= 0, \label{eq: can comp j setminus J}  \\
        \mathfrak{n}^j(l,\varphi, \{t\} \times X^j)  \alpha\big(l,\varphi, \{t\} \times X^{\setminus J}\big) &= 0, \label{eq: intervention can comp} \\
        \mathfrak{n}^j(l,\varphi, \{t\} \times X^j)  \mathfrak{n}^h(l,\varphi, \{t\} \times X^h) &= 0, \label{eq: interv interv}
\end{align}
for each $l, \varphi$ and $t$, where $X^{\setminus J} := \cup_{j \in \mathcal{I}_d \setminus J} X^j$. Under the identifying conditions
\begin{enumerate}[label=(\roman*)]
    \item \label{enum: mpp consistency} \textbf{Consistency:} $Y_t I(\tau^J > t) = \tilde{Y}_t I(\tau^J > t)$ $P$-a.s. for each $t \in \T$,
    \item \label{enum: mpp exchangeability} \textbf{Exchangeability:} $\mathbb{\Lambda}^J$ defines a compensator of $\mathbb{N}^J$ with respect to both $\F_\T$ and $\F_\T^{\tilde Y}$ under $P$, where $\F_\T^{\tilde Y} = \{\F_t \vee \sigma(\tilde{Y})\}_{t \in \T}$,
    \item \label{enum: mpp likelihood ratio regularity} \textbf{Likelihood-ratio regularity:} $E_P[\mathcal{E}(\mathbb{K}^J)_t] = 1$ for each $t \in \T$,
\end{enumerate}
we have for each $t \in \T$ that
\begin{align}
    E_P[\tilde{Y}_t] = E_P[\mathcal{E}(\mathbb{K}^J)_t Y_t] = E_Q[Y_t], \label{eq: g-formula multiple interventions}
\end{align}
where $dQ = \mathcal{E}(\mathbb{K}^J)_T dP$. Assuming the canonical setting where the underlying measurable space is $(S \times \mathcal N_T^X, \mathcal S \otimes \mathcal H_T^X)$ and $(L, N) = Id$, \footnote{See Appendix \ref{appendix: the canonical space of point process realizations} or \cite{LastBrandt1995marked,jacobsen2006point} for details on the canonical setting.} the $g$-formula can be represented as
    \begin{align}
         E_{P}[\tilde Y_t] = E_{Q}[Y_t] = \int_S \int_{(\T \times X)^{\mathbb N}} Y_t\big( (t_k, x_k)_k \big) dF_\infty^Q(t_1, x_1, \dots | l) dF_L(l), \label{eq: thm 3 g formula}
    \end{align}
    where $dF_L$ is the distribution of $L$ under $P$, and $dF_\infty^Q(t_1, x_1, \dots | l)$ is a regular version of the conditional distribution of $N$ given $L=l$ under $Q$. Similarly to \eqref{eq: P density n new}, its finite-dimensional conditional distribution takes the form 
    \begin{align*}
        dF_n^Q(t_1, x_1, \dots, t_n, x_n|l) &= I(t_1 < \dots < t_n) \cdot Z^{(0)}(l,dt_1 \times dx_1) \\
        &\cdot Z^{(1)}( (t_k, x_k)_k |_{t_1} , l, dt_2 \times dx_2) \cdot \\
        & \cdot \ldots \cdot Z^{(n-1)}((t_k, x_k)_k |_{t_{n-1}} , l, dt_n \times dx_n).
    \end{align*}
    Similarly to \eqref{eq: P regular conditional distribution new}, this finite-dimensional distribution is determined by the kernels
    \begin{equation}
\resizebox{\textwidth}{!}{$\displaystyle
\begin{aligned}
Z^{(i)}((t_k, x_k)_k |_{t_i},l,dt \times dx ) &= \prodi_{t_{i} < u < t} ( 1 -  \mathcal L((t_k, x_k)_k |_{t_i},l, du \times X ) ) \mathcal L((t_k, x_k)_k |_{t_i},l, dt \times dx),  \\
        Z^{(0)}(l,dt \times dx) &= \prodi_{0 < u < t} ( 1 -   \mathcal L((t_k, x_k)_k|_0,l, du \times X) )  \mathcal L((t_k, x_k)_k|_0,l, dt \times dx),
\end{aligned}
$}
\notag
\end{equation}
    where $\mathcal L$ in the previous equation is the $(Q, \F_\T)$-compensating measure of $N$, which under the regularity conditions \eqref{eq: can comp j h}-\eqref{eq: interv interv} is given by
    \begin{align}
        \mathcal L(dt \times dx) = \sum_{j \in J} I(x \in X^j) \mathfrak n^j(L, N, dt \times dx) + \sum_{j \in \mathcal I_d \setminus J} I(x \in X^j) \Lambda^j(dt \times dx), \label{eq: Q compensating measure}
    \end{align}
    where we have in \eqref{eq: Q compensating measure} suppressed the dependence on $(l, \varphi) = \omega \in \Omega = S \times \mathcal N_T^X$. Here, $\Lambda^j$ is the compensator of $N^j$ with respect to $P$ and $\F_\T$.
\end{theorem}



The regularity conditions \eqref{eq: can comp j h}-\eqref{eq: interv interv} are sufficient to ensure the joint existence, on the same probability space, of observed and potential outcomes with the respective correct laws such that the consistency and exchangeability conditions in Theorem \ref{theorem: multiple marked interventions} are satisfied. 
See Appendix \ref{appendix: construction of potential outcomes} for details; in particular, compare the conditions \eqref{eq: can comp j h}-\eqref{eq: interv interv} with \eqref{eq: ort obs multi marked}. These regularity conditions mirror the conditions used to establish Proposition \ref{proposition: mcp construction of observed and potential outcomes} in Section \ref{section: construction of potential outcomes}; \eqref{eq: can comp j h}-\eqref{eq: can comp j setminus J} parallel \eqref{eq: mcp can comp orthogonality}, and \eqref{eq: intervention can comp}-\eqref{eq: interv interv} parallel \eqref{eq: compensator jump regularity}. In plain English, under \eqref{eq: can comp j h}-\eqref{eq: interv interv}, the consistency and exchangeability conditions in Theorem \ref{theorem: multiple marked interventions} impose no hidden restrictions on the observed data law.

The exchangeability condition in Theorem \ref{theorem: multiple marked interventions} \ref{enum: mpp exchangeability} is implied by the assumption that the $(P,\F_\T)$-compensator of $~^{\tau^J}N^j$ is also a $(P,\F_\T^{\tilde Y})$-compensator of $~^{\tau^J}N^j$ for each $j \in J$, as we show in Lemma \ref{lemma: MPP exchangeability 2026} (Appendix \ref{appendix: supporting calculations}). This latter condition is in turn implied by certain sequential independence conditions involving latent event times and marks; see \eqref{eq: Tkj cond indep}-\eqref{eq: Vkj Tkj cond indep} in Appendix \ref{appendix: construction of potential outcomes} (stated there with different notation and for the full potential outcome process, rather than the potential outcome of interest).

\subsection{Illustrative examples}
\label{subsec: illustrative examples}
Similar to the examples in Section \ref{subsection: examples of interventions and potential optcomes processes}, the following two examples apply Theorem \ref{theorem: multiple marked interventions} to an extension of the clinical scenario from Example \ref{example: deterministic dynamic predictable refinement}, which resemble settings studied in the literature. Our formalization differs from other approaches in that we use our explicit definitions of interventions, formally derive formulas from identification conditions, and make the assumed decision-making narratives explicit. We provide specific comments following the examples.

The observed data structure consists of a baseline random element $L$ measured at or before $t = 0$, and an MPP $N = (N^a, N^v, N^c, N^d, N^\ell)$, where the components are
\begin{itemize}
    \item $N^a$, an MPP for treatment allocations,
    \item $N^v$, a counting process recording clinic visits (during which treatment assessments occur),
    \item $N^c$, a right-censoring counting process,
    \item $N^d$, a death counting process,
    \item $N^\ell$, an MPP recording other covariates.
\end{itemize}
The observed data filtration $\F_\T$ is generated by $L$ and $N$.

As in Example \ref{example: deterministic dynamic predictable refinement}, we assume there is a delay $\delta > 0$ between assessment and allocation. (We discuss an alternative observed data structure in Section \ref{subsubsection: evaluating identification claims by RGvdL}.) That is, the doctor first reviews the patient's information (assessment) and then, after a short delay, prescribes treatment (allocation). Specifically, we assume that
\begin{align}
    \bar N^a(dt) = N^v(dt - \delta), \label{eq: delay assumption}
\end{align}
where $\bar N^a(dt) := N^a(dt \times X^a)$ counts treatment allocations of all types, and where we use the convention $Z(dt)I(t < 0) = 0$ for a random measure $Z$. The detailed structure of the delay is not essential; we therefore take $\delta$ to be deterministic for simplicity. \footnote{More generally, we can e.g. allow for delay structures on the form $T_k^a = T_k^v + \delta_k$ with $\delta_k \geq \epsilon > 0$ and  $\F_{T_k^v}$-measurable $\delta_k$, where $\{ T_k^a \}_k$ and $\{T_k^v\}_k$ are the jump times of $\bar N^a$ and $N^v$ respectively. This allows for subject-specific delays which may depend on clinic, physician, etc.} 
Formally, the delay assumption implies that the process $\bar N^a$ counting treatment allocations is predictable with respect to the observed filtration. 

Theorem \ref{theorem: multiple marked interventions} is formulated for MPP components with general mark spaces. To connect the examples to the theorem, we embed each one-dimensional counting process component as follows: for $i \in \{v,c,d\}$, we set the corresponding mark set to the singleton $X^i = \{ i \}$ and identify the counting process $N^i(dt)$ with the MPP $N^i(dt \times dx) := \delta_i(dx) N^i(dt)$. If $i$ is also an intervention component, i.e. $i \in J$, we similarly make the identification $\mathfrak n^i(l,\varphi,dt \times dx) = \delta_i(dx) \mathfrak n^i(l,\varphi,dt)$. 
With these identifications, the observed data structure and the interventions $\mathfrak n^j$ match the setup of Section \ref{subsec: setup multiple interventions}.

The outcome of interest is survival at time $t$, $Y_t = I(N_t^d = 0)$, with potential outcome of interest $\tilde Y_t = I(\tilde N_t^d = 0)$, where $\tilde N$ is the potential outcome process arising under a given set of interventions.

\begin{example}[Joint intervention on schedule, assignment, and censoring]\label{example: intervention on treatment and schedule}
Consider a medical study with the observed data structure just described. 
A researcher is interested in the effect of a protocol that specifies both the timing of treatment visits (the schedule) and subsequent treatment assignments, while preventing censoring. This joint intervention modifies visit schedule, treatment assignment, and ensures that everybody remains under follow-up (no censoring). Consequently, the intervention components are $J = \{ v, a, c \}$.

The schedule intervention is
$$
\mathfrak{n}^v(l,\varphi, dt) = \kappa^*(l,\varphi, dt),
$$
where $\kappa^*$ is a predictable kernel specifying the visit schedule based on strictly earlier events; that is, $\kappa^*(l,\varphi, dt) = \kappa^*(l,\varphi|_{t-}, dt)$, and $\kappa^*(l,\varphi,(0,\cdot])$ is a one-dimensional counting process trajectory for each $l$ and $\varphi$. 

The allocation intervention is 
\begin{align}
    \mathfrak{n}^a(l,\varphi, dt \times dx) = \pi^*(l, \varphi, t, dx) \varphi^v(dt - \delta), \label{eq: allocation assignment}
\end{align}
where $\varphi^v(dt - \delta) = \bar\varphi^a(dt)$ by the delay assumption \eqref{eq: delay assumption} (i.e., allocations occur $\delta$ time units after the corresponding assessments). $\pi^*$ is a predictable kernel specifying treatment assignment, i.e. $\pi^*(l,\varphi,t,dx)$ is a rule that fixes treatment allocation at $t$ given the strictly earlier events $l,\varphi|_{t-}$. The censoring prevention intervention is $\mathfrak n^c = 0$.

Define deviation times $\tau^v, \tau^a, \tau^c$ as in \eqref{eq: tau j deviation time}, set $\tau^J = \tau^v \wedge \tau^a \wedge \tau^c$, and suppose the $(P, \F_\T)$-compensator $\mathbb \Lambda^J$ of $I(\tau^J \leq \cdot)$ satisfies the positivity condition \eqref{eq: J int positivity}, so that $\mathbb K^J$ in \eqref{eq: mbbk J well defined} is well-defined. Under the conditions in Theorem \ref{theorem: multiple marked interventions}, we get
$$
    E_P[\tilde{Y}_t] = E_P[\mathcal{E}(\mathbb{K}^J)_t Y_t] = E_Q[Y_t],
$$
where $dQ = \mathcal{E}(\mathbb{K}^J)_T dP$. The compensator \eqref{eq: Q compensating measure} figuring in the $g$-formula \eqref{eq: thm 3 g formula} is
\begin{align*}
    \mathcal{L}(dt \times dx) &= \delta_{v}(dx) \kappa^*(L,N,dt) + I(x \in X^a) \pi^*(L,N,t, dx)\bar N^a(dt) \\
    &\quad + \delta_c(dx) \cdot 0 + I(x \in X^\ell) \Lambda^\ell(dt \times dx) + \delta_d(dx)\Lambda^d(dt),
\end{align*}
where each $\Lambda^i$ is the $(P, \F_\T)$-compensator $N^i$.
\end{example}

\begin{example}[Intervention on treatment allocation while preventing censoring]\label{example: delayed allocation}
Building on the same data structure, 
suppose an investigator wishes to estimate the effect of assigning treatment according to a deterministic rule at allocation times (depending only on the strict history), while preventing censoring. In this case, the intervention index set is $J = \{a, c\}$, where the allocation intervention $\mathfrak n^a$ is as in \eqref{eq: allocation assignment} and $\mathfrak{n}^c = 0$.

Define $\tau^a, \tau^c$ as in \eqref{eq: tau j deviation time}, and set $\tau^J = \tau^a \wedge \tau^c$. We assume $\mathbb \Lambda^J$, the compensator of $I(\tau^J \leq \cdot)$, satisfies the positivity condition \eqref{eq: J int positivity}, so that $\mathbb K^J$ in \eqref{eq: mbbk J well defined} is well-defined. Under the conditions of Theorem \ref{theorem: multiple marked interventions}, we then get
$$
E_P[\tilde{Y}_t] = E_P[\mathcal{E}(\mathbb{K}^J)_t Y_t] = E_Q[Y_t],
$$
where the compensating measure $\mathcal L$ figuring in the $g$-formula \eqref{eq: thm 3 g formula} is
\begin{align*}
    \mathcal{L}(dt \times dx) &= \delta_v(dx)\Lambda^v(dt) + I(x\in X^a) \pi^*(L,N,t, dx)\bar N^a(dt) \\
    &\quad + \delta_c(dx) \cdot 0 + I(x \in X^\ell) \Lambda^\ell(dt \times dx) + \delta_d(dx)\Lambda^d(dt).
\end{align*}
Write $\{ T_k^a \}_{k \geq 1}$ for the ordered jump times of $\bar N^a = N^a((0,\cdot] \times X^a)$, i.e., the observed treatment allocation times. If there is no censoring, the exchangeability condition Theorem \ref{theorem: multiple marked interventions} \ref{enum: mpp exchangeability} is equivalent to \footnote{\eqref{eq: delayed allocation seq ex} is also equivalent to $I(\tau^a > T_k^a) \indep \tilde{Y} | \mathcal{F}_{T_k^a-},$ for each $k \geq 1$ (see Appendix \ref{appendix: example: delayed allocation}).} 
\begin{align}
    I(\tau^a > T_k^a) \indep \tilde{Y} |\mathcal{F}_{T_k^a-}, \tau^a \geq T_k^a \quad \text{for each } k \geq 1,  \label{eq: delayed allocation seq ex}
\end{align}
see Appendix \ref{appendix: example: delayed allocation}. 
\end{example}



\subsubsection*{On preventing censoring}
Examples \ref{example: intervention on treatment and schedule} and \ref{example: delayed allocation} follow central discrete-time causal inference works \cite{robins1986parametric_g, Hernanrobins2021causal}, where right-censoring is addressed through a hypothetical intervention that prevents it. Identification results for this intervention in the MPP setting are covered as a special case of our theorems. The relationship between this identification strategy and commonly studied identification strategies in the classical survival literature will be explored in subsequent work.

\subsection{Comparison with RGvdL's examples and identification claims}
\label{subsubsection: evaluating identification claims by RGvdL}
As noted in Remark \ref{remark: rgvdl} in Section \ref{section: previous work on continuous-time causal inference}, RGvdL develop estimation methods under the assumption that their formulas follow from existing results in GR01. On the other hand, GR01 state that extending their discrete-time theory to counting process data is an open problem. In this section, we assess some of RGvdL's statements in light of Theorem \ref{theorem: multiple marked interventions} and our derived Examples \ref{example: intervention on treatment and schedule}-\ref{example: delayed allocation}. We also draw connections to the broader literature on causal inference with counting process data that studies distributions satisfying the invariance property (IP) (Definition \ref{def: invariance property} in Section \ref{subsection: the compensator under the g-formula distribution}).

RGvdL's observed data structure and the data structure in Examples \ref{example: intervention on treatment and schedule}-\ref{example: delayed allocation} differ in the underlying decision-making narratives. RGvdL perhaps implicitly adopt the perspective that raw time stamps directly represent the true causal ordering. Specifically, they define "Let $T^a_1 < T^a_2 < \dots$ be the random times at which the treatment regime ... may change" \cite[p. 2473]{Rytgaard2022}, where "at visit $T_k^a$, the doctor considers ... to decide to continue or to change the treatment" \cite[p. 2471]{Rytgaard2022}. However, in reality, a doctor cannot simultaneously review information, decide on treatment, and implement an action. As argued in Section \ref{subsection: interventions}, a real-life decision involves a delay between information gathering and action.

A more plausible assumption is that these events are temporally separated: after the patient arrives, the doctor reviews information, \textit{before} deciding on treatment. This motivates the data structure in Examples \ref{example: intervention on treatment and schedule}-\ref{example: delayed allocation} (building on the clinical scenario from Example \ref{example: deterministic dynamic predictable refinement}), which distinguish between treatment visits ($N^v$) and treatment allocations ($N^a$) through the delay assumption \eqref{eq: delay assumption}. These examples align with the broader theme of this article (see also Section \ref{subsection: interventions} and Examples \ref{example: atomic}-\ref{example: deterministic dynamic predictable refinement}) and of seminal work in discrete time \cite{robins1986parametric_g,robins2011alternative,richardson_single_2013}: making decision-making narratives explicit. 

Nonetheless, RGvdL's observed data structure is, technically, a special case of our data structure if we set $\delta = 0$ and redefine our MPP appropriately: specifically, we can align the data structure in Section \ref{subsec: illustrative examples} with theirs by setting $\delta = 0$ and redefining the visit counting process $N^v$ and the allocation MPP $N^a$ to be a single MPP $N^a(dt \times dx)$ that counts treatment visits with allocations occurring at the same time.

\subsubsection*{Comparison to RGvdL's "intervention on treatment and schedule."}
Example \ref{example: intervention on treatment and schedule} is conceptually similar to RGvdL's "intervention on treatment and schedule," \cite[Definition 2]{Rytgaard2022}, with $\delta = 0$, i.e., with instantaneous treatment allocations at treatment visit times. This $\delta = 0$ case is technically covered by Theorem \ref{theorem: multiple marked interventions}, \footnote{Specifically, when $\delta = 0$, using that $N^a(dt \times dx)$ is the MPP counting treatment visits and treatment allocations with no delay, consider $\mathfrak n^a(l, \varphi, dt \times dx) = \pi^*(l,\varphi, t, dx)\kappa^*(l,\varphi, dt)$, with $\pi^*$ and $\kappa^*$ as in Example \ref{example: intervention on treatment and schedule}. The assignment function $\mathfrak n^a$ satisfies $\na{}(l, \varphi, dt \times dx) = \na{}(l, \varphi|_{t-}, dt \times dx)$, and Theorem \ref{theorem: multiple marked interventions} therefore applies.} and we can thus directly compare formulas and assumptions. 

Both our $g$-formula distribution and the distributions studied by RGvdL satisfy the invariance property (IP), which also appears in many related causal inference works. However, there seems to be a distinction between our and RGvdL's targets of inference: we define interventions on schedule through $\kappa^*$, where $\kappa^*(l,\varphi,(0,\cdot])$ is a counting process trajectory that explicitly specifies visit times based on past events, while RGvdL replace the observed treatment compensator with 'some choice' $\Lambda^{a,*}$.

As emphasized in Remark \ref{remark: identifying vs invariance distributions} in Section \ref{subsection: the compensator under the g-formula distribution}, there are many distributions that satisfy (IP) without corresponding to explicit (deterministic) interventions. Thus, without specific details, RGvdL's approach appears similar to RRND and many related works on causal inference with counting processes that study distributions obtained by modifying the treatment compensator; see, e.g., \cite{roysland2011,roysland2012counterfactual,ryalen2018pcancer,ryalen2019additive,roysland2022graphical}. These works study distributions satisfying (IP), but they do not generally specify explicit interventions.

RGvdL differ from these related works in that they claim their formulas follow from existing identification results in GR01; specifically, for identifying their "intervention on treatment and schedule" estimand, RGvdL assert without further justification that a specific sequential independence condition is sufficient to derive their formula \cite[p. 2477, stated in natural language]{Rytgaard2022}.

In Appendix \ref{appendix: construction of potential outcomes}, we develop sequential independence assumptions which verifies the formulas in Example \ref{example: intervention on treatment and schedule} under the regularity conditions \eqref{eq: can comp j h}-\eqref{eq: interv interv} --- see  \eqref{eq: Tkj cond indep}-\eqref{eq: Vkj Tkj cond indep} in Appendix \ref{appendix: construction of potential outcomes}. The argument showing that these conditions are sufficient relies on the presence of latent event times which are not present in RGvdL's asserted sufficient condition. In Proposition \ref{proposition: sequential exchangeability predictable times} in Appendix \ref{appendix: seq exchangeability proof sketch}, we sketch an alternative argument that would verify the formulas in Example \ref{example: intervention on treatment and schedule}, under an alternative sequential independence condition. However, this argument relies on the implausible additional assumption that all observed event times are predictable --- a restriction that is unlikely to hold in real-world applications. 

\subsubsection*{Comparison to RGvdL's "intervention on treatment assigned."}
RGvdL's "intervention on treatment assigned," \cite[Definition 1]{Rytgaard2022}, corresponds to Example \ref{example: delayed allocation} with $\delta = 0$. RGvdL claim that a sequential independence condition resembling \eqref{eq: delayed allocation seq ex} is sufficient to identify their proposed estimand; see \cite[Assumption 2]{Rytgaard2022}.

Our derived exchangeability condition \eqref{eq: delayed allocation seq ex} hold by design if the treatment were sequentially randomized at assignment times. Yet when $\delta = 0$, $\mathfrak n^a$ in \eqref{eq: allocation assignment} is optional, i.e. $\mathfrak{n}^a(l, \varphi, dt \times dx) =\mathfrak{n}^a(l, \varphi|_{t}, dt \times dx)$, but it is not predictable, and Theorem \ref{theorem: multiple marked interventions} does not apply in general. Consequently, it is unclear whether the sequential randomization condition \eqref{eq: delayed allocation seq ex} is sufficient for identifying their estimand. Subsequent work \cite{Johan2026Identification}, building on an earlier draft of this manuscript, will describe independence conditions that validate one of RGvdL's proposed formulas.

\subsection{Interpretation and broader implications}
\label{subsection: interpretation and broader implications}
Because our identifying conditions and the associated likelihood ratio process are derived from explicit intervention definitions, the interpretation of our estimands is transparent. In particular, investigators applying our methods must first specify their particular intervention of interest, for example, as briefly sketched in our Examples \ref{example: atomic}-\ref{example: delayed allocation}. This clarity is valuable in the MPP setting. The path and distribution spaces are comparatively rich and, as noted in Remark \ref{remark: identifying vs invariance distributions}, most MPP distributions do not correspond to explicit interventions.

In Example \ref{example: delayed allocation}, treatment was assigned at allocation times according to a rule that depends on strictly earlier events. The exchangeability condition \eqref{eq: delayed allocation seq ex} matches single-world conditional independencies that would hold by design if treatment were sequentially randomized at each allocation time. Our identifying conditions can, in principle, be tested in a future experiment. (See \cite{robins1997complex} for a similar formulation in the sequential treatment setting.)

\section{Estimation}
\label{section: estimation}
The identifying likelihood ratio process in Theorem \ref{theorem: multiple marked interventions} has a generic form specified by an optional time $\tau^J$ and its compensator $\mathbb \Lambda^J$. This generic form motivates customized estimation techniques. Developing robust and efficient estimators under our conditions is an ongoing research project which will be reported in separate work.

In this section, we give a high-level sketch of an alternative estimation strategy motivated by the IPW formula in Theorem \ref{theorem: multiple marked interventions}: \footnote{This representation of $\E(\mathbb K^J)$ can be derived using the techniques from the proof of Proposition \ref{prop: characterizations of W}.}
\begin{align*}
    \E(\mathbb K^{J}) = \frac{I(\tau^{J} > \cdot)}{\prodi_{0< s \leq \cdot}(1 - d\mathbb \Lambda_s^{J})}.
\end{align*}

For concreteness, we study estimators for $E_P[\tilde Y_T]$ where $\tilde Y_T$ is survival at time $T$ under the interventions and identification assumptions in Example \ref{example: delayed allocation}. Recall from that example that $\tau^J = \tau^a \wedge \tau^c$, where $\tau^a$ is the deviation time for the regime on treatment allocation (defined as in \eqref{eq: tau j deviation time}), and $\tau^c$ is the censoring time. In this case, we have that $\mathbb \Lambda^{J} = ~^{\tau^{J}}\mathbb \Lambda^a + ~^{\tau^{J}}\Lambda^c$, where $\mathbb \Lambda^a$ is the compensator of the treatment-deviation counting process $\mathbb N^a = I(\tau^a \leq \cdot)$ and $\Lambda^c$ is the compensator of the censoring counting process $N^c = I(\tau^c \leq \cdot)$. We assume a model for $\mathbb \Lambda^{J}$, e.g. informed by subject matter knowledge of what makes subjects deviate from the regime. In practice, we would have to work with a specific choice of either 1) $\mathbb \Lambda^{J} = \int_0^{\tau^{J} \wedge \cdot} d\eta_s(L, N|_{s-})$, where $\eta$ is a model for the cumulative hazard of $\tau^{J}$ given covariates $L$ and the strict history of $N$, or 2) $\mathbb \Lambda^{J} = \int_0^{\tau^{J} \wedge \cdot} d\eta_s^a(L, N|_{s-}) + \int_0^{\tau^{J} \wedge \cdot} d\eta_s^c(L, N|_{s-})$ where $\eta^a$ (resp. $\eta^c$) is the cumulative hazard for $\tau^a$ (resp. $\tau^c$) given covariates $L$ and the strict history of $N$. We omit a detailed discussion of such models and associated statistical inference here for brevity, and refer instead to \cite{Kosorok2008} for a thorough treatment which is relevant to our strategy. Having fitted the models in 1) or 2) based on an i.i.d. population of size $n$, we can construct a weight estimator using individual $i$'s covariate history to predict that individual's weight:
\begin{align*}
    \hat V^{i} = \frac{I(\tau^{J,i} > \cdot)}{\prodi_{0 < s \leq \cdot} \big(1 - d \hat {\mathbb \Lambda}_s^{J,i}\big)},
\end{align*}
where $\tau^{J,i}$ is the minimum of $\tau^{a}$ and $\tau^{c}$ for individual $i$, and $\hat{\mathbb \Lambda}^{J,i}$ is an estimate (prediction) of $\mathbb \Lambda^{J}$ under individual $i$'s history. We can then consider the simple plug-in estimator
\begin{align}
    \frac{1}{n} \sum_{i=1}^n \hat V_T^{i} Y_T^i, \label{eq: simple ipw}
\end{align}
which resembles IPW estimators that have been described in the discrete literature \cite{hernan2000marginal}. In addition to \eqref{eq: simple ipw}, one could consider closely related estimators such as Hajek estimators, e.g. with stabilized weights \cite[Technical Point 12.2]{Hernanrobins2021causal}, but we do not explore such alternatives here. 

Assuming standard integrability conditions and that the trajectories of $\mathbb \Lambda^{J}$ and the elements of $\{ \hat {\mathbb \Lambda}^{J,i} \}$ take values in a Donsker class, the estimator is consistent for $E_P[\E(\mathbb K^{J})_T Y_T]$ and asymptotically normal provided that the models are correctly specified. To formally establish the result, one can use empirical process results ensuring weak convergence of root-$n$ residuals of empirical means involving random functions, see e.g. \cite[Lemma 19.24]{vandervaart1998asymptotic}, the Hadamard differentiability of the product-integral, and the functional delta method. 

Another alternative is to construct parametric $g$-formula estimators of \eqref{eq: thm 3 g formula} under assumptions on hazard models for the compensator (e.g., proportional or additive hazards \cite{aalen2008survivalandevent,Andersen}), and then use Monte Carlo simulation based on the fitted hazard models to estimate the integral.

The above simple estimation strategies offer valid inference for all data-generating laws where the model for $\mathbb{\Lambda}^{J}$ is correctly specified and the Donsker and integrability conditions are met.


\section{Discussion} \label{section: discussion}
This work provides several conceptual and methodological contributions. First, we formalize interventions in MPP settings as explicit actions $\mathfrak{n}^j$ that fix treatment trajectories, which is of interest to decision-makers who want to assess the consequences of actions. Second, we characterize potential outcomes arising from such interventions, thus enabling formal reasoning about data-generating mechanisms. We then operationalize MPP analogues of canonical identifying conditions (consistency, exchangeability, and positivity) and formally derive and characterize corresponding identification formulas. The simple identifying likelihood-ratio process $\mathcal{E}(\mathbb{K}^J)$, which we have not seen emphasized in related works, is derived directly from the intervention rules and transparently connects interventions and identifying conditions to $g$‑formulas. 

Next, by operationalizing identifying conditions that reduce to established discrete-time conditions under appropriate restrictions on compensators, we bridge the gap between discrete-time causal inference and the counting process literature for survival and event history analysis.  

While the original aim of the article was to develop nonparametric identification conditions for effects of dynamic regimes for MPP data structures, we discovered conditions that apply much more broadly to data described by stochastic processes. Our martingale-centered approach furthermore circumvents the complications associated with conditioning on events with probability zero, a problem that is commonly encountered in other approaches.

Future work involves developing analogues to results found in related missing data literature, such as the projection formula \cite{robins1992recovery,vaart2004onrobinsformula}, to develop efficient and robust estimators. Additional research directions include formalizing natural value processes and studying interventions that depend on these processes. Another intriguing area is to provide a graphical representation of our exchangeability condition. The irrelevance relation is similar to that of local independence \cite{didelez2008local,roysland2022graphical}, but differs from existing approaches in that the irrelevance concerns a random element recorded at baseline rather than a stochastic process which unfolds over time. Research in this direction likely requires development of new graphical identification criteria.

\section*{Acknowledgements}
We are grateful to Johan Sebastian Ohlendorff for feedback on an earlier draft of this manuscript. His comments revealed an oversight in the scope of our identification results. This motivated our development of the explicit construction and regularity conditions in Appendix \ref{appendix: construction of potential outcomes}, and led us to re-examine and refine our intervention definition. The resulting revisions sharpened the presentation. We also thank Niklas Nyboe Maltzahn for many stimulating discussions during the preparation of this work, and for feedback on earlier drafts of this manuscript.

 \section*{Funding}
P.C.R. and K.R. were supported by the Research Council of Norway grant "315323 STEINF - NFR  Stochastic differential equations for robust evaluation of cancer treatments." M.J.S. was supported by the Swiss National Science Foundation (project funding, grant number: 207436).

\bibliography{References}

\begin{thebibliography}{10}

\bibitem{aalen2008survivalandevent}
Odd Aalen, {\O}rnulf Borgan, and Hakon Gjessing.
\newblock {\em Survival and Event History Analysis: A Process Point of View
  (Statistics for Biology and Health)}.
\newblock Springer, 2008.

\bibitem{Andersen}
Per Andersen, Ørnulf Borgan, Richard Gill, and Niels Keiding.
\newblock {\em Statistical Models Based on Counting Processes}.
\newblock Springer Series in Statistics. Springer-Verlag, New York, 1993.

\bibitem{Andrews2020}
Ryan~M. Andrews and Vanessa Didelez.
\newblock Insights into the cross-world independence assumption of causal
  mediation analysis.
\newblock {\em Epidemiology}, 32(2):209–219, December 2020.

\bibitem{Bremaud1981point}
Pierre Brémaud.
\newblock {\em Point Processes and Queues}.
\newblock Springer Series in Statistics. Springer, New York, NY, 1981 edition,
  September 1981.

\bibitem{cohen2015stochastic}
Samuel~N Cohen and Robert~J Elliott.
\newblock {\em Stochastic Calculus and Applications}.
\newblock Probability and Its Applications. Springer, New York, NY, 2 edition,
  November 2015.

\bibitem{commenges2009dynamical}
Daniel Commenges and Anne Gégout-Petit.
\newblock A general dynamical statistical model with causal interpretation.
\newblock {\em Journal of the Royal Statistical Society: Series B (Statistical
  Methodology)}, 71(3):719--736, 2009.

\bibitem{CookLawless2007}
Richard~J. Cook and Jerald~F. Lawless.
\newblock {\em The Statistical Analysis of Recurrent Events}.
\newblock Springer New York, 2007.

\bibitem{Dawid2021decision}
A.~Philip Dawid.
\newblock Decision-theoretic foundations for statistical causality.
\newblock {\em Journal of Causal Inference}, 9(1):39--77, January 2021.

\bibitem{Dawid2010identifying}
A.~Philip Dawid and Vanessa Didelez.
\newblock Identifying the consequences of dynamic treatment strategies: A
  decision-theoretic overview.
\newblock {\em Statistics Surveys}, 4(none), January 2010.

\bibitem{didelez2008local}
Vanessa Didelez.
\newblock Graphical models for marked point processes based on local
  independence.
\newblock {\em Journal of the Royal Statistical Society: Series B (Statistical
  Methodology)}, 70(2):245--264, February 2008.

\bibitem{FlemingHarrington2005}
Thomas~R. Fleming and David~P. Harrington.
\newblock {\em Counting Processes and Survival Analysis}.
\newblock John Wiley {\&} Sons, Inc., September 2005.

\bibitem{Gill1994lecturenotes}
Richard~D. Gill.
\newblock {\em Lectures on survival analysis}, page 115–241.
\newblock Springer Berlin Heidelberg, 1994.

\bibitem{gill1990survey}
Richard~D. Gill and Soren Johansen.
\newblock A survey of product-integration with a view toward application in
  survival analysis.
\newblock {\em The Annals of Statistics}, 18(4):1129--1170, December 1990.

\bibitem{gill2001complex}
Richard~D. Gill and James~M. Robins.
\newblock Causal inference for complex longitudinal data: The continuous case.
\newblock {\em The Annals of Statistics}, 29(6), December 2001.

\bibitem{gill2004continuous}
Richard~D. Gill and James~M. Robins.
\newblock Causal inference for complex longitudinal data: The continuous time
  g-computation formula.
\newblock {\em arXiv preprint}, 2004.

\bibitem{He1992Semimartingale}
Sheng-Wu He, Jia-Gang Wang, and Jia-An Yan.
\newblock {\em Semimartingale theory and stochastic calculus}.
\newblock CRC Press, Boca Raton, FL, September 1992.

\bibitem{hernan2000marginal}
M.~A. Hernán, B.~Brumback, and J.~M. Robins.
\newblock Marginal structural models to estimate the causal effect of
  zidovudine on the survival of hiv-positive men, 2000.

\bibitem{Hernanrobins2021causal}
Miguel~A. Hernán and James~M. Robins.
\newblock {\em Causal Inference}.
\newblock CRC Press, Boca Raton, FL, February 2021.

\bibitem{Hu2019causal}
L.~Hu and J.~W. Hogan.
\newblock Causal comparative effectiveness analysis of dynamic continuous-time
  treatment initiation rules with sparsely measured outcomes and death.
\newblock {\em Biometrics}, 75(2):695--707, 2019.

\bibitem{jacobsen2006point}
Martin Jacobsen.
\newblock Point process theory and applications marked point and piecewise
  deterministic processes.
\newblock In {\em Probability and Its Applications}, pages 3--7.
  Birkhäuser-Verlag, Boston, 2006.

\bibitem{jacod1975}
Jean Jacod.
\newblock Multivariate point processes: Predictable projection, radon-nikodym
  derivatives, representation of martingales.
\newblock {\em Probability Theory and Related Fields}, 1975.

\bibitem{JacodShiryaev}
Jean Jacod and Albert~N. Shiryaev.
\newblock {\em Limit Theorems for Stochastic Processes}, volume 288 of {\em
  Grundlehren der Mathematischen Wissenschaften [Fundamental Principles of
  Mathematical Sciences]}.
\newblock Springer-Verlag, Berlin, second edition, 2003.

\bibitem{johnson2005semiparametric}
Brent~A. Johnson and Anastasios~A. Tsiatis.
\newblock Semiparametric inference in observational duration-response studies,
  with duration possibly right-censored.
\newblock {\em Biometrika}, 92(3):605--618, 2005.

\bibitem{Kallenberg2021foundations}
Olav Kallenberg.
\newblock {\em Foundations of Modern Probability}.
\newblock Springer International Publishing, 2021.

\bibitem{Kosorok2008}
Michael~R. Kosorok.
\newblock {\em Introduction to Empirical Processes and Semiparametric
  Inference}.
\newblock Springer New York, 2008.

\bibitem{LastBrandt1995marked}
Günter Last and Andreas Brandt.
\newblock {\em Marked point processes on the real line}.
\newblock Probability and Its Applications. Springer, New York, NY, 1995
  edition, August 1995.

\bibitem{Lok2001statistical}
Judith~J. Lok.
\newblock {\em Statistical Modelling of Causal Effects in Time}.
\newblock PhD thesis, Vrije Universiteit Amsterdam, 2001.
\newblock Naam instelling promotie: VU Vrije Universiteit Naam instelling
  onderzoek: VU Vrije Universiteit.

\bibitem{Lok2008statistical}
Judith~J. Lok.
\newblock Statistical modeling of causal effects in continuous time.
\newblock {\em The Annals of Statistics}, 36(3), June 2008.

\bibitem{Lok2004estimating}
Judith~J. Lok, Richard~D. Gill, Aad van~der Vaart, and James~M. Robins.
\newblock Estimating the causal effect of a time-varying treatment on
  time-to-event using structural nested failure time models.
\newblock {\em Statistica Neerlandica}, 58(3):271--295, August 2004.

\bibitem{nagel2017probability}
Werner Nagel and Rolf Steyer.
\newblock {\em Probability and Conditional Expectation: Fundamentals for the
  Empirical Sciences}.
\newblock Wiley, March 2017.

\bibitem{Johan2026Identification}
Johan~Sebastial Ohlendorff, Kjetil Røysland, Anders Munch, and Thomas Gerds.
\newblock Identification and estimation of causal effects under
  treatment-assigned-at-visit interventions in continuous time.
\newblock {\em In preparation}, page 823–843, August 2026.

\bibitem{pearl}
Judea Pearl.
\newblock {\em Causality: Models, Reasoning and Inference}.
\newblock Cambridge University Press, New York, NY, USA, 2nd edition, 2009.

\bibitem{protter}
P.~Protter.
\newblock {\em Stochastic Integration and Differential Equations}.
\newblock Springer, 2005.

\bibitem{richardson_single_2013}
Thomas~S. Richardson and James~M. Robins.
\newblock Single {World} {Intervention} {Graphs} {(SWIGs):} {A} {Unification}
  of the {Counterfactual} and {Graphical} {Approaches} to {Causality}.
\newblock Technical Report 128, Center for Statistics and the Social Sciences,
  University of Washington, 2013.
\newblock Available online at:
  https://csss.uw.edu/research/working-papers/single-world-intervention-graphs-swigs-unification-counterfactual-and.

\bibitem{Robins2000marginalvsstructural}
J.~M. Robins.
\newblock {\em Marginal Structural Models versus Structural Nested Models as
  Tools for Causal Inference}, pages 95--133.
\newblock Springer New York, 2000.

\bibitem{robins1992recovery}
J.~M. Robins and A.~Rotnitzky.
\newblock {\em Recovery of Information and Adjustment for Dependent Censoring
  Using Surrogate Markers}, pages 297--331.
\newblock Birkhäuser Boston, 1992.

\bibitem{robins1986parametric_g}
James Robins.
\newblock A new approach to causal inference in mortality studies with
  sustained exposure periods – application to control of the healthy worker
  survivor effect.
\newblock {\em Mathematical Modeling}, 7(9):1393--1512, 1986.

\bibitem{robins1987addendum}
James~M. Robins.
\newblock Addendum to “a new approach to causal inference in mortality
  studies with a sustained exposure period—application to control of the
  healthy worker survivor effect”.
\newblock {\em Computers \& Mathematics with Applications}, 14(9-12):923--945,
  1987.

\bibitem{robins1989analysis}
James~M. Robins.
\newblock The analysis of randomized and nonrandomized aids treatment trials
  using a new approach to causal inference in longitudinal studies.
\newblock In Lee Sechrest, Howard Freeman, and Albert Mulley, editors, {\em
  Health Service Research Methodology: A Focus on AIDS}, page 113–159. U.S.
  Public Health Service, National Center for Health Services Research,
  Washington, DC, 1989.

\bibitem{robins1997complex}
James~M. Robins.
\newblock Causal inference from complex longitudinal data.
\newblock In Maia Berkane, editor, {\em Latent Variable Modeling and
  Applications to Causality}, pages 69--117, New York, NY, 1997. Springer New
  York.

\bibitem{Robins2004optimal}
James~M. Robins.
\newblock Optimal structural nested models for optimal sequential decisions.
\newblock In {\em Proceedings of the Second Seattle Symposium in
  Biostatistics}, pages 189--326. Springer, 2004.

\bibitem{Robins2000correcting}
James~M. Robins and Dianne~M. Finkelstein.
\newblock Correcting for noncompliance and dependent censoring in an aids
  clinical trial with inverse probability of censoring weighted (ipcw)
  log‐rank tests.
\newblock {\em Biometrics}, 56(3):779–788, September 2000.

\bibitem{robins2004effects}
James~M Robins, Miguel~A Hernán, and Uwe Siebert.
\newblock Effects of multiple interventions.
\newblock {\em Comparative quantification of health risks: global and regional
  burden of disease attributable to selected major risk factors}, 1:2191--2230,
  2004.

\bibitem{robins2011alternative}
James~M. Robins and Thomas~S. Richardson.
\newblock Alternative graphical causal models and the identification of direct
  effects.
\newblock In {\em Causality and Psychopathology: Finding the Determinants of
  Disorders and their Cures}. Oxford University Press, February 2011.

\bibitem{ROSENBAUM1983propensity}
P.~R. Rosenbaum and D.~B. Rubin.
\newblock The central role of the propensity score in observational studies for
  causal effects.
\newblock {\em Biometrika}, 70(1):41--55, 1983.

\bibitem{rubin1974estimating}
D.~B. Rubin.
\newblock Estimating causal effects of treatments in randomized and
  nonrandomized studies.
\newblock {\em Journal of Educational Psychology}, 66(5):688--701, 1974.

\bibitem{ryalen2018pcancer}
Pål Ryalen, Mats Stensrud, Sophie Fosså, and Kjetil Røysland.
\newblock Causal inference in continuous time: An example on prostate cancer
  therapy.
\newblock {\em Biostatistics}, 19(4):600--613, 2018.

\bibitem{ryalen2019additive}
Pål~C. Ryalen, Mats~J. Stensrud, and Kjetil Røysland.
\newblock The additive hazard estimator is consistent for continuous-time
  marginal structural models.
\newblock {\em Lifetime Data Analysis}, 25(4):611--638, February 2019.

\bibitem{Rytgaard2022}
Helene~C. Rytgaard, Thomas~A. Gerds, and Mark~J. van~der Laan.
\newblock Continuous-time targeted minimum loss-based estimation of
  intervention-specific mean outcomes.
\newblock {\em The Annals of Statistics}, 50(5), October 2022.

\bibitem{roysland2011}
Kjetil Røysland.
\newblock A martingale approach to continuous-time marginal structural models.
\newblock {\em Bernoulli}, 2011.

\bibitem{roysland2012counterfactual}
Kjetil Røysland.
\newblock Counterfactual analyses with graphical models based on local
  independence.
\newblock {\em The Annals of Statistics}, 40(4):2162--2194, 2012.

\bibitem{roysland2022graphical}
Kjetil Røysland, Pål Ryalen, Mari Nygård, and Vanessa Didelez.
\newblock Graphical criteria for the identification of marginal causal effects
  in continuous-time survival and event-history analyses, 2022.

\bibitem{sarvet2025natural}
Aaron~L Sarvet and Mats~J Stensrud.
\newblock The natural value of treatment and its importance for causal
  inference.
\newblock {\em Annual Review of Statistics and Its Application}, 13, 2025.

\bibitem{Sattenrobins2001estimating}
Glen~A. Satten, Somnath Datta, and James Robins.
\newblock Estimating the marginal survival function in the presence of time
  dependent covariates.
\newblock {\em Statistics \& Probability Letters}, 54(4):397–403, October
  2001.

\bibitem{martinussen2006dynamic}
Thomas~H. Scheike and Torben Martinussen.
\newblock {\em Dynamic Regression Models for Survival Data}.
\newblock Springer, NY, 2006.

\bibitem{shiryaevkallsen2002cumulant}
Albert~N. Shiryaev and Jan Kallsen.
\newblock {The cumulant process and Esscher's change of measure}.
\newblock {\em Finance and Stochastics}, 6(4):397--428, 2002.

\bibitem{Sokol2015}
Alexander Sokol and Niels~Richard Hansen.
\newblock Exponential martingales and changes of measure for counting
  processes.
\newblock {\em Stochastic Analysis and Applications}, 33(5):823–843, August
  2015.

\bibitem{sun2022causal}
Jinghao Sun and Forrest~W. Crawford.
\newblock Causal identification for continuous-time stochastic processes, 2022.

\bibitem{vandervaart1998asymptotic}
A.~W. van~der Vaart.
\newblock {\em Asymptotic Statistics}.
\newblock Cambridge Series in Statistical and Probabilistic Mathematics.
  Cambridge University Press, 1998.

\bibitem{vaart2004onrobinsformula}
Aad~W. van~der Vaart.
\newblock On robins' formula.
\newblock {\em Statistics \& Decisions. International Mathematical Journal for
  Stochastic Methods and Models}, 22(3):171–200, 2004.

\bibitem{williams1991probability}
David Williams.
\newblock {\em Probability with Martingales}.
\newblock Cambridge University Press, February 1991.

\bibitem{Yang2021}
Shu Yang.
\newblock Semiparametric estimation of structural nested mean models with
  irregularly spaced longitudinal observations.
\newblock {\em Biometrics}, 78(3):937–949, April 2021.

\bibitem{ying2024functionaldynamic}
Andrew Ying.
\newblock Causality for complex continuous-time functional longitudinal studies
  with dynamic treatment regimes, 2024.

\bibitem{ying2024functional}
Andrew Ying.
\newblock Causality for functional longitudinal data.
\newblock In Francesco Locatello and Vanessa Didelez, editors, {\em Proceedings
  of the Third Conference on Causal Learning and Reasoning}, volume 236 of {\em
  Proceedings of Machine Learning Research}, page 665–687. PMLR, April 1–3
  2024.

\bibitem{Young2024story}
Jessica~G. Young.
\newblock Story-led causal inference.
\newblock {\em Epidemiology}, 35(3):289–294, April 2024.

\bibitem{Young2014natural}
Jessica~G Young, Miguel~A Hernán, and James~M Robins.
\newblock Identification, estimation and approximation of risk under
  interventions that depend on the natural value of treatment using
  observational data.
\newblock {\em Epidemiologic Methods}, 3(1):1--19, Jan 2014.

\bibitem{yu2002construction}
Zhengguo Yu and Mark~J. van~der Laan.
\newblock Construction of counterfactuals and the g-computation formula.
  department of biostatistics.
\newblock {\em University of California}, 2002.

\bibitem{Zhang2011}
Mingyuan Zhang, Marshall~M. Joffe, and Dylan~S. Small.
\newblock Causal inference for continuous-time processes when covariates are
  observed only at discrete times.
\newblock {\em The Annals of Statistics}, 39(1), February 2011.

\end{thebibliography}
\bibliographystyle{plain}

\appendix

\section*{Appendix Overview}
Appendix \ref{appendix: Basic definitions} reviews concepts from stochastic process theory and describes the canonical space of point process realizations. Appendix \ref{appendix: construction of potential outcomes} provides a joint construction of observed and potential outcome processes such that identifying conditions are satisfied. Appendix \ref{appendix: general MPP potential outcomes} characterizes potential outcome processes in the general MPP setting. Appendix \ref{appendix: supporting calculations} contains supporting lemmas and technical results. Proofs of results not found in the main text are in Appendix \ref{appendix: proofs}. An overview of notation is provided in Appendix \ref{appendix: notation}.

\section{Basic definitions and notation}
\label{appendix: Basic definitions}

\subsection{Basic definitions}
This definitions in this section can be found in \cite{JacodShiryaev,protter, LastBrandt1995marked}. It is assumed that the reader is familiar with $\sigma$-algebras, probability measures, probability spaces, etc. We refer to \cite{protter} for a treatment of stochastic integration with semimartingales.

A \textit{filtered probability space} is a probability space $(\Omega, \G, Q)$ equipped with a \textit{filtration} $\G_\T = \{\G_t\}_{t \in [0,T]}$ of increasing (in $t$) $\sigma$-algebras which are sub-$\sigma$-algebras of $\G$. The filtration is \textit{right-continuous} if $\G_t = \cap_{s > t} \G_s$, and \textit{complete} if it includes all null sets of $Q$; i.e. all subsets $A$ of $B$ whenever $B \in \G$ and $Q(B) = 0$. The filtration generated by an MPP is right-continuous (see e.g. \cite[Theorem 2.2.4]{LastBrandt1995marked}), but not complete.

For a stochastic process $Z$ and fixed $\omega \in \Omega$, the function $t \mapsto Z_t(\omega)$ is the \textit{sample path} or \textit{trajectory} of $Z$ corresponding to $\omega$. A process is \textit{cadlag} if almost all its paths are right-continuous with left limits. Two processes are \textit{indistinguishable} if almost all of their sample paths agree. The variation of $Z$ over $\T$ is defined to be
$$ V_\T(Z)(\omega) = \sup_{\pi \in \mathcal P_\T} \sum_{t_i \in \pi }| Z_{t_{i+1}}(\omega) - Z_{t_i}(\omega) |, $$
where $\mathcal P_\T$ are all finite partitions of $\T$. $Z$ is of \textit{finite variation} if almost all its paths are of finite variation, i.e. $V_\T(Z)(\omega) < \infty$ for almost all $\omega$. It is of integrable variation if $E_Q[V_\T(Z)] < \infty$.

A $\T \cup \{ \infty \}$-valued measurable map $\tau$ on $\Omega$ is a $\G_\T$-\textit{optional time} if $\{ \tau \leq t \} \in \G_t$ for $t \in [0,T]$. For a process $Z$, a property $p$ holds \textit{locally} if there is a sequence of optional times $\{ \sigma_n \}_n$ increasing up to $\infty$ $Q$-a.s. such that $~^{\sigma_n}Z I(\sigma_n > 0)$ has the property $p$ for each $n$. For an optional time $\tau$, the \textit{stopped $\sigma$-algebra at $\tau$}, $\G_\tau = \{ A \in \G_T | \{ \tau \leq t \} \cap A \in \G_t \forall t \in \T \}$, is a further $\sigma$-algebra on $\Omega$.

A process $Z$ is $\G_\T$-\textit{adapted} if $Z_t$ is $\G_t$-measurable for each $t \in [0,T]$. The smallest filtration which makes $Z$ adapted is called the \textit{natural filtration} of $Z$. A process is $\G_\T$-\textit{predictable} if it, as a map from $\T \times \Omega$ to $ \mathbb{R}$, is measurable with respect to the predictable $\sigma$-algebra $\mathcal P$ on $\T \times \Omega$. This $\sigma$-algebra is generated, for instance, by the $\G_\T$-adapted processes whose paths are continuous (equivalently left-continuous, with the left limit at $t = 0$ defined as the value of the path at $t = 0$) functions of $t$. A process is $\G_\T$-optional if it is measurable with respect to the optional $\sigma$-algebra. This is the $\sigma$-algebra on $\T \times \Omega$ which is generated by the $\G_\T$-adapted and cadlag processes. When working with natural filtrations, which are not completed, predictability and optionality are not  properties that depend on the probability measure.

A $\G_\T$-adapted and cadlag stochastic process $Z$ is a \textit{martingale} (resp. \textit{submartingale}) with respect to $\G_\T$ and $Q$ if $Z_t$ is integrable under $Q$ and $E_Q[Z_t | \G_s] = Z_s$ $Q$-a.s. (resp. $\geq Z_s$ $Q$-a.s.) for each $s\leq t$ and each $t$. It is square integrable if also $ \sup_{t \in T} E_Q[Z_t^2] < \infty.$ 
$Z$ is a \textit{semimartingale} if there exist processes $M$ and $V$ where $M$ is a local martingale and $V$ is a finite variation process such that $Z = Z_0 + M + V$. If the compensator of an adapted counting process has finite variation, then the difference between the compensator and the counting process is a local martingale. It is evidently also a semimartingale. 

For a semimartingale $Z$ with $Z_0=0$ we denote by $\E(Z)$ the stochastic exponential, also called the Doléans-Dade exponential, of $Z$. This is the unique solution to the SDE 
\begin{align*}
    U_t = 1 + \int_0^t U_{s-}dZ_s,
\end{align*}
where we adopt the assumption that $Z_{t}=0$ when $Z$ is a process and $t < 0$. For semimartingales $Z,U$, the \textit{quadratic covariation process} $[Z,U]$ is defined by
\begin{align*}
    [Z,U]_t := Z_t U_t - \int_0^t Z_{s-} dU_s - \int_0^t U_{s-}dZ_s,
\end{align*}
We use the notational convention $[Z,Z]=[Z]$. The predictable variation process $\langle Z,U \rangle^P$ is, when it exists, the compensator of $[Z, U]$, and it is then the unique predictable process which makes $[Z,U] - \langle Z,U \rangle^P$ a local martingale. 

Consider another measurable space $(E,\mathcal E)$ which we assume is not too "large." \footnote{ e.g. a Blackwell space; see  \cite[p. 65]{JacodShiryaev}.} A \textit{random measure} on $[0,T] \times E$ is a family $\mu  = \{ \mu (\omega; dt \times dx): \omega \in \Omega \}$
of non-negative measures on $([0,T]  \times E, \mathcal B_{[0,T]} \otimes \mathcal E)$ satisfying $\mu(\omega, \{0 \} \times E) = 0$ identically. Let $\tilde  {\mathcal P} = \mathcal P \otimes \mathcal E$ and $\tilde \Omega = \Omega \times [0,T] \times E$. A function on $\tilde \Omega$ that is $\tilde {\mathcal P}$-measurable is called \textit{predictable}. A random measure $\mu$ is  \textit{predictable} if 
$$ \int_{[0,\cdot] \times E} H(\omega, s, x)\mu(\omega, ds \times dx) $$
is predictable for every predictable function $H$ for which the integral exists. Optionality is analogously defined.

Let $(S, \mathcal S)$ and $(Y, \mathcal Y)$ be measurable spaces. A \textit{kernel} from $(S, \mathcal S)$ to $(Y, \mathcal Y)$ (abbreviated: from $S$ to $Y$) is a mapping $K$ from $S \times \mathcal{Y}$ to $[0,\infty]$ such that $K(\cdot, A)$ is measurable for each $A \in \mathcal{Y}$ and $K(x,\cdot)$ is a measure for each $x \in S$.
  
Recall  from Section \ref{section: set-up and notation} that $\F_\T$ is the natural filtration of $N$.  Since $\F_\T$ is right-continuous, we get by \cite[I Theorem 3]{protter} that for any $\F_\T$-adapted and cadlag process $Z$ and open $B \subset \mathbb R$, the real-valued function
\begin{align*}
    \omega \mapsto \inf\{ s>0 : Z_s(\omega) \in B \}
\end{align*}
on $\Omega$ is an $\F_\T$-optional time, a result we repeatedly invoke without mention in this text.

\subsection{Conditional independence and expectation with respect to a $\sigma$-algebra mixed with an event}
\label{subsection: conditional independence and expectation}

For a probability space $(\Omega,\F,P)$ and $A \in \F$, the assignment $P_{A}(\cdot) = P(A \cap \cdot)/P(A)$ is a probability measure on $(\Omega, \F)$ provided that $P(A) > 0$.  
In that case we make the identification 
\begin{align}
    Z \indep_P Y | \G,A \quad \Leftrightarrow \quad Z \indep_{P_{A}} Y | \G, \label{eq: conditional independence event probability measure}
\end{align}
for random variables $Z,Y$ and sub-$\sigma$-algebra $\G$ of $\F$, where we often omit the $P$ subscript in the left independence statement. That is, we understand conditional independence with respect to $\G,A$ under $P$ as conditional independence with respect to $\G$ under $P_A$.

If $Z$ is integrable, we define 
\begin{align*}
    E_P[Z | \G, A] := \frac{E_P[Z I_A | \G]}{E_P[I_A|\G]},
\end{align*}
with the usual convention $"\frac{0}{0}=0"$. As it happens, we are in Appendix \ref{appendix: seq exchangeability proof sketch} interested in logical implications of the above type of conditional independence on partial conditional expectations on the form "$E_P[ZY|\G,A]$", where the variable $Y$ has support in $A$. A result of this is that, if $A$ is a $P$-null set, conditional expectations of this type are degenerate. 

In the calculations in Appendix \ref{appendix: seq exchangeability proof sketch} we will study conditional independences of the type \eqref{eq: conditional independence event probability measure} for sets of the form $A = \{ \Nastop > T_n \}$, and following the above we can assume without loss of generality that $P(A) > 0$, and that $P_{A}(\cdot) = P(A \cap \cdot) / P(A)$ is a probability measure on $(\Omega, \F)$.

\subsection{The canonical space of point process realizations} 
\label{appendix: the canonical space of point process realizations}

We closely follow the presentation in \cite{LastBrandt1995marked} here. Readers who are interested in a more detailed exposition can consult Section 2.2. therein. 

We consider a mark space $(X,\X)$ which is a Borel space, and write $\bar X = X \cup \{ \nabla \}$, where $\nabla$ is the 'irrelevant mark.' Denote by $\N^X$ the space of all double sequences $(t_n,x_n)_{n} \in ( \overline{\mathbb R}_+  \times \overline X )^{\mathbb N}$ which satisfies
\begin{align*}
    0 &< t_1 \leq t_2 \leq \dots \\
    t_n &< t_{n+1} \text{ and } x_n \in X \text{ if } t_n < \infty \\
    t_n &= t_{n+1} \text{ and } x_n = x_{n+1} = \nabla \text{ if } t_n = \infty.
\end{align*}
Following foundational literature on MPPs on the real line \cite{jacobsen2006point, LastBrandt1995marked} there is a one-to-one correspondence between each element $\varphi = (t_n,x_n)_{n} \in \N^X$ and a corresponding counting measure
 $$ \varphi = \sum_{n \geq 1} I(t_n < \infty) \delta_{(t_n, x_n)}$$
 on $(\Rplus\times X, \mathcal B(\Rplus) \otimes \X $ ). If the mark space is $(\I_p, 2^{\I_p})$, it is also represented as a $p$-dimensional counting process,
 \begin{align*}
     \varphi = (\varphi^1, \dots, \varphi^p),
 \end{align*}
 where $\varphi^i$ is recovered from the counting measure and double sequence by $\varphi^i = \varphi((0,\cdot] \times \{ i \}) = \sum_k I(t_k \leq \cdot, x_k = i)$.
 
 On $\N^X$, we let $\H^X$ be the smallest $\sigma$-algebra which makes the projection mappings
\begin{align}
    (t_n,x_n)_{n} &\mapsto t_j, \label{eq: proj 1} \\
    (t_n,x_n)_{n} &\mapsto x_j,  \label{eq: proj 2}
\end{align}
 measurable for each $j$.  Denote by $\N^X_T$ the subset of $\N^X$ defined by the restriction of each element of $\N^X$ to $\T \times X$, where the restriction is defined by the map $( t_n,x_n )_{n } \mapsto ( t_n^\star, x_n^\star)_{n}$, where 
\begin{align*}
    (t_n^\star, x_n^\star) = \begin{cases}
   (t_n, x_n) \text{ if }t_n \leq T \\
   (\infty, \nabla) \text{ otherwise.}
\end{cases}
\end{align*}
This restriction is also defined by the map $\varphi \mapsto \varphi|_T$. Denote by $\H_T^X$ the smallest $\sigma$-algebra which makes the restriction of the projection mappings \eqref{eq: proj 1}-\eqref{eq: proj 2} measurable. 

Following \cite[Remark 2.2.5]{LastBrandt1995marked}, we define the canonical space of point processes on the real line and on the bounded interval $\T$. We allow for a random element $L$ taking values in a measurable space $(S,\S)$ to be included, where, in practice, $L$ signifies variables that are measured at or before time $t = 0$. 
 \begin{definition}[Canonical setting]\label{definition: canonical setting}
    We say that we are in the \emph{canonical setting} if
    \begin{enumerate}[label=\textnormal{(\roman*)}]
        \item \label{enum: canonical space} $(\Omega, \F) = (\N^X, \H^X)$, resp. $(\Omega, \F) = (S\times \N^X, \S \otimes \H^X)$.
        \item \label{enum: N is the identity map} $N=Id$, resp. $(L,N) = Id$.
    \end{enumerate}
\end{definition}
 The canonical setting for realizations restricted to $\T$ is similarly defined. Notationally, we add $T$ subscripts in the relevant positions. We write $\Nd{} := \N^{\I_d}$ and $\Hd{} := \H^{\I_d}$ (resp. $\Nd{T}, := \N_T^{\I_d}$ and $\Hd{T} := \H_T^{\I_d}$) for the case of a point process with marks in $\I_d$ (on $\T$), which can be identified with a $d$-dimensional counting process (on $\T$).

\subsubsection{Canonical compensators.}
 \label{appendix: canonical compensators and optional times}
Several objects that are defined on a filtered measurable space $(\Omega, \F_\T^N, \A)$ for an MPP $N$ with mark space $X$ and $\F_t^N = \sigma(N|_t) \vee \sigma(L)$ can be described in terms of objects defined on the canonical space $(S \times \N^X_T, \S \otimes \H^X_T)$, where $L$ takes values in $(S,\S)$.  
In particular, there are one-to-one correspondences between several measure-theoretic concepts defined on $(\Omega, \F_\T^N, \A, P)$ with so-called "canonical" concepts on the canonical space. We highlight one such result which is relevant for the current text (see e.g. 
Theorem 4.2.2 in \cite{LastBrandt1995marked} for details and proof):

\begin{result}[Canonical compensator] \label{result: canonical compensator}
Suppose that $\Gamma(dt \times dx)$ is an $\F_{\T}^N$- compensating measure of $N$ under $P$. There is a unique  random measure $\alpha$, called \emph{the canonical compensator} of $N$ with respect to $P$ and $\F_\T^N$, which satisfies
\begin{align*}
    \Gamma(dt \times dx) = \alpha(L,N, dt \times dx) \quad P\text{-a.s.}
\end{align*}
$\alpha$ is a kernel from $S \times \N_T^X$ to $\T \times X$ which can be chosen such that, for each  $(l,\varphi) \in S \times \N_T^X$ and $t \in \T$,
\begin{enumerate}
\item \label{item: canonical 1} $\alpha(l,\varphi, \{ 0 \}\times X)=0$
\item \label{item: canonical 2} $\alpha(l,\varphi, \{ t \} \times X )\leq 1$
\item \label{item: canonical 3} $\alpha\big(l,\varphi , [\pi'_\infty(\varphi), \infty ) \times X \big)=0$,
\end{enumerate}
where we have defined $\pi'_\infty(\varphi) = \inf \big\{ t > 0 \big| \alpha\big(l, \varphi, (0,t] \times X \big) = \infty \big\}$. $\alpha$ can also be chosen to be predictable on the canonical space.
\end{result}

\section{A joint construction of observed and potential outcomes such that identifying conditions are satisfied}
\label{appendix: construction of potential outcomes}

In this appendix, we construct observed and potential outcomes so that our exchangeability and consistency conditions are satisfied.

We work on the time domain $\mathbb R_+$ instead of the bounded interval $\T$ used in the main text; all results in this appendix transfer to the bounded interval setting by removing from the canonical compensator and the interventions their mass on $(T,\infty)$. Apart from the different time domains, the setup and notation is identical to the one used in Section \ref{section: identification with multiple interventions in the general MPP setting}. Specifically:
\begin{itemize}
    \item $N = (N^i)_{i \in \mathcal{I}_d}$ is an MPP with mark space $(X, \X)$, where $X$ is a Borel space. Each component $N^i$ operates on its own mark subspace $(X^i, \mathcal{X}^i)$, where the $X^i$'s are pairwise disjoint subsets of $X$. 
    \item $L$ is a baseline random element taking values in $(S, \mathcal{S})$, and $(L, N)$ take values in the canonical space $(S \times \N^X, \S \otimes \H^X)$ (see Appendix \ref{appendix: the canonical space of point process realizations}). Each component $N^i$ of $N$ takes values in the canonical space $(\N_T^{X^i}, \H_T^{X^i})$.
    \item We consider interventions on components indexed by $J \subseteq \mathcal{I}_d$. For each $j \in J$, the intervention $\mathfrak{n}^j$ is a predictable counting measure on the canonical space; a kernel from $S \times \N_T^X$ to $\mathbb R_+ \times X^j$ satisfying $\mathfrak n^j(l, \varphi, dt \times dx) = \mathfrak n^j(l, \varphi|_{t-}, dt \times dx)$. $\mathfrak n^j(l, \varphi)$ is for each $(l, \varphi) \in S \times \N_T^X$ a feasible trajectory for the process $N^j$; thus, the mapping $(l, \varphi) \mapsto \mathfrak n^j(l, \varphi)$ is a map from $S \times \N_T^X$ to $\N_T^{X^j}$.
\end{itemize}


To ensure the processes have the correct laws, we work with canonical compensators \cite[Theorem 4.2.2, Theorem 4.3.9]{LastBrandt1995marked}. We fix a version of the canonical compensator $\alpha$ of $N$ with respect to $P$ and the filtration generated by $N$ and $L$. $\alpha$ is specified in \eqref{eq: alpha canonical compensator}, and it depends on both the baseline variable $L$ and the MPP $N$. 

The canonical compensator and interventions $\mathfrak n^j$ characterize $\tilde{\alpha}$, the canonical compensator of the potential outcomes process. The canonical compensators are related via
\begin{align}
    \alpha(l, \varphi, dt \times dx) &= \sum_{j \in J} \alpha^j(l, \varphi, dt \times dx) + \alpha^{\setminus J}(l, \varphi, dt \times dx), \label{eq: obs can comp marked} \\
    \tilde \alpha(l, \varphi, dt \times dx) &= \sum_{j \in J} \mathfrak{n}^j(l, \varphi, dt \times dx \cap (\mathbb R_+ \times X^j) ) + \alpha^{\setminus J}(l, \varphi, dt \times dx), \label{eq: pot can comp marked}
\end{align}
where $\alpha^j(l, \varphi, dt \times dx) = \alpha(l, \varphi, dt \times dx \cap (\mathbb{R}_+ \times X^j))$, $\alpha^{\setminus J}(l, \varphi, dt \times dx) = \alpha(l, \varphi, dt \times dx \cap (\mathbb{R}_+ \times X^{\setminus J}))$, and $X^{\setminus J} = \cup_{i \in \mathcal I_d \setminus J} X^i$.

Our construction relies on certain orthogonality assumptions. Specifically, for each $j, h \in J$ satisfying $j \neq h$, we assume
\begin{align}
    \begin{split}
        \bar{\alpha}^j(l, \varphi, \{ t \}) \bar{\alpha}^h(l, \varphi, \{ t \}) &= 0,  \\
        \bar{\alpha}^j(l, \varphi, \{ t \}) \bar{\alpha}^{\setminus J}(l, \varphi, \{ t \}) &= 0,   \\
        \bar{\mathfrak n}^j (l, \varphi, \{ t \}) \bar{\alpha}^{\setminus J}(l, \varphi, \{ t \}) &= 0, \\
        \bar{\mathfrak n}^j (l, \varphi, \{ t \}) \bar{\mathfrak n}^h (l, \varphi, \{ t \}) &= 0,
    \end{split}
     \label{eq: ort obs multi marked}
\end{align}
for each $l$, $\varphi$, and $t$. Recall that for a random measure $\mu$ operating on a mark set $X^\mu$, we write $\bar\mu(dt) = \mu(dt \times X^\mu)$ for its total measure. In particular, $\bar{\alpha}^j(l, \varphi, dt) = \alpha^j(l, \varphi, dt \times X^j)$, $\bar{\alpha}^{\setminus J}(l, \varphi, dt) = \alpha^{\setminus J}(l, \varphi, dt \times X^{\setminus J})$ and $\bar{\mathfrak n}^j(l, \varphi, dt) = \mathfrak{n}^j(l, \varphi, dt \times X^j)$. 

The conditions in \eqref{eq: ort obs multi marked} correspond to the regularity conditions \eqref{eq: can comp j h}-\eqref{eq: interv interv} in Theorem \ref{theorem: multiple marked interventions}. These conditions imply the regularity condition \eqref{eq:app-jump} in Appendix \ref{appendix: general MPP potential outcomes} (assuming there are no unmeasured components in Appendix \ref{appendix: general MPP potential outcomes}).

We also assume for each $l, \varphi$ that 
\begin{align}
    \tilde \alpha(l, \varphi, dt \times X)I(\tilde \pi_\infty'(l,\varphi) < t) = 0, \label{eq: tilde alpha explosion regularity condition}
\end{align}
where $\tilde \pi_\infty'(l,\varphi) = \inf\big\{ s > 0 | \tilde \alpha(l,\varphi, (0,s] \times X) = \infty \big\}$. Combining this with \eqref{eq: ort obs multi marked}, the regularity conditions \eqref{eq:app-explosion}-\eqref{eq:app-jump} in Appendix \ref{appendix: general MPP potential outcomes} hold. These conditions ensure a well-defined canonical compensator suitable for constructing an MPP, as described in more detail in that appendix.

We will in the construction have use of the component-specific survival functions
\begin{align}
    U^j(l, \varphi|_s, s, t) &:= \prodi_{s < u \leq t} \big( 1 - \bar{\alpha}^j(l, \varphi|_s, du) \big), \quad j \in J, \label{eq: Uj definition} \\
    U^{\setminus J}(l, \varphi|_s, s, t) &:= \prodi_{s < u \leq t} \big( 1 - \bar{\alpha}^{\setminus J}(l, \varphi|_s, du) \big), \label{eq: UsetminusJ definition} \\
    \mathscr{U}^j(l, \varphi|_s, s, t) &:= \prodi_{s < u \leq t} \big( 1 - \bar{\mathfrak{n}}^j(l, \varphi|_s, du) \big), \quad j \in J. \label{eq: mathscrUj definition}
\end{align}
Due to the orthogonality conditions \eqref{eq: ort obs multi marked}, we have the factorizations
\begin{align}
    U(l, \varphi|_s, s, t) &:= \prodi_{s < u \leq t} (1 - \bar{\alpha}(l, \varphi|_s, du)) = \prod_{j \in J} U^j(l, \varphi|_s, s, t) U^{\setminus J}(l, \varphi|_s, s, t), \label{eq: U definition multi marked} \\
    \tilde{U}(l, \varphi|_s, s, t) &:= \prodi_{s < u \leq t} (1 - \bar{\tilde \alpha}(l, \varphi|_s, du)) =  \prod_{j \in J} \mathscr{U}^j(l, \varphi|_s, s, t) U^{\setminus J}(l, \varphi|_s, s, t). \label{eq: tilde U decomp multi marked}
\end{align}
We will in the construction also use the kernels $K, \tilde K$ characterized by the disintegrations
\begin{align}
    \alpha(l,\varphi, dt \times dx )  &=  K(l, \varphi, t, dx) \bar \alpha(l, \varphi, dt), \label{eq: K kernel definition} \\
    \tilde \alpha(l,\varphi, dt \times dx )  &=  \tilde K(l, \varphi, t, dx) \bar {\tilde \alpha}(l, \varphi, dt), \label{eq: tilde K kernel definition} 
\end{align}
and the derived normalized kernels
\begin{align}
    K^j(l,\varphi, t, dx) &:= \frac{K(l,\varphi, t, dx \cap X^j)}{K(l,\varphi, t, X^j)}, \label{eq: Kj kernel definition} \\
    K^{\setminus J}(l,\varphi, t, dx) &:= \frac{K(l,\varphi, t, dx \cap X^{\setminus J})}{K(l,\varphi, t, X^{\setminus J})}. \label{eq: KsetminusJ kernel definition}
\end{align}


\begin{lemma}[No shared mass]\label{lemma: no shared mass II marked}
    Suppose that \eqref{eq: ort obs multi marked} holds. With the notation and variables defined in Proposition \ref{proposition: mcp construction of observed and potential outcomes marked} (see Algorithm \ref{alg:scheme multi marked}), we have for each $j, h \in J$ with $j \neq h$, and for each $k \geq 1$, that
    \begin{align*}
        P'(T_k^j = T_k^h < \infty) &= P'(T_k^j = T_k^{\setminus J} < \infty) \\
        &= P'(\mathscr T_k^j = T_k^{\setminus J} < \infty ) = P'(\mathscr T_k^j = \mathscr T_k^h < \infty) = 0.
    \end{align*}
\end{lemma}
\begin{proof}
    This follows immediately from the definitions in Algorithm \ref{alg:scheme multi marked} under the orthogonality assumptions \eqref{eq: ort obs multi marked}.
\end{proof}

\begin{proposition}[Construction of observed and potential outcomes that satisfy identification conditions]\label{proposition: mcp construction of observed and potential outcomes marked}
    Suppose the canonical compensators $\alpha$ and $\tilde{\alpha}$ in \eqref{eq: obs can comp marked}-\eqref{eq: pot can comp marked} satisfy the orthogonality conditions \eqref{eq: ort obs multi marked}. Consider a probability space $(\Omega', \F', P')$ supporting a random element $L'$ at baseline such that $P'(L' \in \cdot) = P(L \in \cdot)$. 
    
    Construct MPPs $N' = (T_k', X_k')_{k \geq 1}$ and $\tilde N' = (\tilde T_k', \tilde X_k')_{k \geq 1}$ on $(\Omega', \F', P')$ via Algorithm \ref{alg:scheme multi marked} below, using the random variables
    \begin{align}
        \{\xi_k^j, \xi_k^{\setminus J}, \tilde \xi_k, \eta_k^j, \eta_k^{\setminus J}, \tilde \eta_k\}_{k \geq 1, j \in J}. \label{eq: uniform random variables marked}
    \end{align}
    Suppose that each random variable in \eqref{eq: uniform random variables marked} is uniformly distributed on $[0,1]$ under $P'$, and that the following independence conditions hold under $P'$: \footnote{The independence conditions \ref{enum: mutual indep multi marked}-\ref{enum: exch indep multi marked} can e.g. be made to hold by choosing the randomizers in \eqref{eq: uniform random variables marked} to be mutually independent and independent of $L'$ under $P'$.}
\begin{enumerate}[label=(I\arabic*)]
    \item \label{enum: mutual indep multi marked}
    For each $k \geq 1$:
    \begin{align*}
        &\text{The variables }\{ \xi_k^{j} , \eta_k^{j}, \xi_k^{\setminus J} , \eta_k^{\setminus J}, \tilde \xi_k, \tilde \eta_k, L'\}_{j \in J} \text{ are mutually independent}, \\
        &\text{and } \{\xi_k^j, \xi_k^{\setminus J}\}_{j\in J} \text{ are mutually independent conditional on } \F'_{T_{k-1}'},
    \end{align*}
    where $\F'_{T_{k-1}'} = \sigma(L', N'|_{T_{k-1}'})$ is the stopped $\sigma$-algebra at time $T_{k-1}'$. 
    
    \item \label{enum: C_k independence multi marked}
    For each $k \geq 1$:
    \begin{align*}
        \{\xi_k^j, \eta_k^j, \xi_k^{\setminus J}, \eta_k^{\setminus J}, \tilde \xi_k, \tilde \eta_k\}_{j \in J} 
        \indep_{P'} 
        \sigma\big(L', \{\xi_z^j, \eta_z^j, \xi_z^{\setminus J}, \eta_z^{\setminus J}, \tilde \xi_z, \tilde \eta_z\}_{z < k, j \in J}\big).
    \end{align*}
    
   \item \label{enum: exch indep multi marked}
    For each $k \geq 1$ and $j \in J$:
      \begin{align*}
        \xi_k^j &\indep_{P'} \{ \xi_k^h \}_{h \neq j} \cup \{\xi_k^{\setminus J}, \eta_k^{\setminus J}\} \cup \{ \xi_m^h, \xi_m^{\setminus J}, \eta_m^h, \eta_m^{\setminus J}, \tilde \xi_m, \tilde \eta_m \}_{h \in J, m > k} |\F'_{T_{k-1}'}, \\
        \eta_k^j &\indep_{P'} \{ \xi_k^h \}_{h \neq j} \cup \{\xi_k^{\setminus J}, \eta_k^{\setminus J}\} \cup \{ \xi_m^h, \xi_m^{\setminus J}, \eta_m^h, \eta_m^{\setminus J}, \tilde \xi_m, \tilde \eta_m \}_{h \in J, m > k} |\F'_{T_{k-1}'} \vee \sigma(T_k^j),
      \end{align*}
      where $T_k^j$ is as in Step \ref{step:draw_Tkj_general} in Algorithm \ref{alg:scheme multi marked}.
\end{enumerate}
    Consider deviation times as defined in \eqref{eq: tau j deviation time}; i.e., define for each $j \in J$
    $$
    \tau'^j := \inf\{s > 0 |  N'^j((0,s] \times D) \neq \mathfrak{n}^j(L',N',((0,s] \times D) \text{ for some } D \in X^j\},
    $$
    and set $\tau'^J := \wedge_{j \in J} \tau'^j$. Then, 
    \begin{enumerate}[label=(P\arabic*)]
        \item \label{enum: prop consistency marked} 
        The observed and potential outcome processes coincide until just before the first regime deviation time, i.e.
        $$
        N'|_t I(\tau'^J > t) = \tilde N'|_t I(\tau'^J > t) \quad \text{for each } t \geq 0;
        $$
        
        \item \label{enum: prop exchangeability marked} 
        $~^{\tau'^J}\alpha^j(L', N', \cdot)$ defines a compensator of $~^{\tau'^J}N'^j$ with respect to both $\{\F_t'\}_t$ and $\{\F_t' \vee \sigma(\tilde N') \}_t$ under $P'$ for each $j \in J$, where $\{\F_t'\}_t$ is the filtration generated by $(L', N')$;
        
        \item \label{enum: prop correct law marked} The canonical compensator of 
        $N'$ with respect to $\{ \F_t' \}_t$ and $P'$ is $\alpha$, and the canonical compensator of 
        $\tilde N'$ with respect to $\{ \tilde \F_t' \}_t$ and $P'$ is $\tilde \alpha$, where $\{ \tilde \F_t' \}_t$ is the filtration generated by $(L', \tilde N')$. In particular, $\text{Law}_{P'}(L', N') = \text{Law}_P(L, N)$  and $\text{Law}_{P'}(L', \tilde N') = \text{Law}_P(L, \tilde N)$.
    \end{enumerate}
\end{proposition}

\begin{algorithm}
\caption{Scheme for generating observed and potential outcome processes such that identification conditions are satisfied}
\label{alg:scheme multi marked}
\begin{algorithmic}[1]
\State Initialize $T_0' = \tilde T_0' = 0$. 
\For{$k=1, 2, \cdots$}
    \State Define $C_k = \{ \omega \in \Omega' | T_i'(\omega)=\tilde T_i'(\omega) \text{ and } X_i'(\omega) = \tilde X_i'(\omega) \text{ for all } i = 1, \dots,  k-1 \}$ \label{step:set_Ck}
    
    \For{$j \in J$}
        \State Draw $T_k^{j}$ from $U^j(L', N'|_{T_{k-1}'}, T_{k-1}', \cdot)$ using $\xi_k^j$ \label{step:draw_Tkj_general}
        \State Set $\mathscr T_k^{j} = \inf\{ s > \tilde T_{k-1}' : \mathfrak{n}^j(L', \tilde N'|_{\tilde T_{k-1}'}, (\tilde T_{k-1}', s] \times X^j) \neq 0 \}$ \label{step:set_mathscr_Tkj}
    \EndFor
    
    \State Draw $T_k^{\setminus J}$ from $U^{\setminus J}(L', N'|_{T_{k-1}'}, T_{k-1}', \cdot)$ using $\xi_k^{\setminus J}$ \label{step:draw_TksetminusJ}
    \State Draw $\tilde T_k^{*}$ from $\tilde U(L', \tilde N'|_{\tilde T_{k-1}'}, \tilde T_{k-1}', \cdot)$ using $\tilde \xi_k$ \label{step:draw_tilde_Tkstar}
    
    \State Set $T_k' = \wedge_{h \in J} T_k^h \wedge T_k^{\setminus J}$ \label{step:set_Tkprime}
    \State Set $\tilde T_k' = \wedge_{h \in J} \mathscr T_k^h \wedge T_k^{\setminus J}  I_{C_k} + \tilde T_k^{*} I_{C_k^c}$ \label{step:set_tilde_Tkprime}
    
    where $C_k^c$ denotes the complement of $C_k$ in $\Omega'$
    
    \For{$j \in J$}
        \State Set $V_k^{j} = \nabla$ and $\mathscr V_k^{j} = \nabla$ \label{step:initialize_Vkj_mathscrVkj}
        \State On $\{T_k^{j} < \infty\}$: draw $V_k^{j}$ from $K^j(L', N'|_{T_{k-1}'}, T_{k}^j, \cdot)$ using $\eta_k^j$ \label{step:draw_Vkj}
        \State On $\{\mathscr T_k^{j} < \infty\}$: set $\mathscr V_k^{j}$ to be the unique $x \in X^j$ such that $\mathfrak{n}^j(L', \tilde N'|_{\tilde T_{k-1}'}, \{\mathscr T_k^{j}\} \times \{x\}) > 0$ \label{step:set_mathscr_Vkj}
    \EndFor
    
    \State Set $V_k^{\setminus J} = \nabla$ and $\tilde X_k^{*} = \nabla$ \label{step:initialize_VksetminusJ_tilde_Xkstar}
    \State On $\{T_k^{\setminus J} < \infty\}$: draw $V_k^{\setminus J}$ from $K^{\setminus J}(L', N'|_{T_{k-1}'}, T_{k}^{\setminus J}, \cdot)$ using $\eta_k^{\setminus J}$ \label{step:draw_VksetminusJ}
    \State On $\{\tilde T_k^{*} < \infty\}$: draw $\tilde X_k^{*}$ from $\tilde K(L', \tilde N'|_{\tilde T_{k-1}'}, \tilde T_{k}^*, \cdot)$ using $\tilde \eta_k$ \label{step:draw_tilde_Xkstar}
    
    \State Set $X_k' = \sum_{h \in J} V_k^{h} I(T_k^{h} = T_k') + V_k^{\setminus J} I(T_k^{\setminus J} = T_k')$ \label{step:set_Xkprime}
    \State Set $\tilde X_k' = \big[ \sum_{h \in J} \mathscr V_k^{h} I(\mathscr T_k^{h} = \tilde T_k') + V_k^{\setminus J} I(T_k^{\setminus J} = \tilde T_k') \big] I_{C_k} + \tilde X_k^{*} I_{C_k^c}$ \label{step:set_tilde_Xkprime}
\EndFor
\end{algorithmic}
\end{algorithm}

\begin{proof}

\noindent \textbf{Conditional independencies.} 

\ \\
From the independence assumptions in Proposition \ref{proposition: mcp construction of observed and potential outcomes marked} and the definitions in Algorithm \ref{alg:scheme multi marked}, we can deduce the following:
\begin{enumerate}[label=(D\arabic*), leftmargin=*]
    \item \label{enum: D1} For each $k \geq 1$, $\{T_k^h, T_k^{\setminus J}\}_{h \in J}$ are mutually independent conditional on $\F'_{T_{k-1}'}$ under $P'$,
    \item \label{enum: D2} For each $k \geq 1$ and $j \in J$,
    $$
    V_k^j \indep_{P'} \{T_k^h\}_{h \neq j} \cup \{T_k^{\setminus J}\} \big| \F'_{T_{k-1}'} \vee \sigma(T_k^j).
    $$
\end{enumerate}

Property \ref{enum: D1} follows from the mutual independence of $\{\xi_k^h, \xi_k^{\setminus J}\}_{h\in J}$ given $\F'_{T_{k-1}'}$ (as assumed in \ref{enum: mutual indep multi marked}), since for each $j \in J$, $T_k^j$ is a function of $(L', N'|_{T_{k-1}'}, \xi_k^j)$ (Step \ref{step:draw_Tkj_general} of Algorithm \ref{alg:scheme multi marked}), and similarly $T_k^{\setminus J}$ is a function of $(L', N'|_{T_{k-1}'}, \xi_k^{\setminus J})$ (Step \ref{step:draw_TksetminusJ}). Because $\F_{T_{k-1}'}' = \sigma(L', N'|_{T_{k-1}'})$, and measurable transformations preserve conditional independencies, \ref{enum: D1} follows. Formally, this deduction can be made using a monotone class argument.

Property \ref{enum: D2} can be deduced using similar techniques. From \ref{enum: exch indep multi marked}, and the decomposition rule of conditional independence, we have
$$\eta_k^j \indep_{P'} \{ \xi_k^h \}_{h \neq j} \cup \{\xi_k^{\setminus J}\} | \F'_{T_{k-1}'} \vee \sigma(T_k^j).$$
$T_k^h$ is a function of $(L', N'|_{T_{k-1}'}, \xi_k^h)$ by Step \ref{step:draw_Tkj_general} in Algorithm \ref{alg:scheme multi marked} (with $h$ in place of $j$), and $T_k^{\setminus J}$ is a function of $(L', N'|_{T_{k-1}'}, \xi_k^{\setminus J})$ by Step \ref{step:draw_TksetminusJ}. Moreover, $V_k^j$ is a function of $(L', N'|_{T_{k-1}'}, T_k^j, \eta_k^j)$ by Step \ref{step:draw_Vkj}. Since measurable transformations preserve conditional independencies, and $\F_{T_{k-1}'}' = \sigma(L', N'|_{T_{k-1}'})$, \ref{enum: D2} follows, where the result can be shown from first principles via a monotone class argument.

\noindent \textbf{Verifying the consistency condition \ref{enum: prop consistency marked}.}
\ \\
We verify that Algorithm \ref{alg:scheme multi marked} respects the consistency condition \ref{enum: prop consistency marked}. Consider first the initial time step. We show that, if $\tau'^J > T_1'$, then
\begin{align}
    (T_1', X_1') = (\tilde T_1', \tilde X_1'). \label{eq: cons first time step simple}
\end{align}
Suppose then that $\tau'^J > T_1'$, and note that $C_1 = \Omega'$ by Algorithm \ref{alg:scheme multi marked}. There are two cases to check:
\begin{itemize}
    \item \textit{Case I: The first observed event is a treatment event (of some type $j \in J$), and the regime is followed.} By Algorithm \ref{alg:scheme multi marked}, the first observed event is a treatment event if, for some $j \in J$, $T_1^{j} = \wedge_{h \in J} T_1^h \wedge T_1^{\setminus J}$. The regime is followed whenever $\mathscr T_1^{j} = \wedge_{h \in J} \mathscr T_1^h \wedge T_1^{\setminus J}$, $\mathscr T_1^{j} = T_1^{j}$, and the observed mark $V_1^j$ equals the 'intervened' mark $\mathscr V_1^{j}$. By the assignments in Algorithm \ref{alg:scheme multi marked}, we see that 
    \begin{align*}
        T_1' &= T_1^{j} = \mathscr T_1^{j} = \tilde T_1', \\
        X_1' &= V_1^{j} = \mathscr V_1^{j} = \tilde X_1',
    \end{align*}
    which shows \eqref{eq: cons first time step simple}.

    \item \textit{Case II: The first observed event is a non-treatment event, and the regime is followed.}
    The first observed event is a non-treatment event if $T_1^{\setminus J}  = \wedge_{j \in J} T_1^j \wedge T_1^{\setminus J}$. The regime is followed through the first observed event if also $T_1^{\setminus J} = \wedge_{j \in J} \mathscr T_1^j \wedge T_1^{\setminus J}$. By the assignments in Algorithm \ref{alg:scheme multi marked}, we see that
    \begin{align*}
        T_1' &= T_1^{\setminus J} = \tilde T_1', \\
        X_1' &= V_1^{\setminus J} = \tilde X_1',
    \end{align*}
    and \eqref{eq: cons first time step simple} holds in this case also.
\end{itemize}
Now assume inductively that $(T_i', X_i') = (\tilde T_i', \tilde X_i')$ for all $i < k$ whenever $\tau'^J > T_{k-1}'$. Then, we are on the event $C_k$. On $C_k$, the assignments in Algorithm \ref{alg:scheme multi marked} for the $k$th step mirror those for the first step. By arguing as in the two cases above, we conclude that if $\tau'^J > T_k'$, then also $(T_k', X_k') = (\tilde T_k', \tilde X_k')$. 
The condition \ref{enum: prop consistency marked} follows.

\noindent \textbf{Correctness of the laws.}
\ \\
We now verify that $(L',N')$ has the correct law under $P'$. Toward that end we show that $\alpha$ is the canonical compensator of $N'$ under $P'$ with respect to the filtration $\{\F'_t\}_t$. It suffices to show that, for each $k \geq 1$,
\begin{align}
    P'(T_{k+1}' > t | \F'_{T_k'} ) &= U(L', N'|_{T_{k}'}, T_k', t), \label{eq: time ker multi marked} \\
    P'(T_{k+1}' \in dt, X_{k+1}' \in dx | \F'_{T_k'}) &=  U(L', N'|_{T_{k}'}, T_k', t-)\alpha(L', N'|_{T_{k}'}, dt \times dx), \label{eq: mark ker multi marked}
\end{align}
with $U$ as in \eqref{eq: U definition multi marked}. This is sufficient because, by \cite[Theorem 4.1.11, 4.2.2, 4.3.7, 4.3.8]{LastBrandt1995marked}, \eqref{eq: time ker multi marked}-\eqref{eq: mark ker multi marked} determine $\alpha$ as the canonical compensator.

Using \ref{enum: D1} and the definition $T_{k+1}' = \wedge_{j \in J} T_{k+1}^j \wedge T_{k+1}^{\setminus J}$ (in Algorithm \ref{alg:scheme multi marked}), we have
\begin{align}
    \begin{split}
        P'(T_{k+1}' > t |\F'_{T_k'} ) 
        &= P'\Big( \wedge_{j \in J} T_{k+1}^j \wedge T_{k+1}^{\setminus J} > t | \F'_{T_k'} \Big) \\
        &= \prod_{j \in J} P'(T_{k+1}^j > t |\F'_{T_k'} ) \cdot P'(T_{k+1}^{\setminus J} > t |\F'_{T_k'}) \\
        &= \prod_{j \in J} U^j(L', N'|_{T_k'},T_k', t)  U^{\setminus J}(L', N'|_{T_k'},T_k', t) \\
        &= U(L', N'|_{T_k'},T_k', t),
    \end{split}\label{eq: next event time indep multi marked}
\end{align}
where we used the assignments in Algorithm \ref{alg:scheme multi marked}, 
and \eqref{eq: U definition multi marked}. This proves \eqref{eq: time ker multi marked}.

We now establish \eqref{eq: mark ker multi marked}. Consider a measurable set $B \subseteq X^j$ for some $j \in J$. From the assignments in Algorithm \ref{alg:scheme multi marked}, we have
\begin{align}
\begin{split}
    &P'(T_{k+1}' \in dt, X_{k+1}' \in B | \F'_{T_k'}) \\
    &= P'\Big( T_{k+1}^j \in dt, \wedge_{h \in J \setminus \{j\}} T_{k+1}^h \geq t, T_{k+1}^{\setminus J} \geq t, V_{k+1}^j \in B | \F'_{T_k'} \Big).
\end{split}\label{eq: T kplus1 event j marked}
\end{align}
Note that, due to Lemma \ref{lemma: no shared mass II marked}, the measures $I(T_{k+1}^j \in dt,\wedge_{h \in J \setminus \{j\}} T_{k+1}^h \geq t, T_{k+1}^{\setminus J} \geq t, V_{k+1}^j \in B)$ and $I(T_{k+1}^j \in dt, \wedge_{h \in J \setminus \{j\}} T_{k+1}^h > t, T_{k+1}^{\setminus J} > t, V_{k+1}^j \in B)$ agree except on an exceptional set, and we can consequently move freely between them in our calculations.

By \ref{enum: D1} and  \ref{enum: D2}, the probability in \eqref{eq: T kplus1 event j marked} factorizes as
\begin{align*}
    &P'(T_{k+1}^j \in dt, \wedge_{h \neq j} T_{k+1}^h \geq t, T_{k+1}^{\setminus J} \geq t, V_{k+1}^j \in B | \F'_{T_k'}) \\
    &= P'(V_{k+1}^j \in B | \F'_{T_k'}, T_{k+1}^j = t) P'(T_{k+1}^j \in dt | \F'_{T_k'}) \\
    &\quad \times P'(\wedge_{h \neq j} T_{k+1}^h \geq t | \F'_{T_k'}) P'(T_{k+1}^{\setminus J} \geq t | \F'_{T_k'}) \\
    &= K^j(L', N'|_{T_k'}, t, B) \bar{\alpha}^j(L', N'|_{T_k'}, dt) \\
    &\quad \times  \prod_{h \in J} U^h(L', N'|_{T_k'}, T_k', t-)  U^{\setminus J}(L', N'|_{T_k'}, T_k', t-) \\
    &= \alpha(L', N'|_{T_k'}, dt \times B) U(L', N'|_{T_k'}, T_k', t-),
\end{align*}
where the last equality follows from 1): the identity $\alpha(l,\varphi, dt \times B) = K^j(l,\varphi, t, B) \allowbreak \times \bar\alpha^j(l,\varphi, dt)$ for $B \subseteq X^j$ (which combines \eqref{eq: K kernel definition} and \eqref{eq: Kj kernel definition}), and 2): the factorization \eqref{eq: U definition multi marked}. An analogous argument for a set $B \subseteq X^{\setminus J}$ gives
\begin{align*}
    P'(T_{k+1}' \in dt, X_{k+1}' \in B | \F'_{T_k'}) 
    &= K^{\setminus J}(L', N'|_{T_k'}, t, B) \bar{\alpha}^{\setminus J}(L', N'|_{T_k'}, dt) \\
    &\quad \times  \prod_{h \in J} U^h(L', N'|_{T_k'}, T_k', t-)  U^{\setminus J}(L', N'|_{T_k'}, T_k', t-) \\
    &= \alpha(L', N'|_{T_k'}, dt \times B) U(L', N'|_{T_k'}, T_k', t-),
\end{align*}
where we used that $\alpha(l,\varphi, dt \times B) = K^{\setminus J}(l,\varphi, t, B) \bar\alpha^{\setminus J}(l,\varphi, dt)$ for $B \subseteq X^{\setminus J}$ (which combines \eqref{eq: K kernel definition} and \eqref{eq: KsetminusJ kernel definition}). The result for general $B \in \mathcal{X}$ follows: By decomposing $B$ into $B^j = B \cap X^j$ and $B^{\setminus J} = B \cap X^{\setminus J}$, we get pairwise disjoint subsets, and the result is obtained by combining the previous results with the additivity of the measures involved. This verifies \eqref{eq: mark ker multi marked}. We conclude that $\alpha$ is the canonical compensator of $N'$.

By \cite[Theorem 8.2.1]{LastBrandt1995marked}, for each $l \in S$, there exists a probability measure $P_{\alpha(l)}$ on the canonical space where $\alpha(l) = \alpha(l,\cdot)$, and by \cite[Theorem 8.2.2]{LastBrandt1995marked}, $P'(N' \in \cdot|L'=l) = P_{\alpha(l)}$ for $P' \circ L^{\prime -1}$-a.e. $l$. The same theorems give $P(N \in \cdot|L=l) = P_{\alpha(l)}$ for $P \circ L^{-1}$-a.e. $l$. Since by assumption $P'(L'\in \cdot) = P(L\in \cdot)$, it follows from the disintegration theorem that $P'((N', L') \in  \cdot ) = P((N, L) \in  \cdot )$.

We now show that $\tilde{\alpha}$ is the canonical compensator of $\tilde{N}'$ under $P'$ with respect to the filtration $\{ \tilde \F'_t \}_t$ generated by $L'$ and $\tilde N'$. As with the observed process, it suffices to prove
\begin{align}
    P'(\tilde T_{k+1}' > t |\tilde \F'_{\tilde T_k'}) &= \tilde U(L', \tilde N'|_{\tilde T_k'}, \tilde T_k', t) \label{eq: tilde time ker multi marked} \\
    P'(\tilde T_{k+1}' \in dt, \tilde X_{k+1}' \in dx |\tilde \F'_{\tilde T_k'}) &= \tilde U(L', \tilde N'|_{\tilde T_k'}, \tilde T_k', t-) \tilde{\alpha}(L', \tilde N'|_{\tilde T_k'}, dt \times dx). \label{eq: tilde mark ker multi marked}
\end{align}
First, Algorithm \ref{alg:scheme multi marked} defines each $\mathscr T_{k+1}^j$ deterministically from the intervention $\mathfrak{n}^j$ and the history $L',\tilde N'|_{\tilde T_k'}$. 
Because $\tilde \F'_{\tilde T_k'} = \sigma(L', \tilde N'|_{T_k'})$, e.g. \cite[Theorem 2.2.14]{LastBrandt1995marked}, $\mathscr T_{k+1}^j$ is $\tilde \F'_{\tilde T_k'}$-measurable, and
\begin{align}
\mathscr{U}^j(L', \tilde N'|_{\tilde T_k'}, \tilde T_k', t) = \prodi_{\tilde T_k' < u \leq t}\big( 1 - \bar{\mathfrak{n}}^j(L', \tilde N'|_{\tilde T_k'}, du) \big) = I(\mathscr T_{k+1}^j > t). \label{eq: intervention prodi}
\end{align}
To establish \eqref{eq: tilde time ker multi marked}, set $\mathcal{G} = \sigma(\tilde \F'_{\tilde T_k'}, I_{C_{k+1}})$. By the assignments in Algorithm \ref{alg:scheme multi marked} we get
\begin{align}
    P'(\tilde T_{k+1}' > t |\mathcal{G}) = P'\Big( \wedge_{j \in J} \mathscr T_{k+1}^j \wedge T_{k+1}^{\setminus J} > t \Big| \mathcal{G} \Big) I_{C_{k+1}} 
    + P'(\tilde T_{k+1}^{*} > t |\mathcal{G}) I_{C_{k+1}^c}. \label{eq: G sum multi marked}
\end{align}
Consider now the first term on the right-hand side of \eqref{eq: G sum multi marked}. Because each $\mathscr T_{k+1}^j$ is $\mathcal{G}$-measurable we get
\begin{align}
\begin{split}
    P'\Big( \wedge_{j \in J} \mathscr T_{k+1}^j \wedge T_{k+1}^{\setminus J} > t \Big| \mathcal{G} \Big)
&=  \prod_{j \in J} I(\mathscr T_{k+1}^j > t)  
   P'(T_{k+1}^{\setminus J} > t |\mathcal{G}) \\
&=  \prod_{j \in J} \mathscr{U}^j(L', \tilde N'|_{\tilde T_k'}, \tilde T_k', t) P'(T_{k+1}^{\setminus J} > t |\mathcal{G}),
\end{split}\label{eq: hee fact}
\end{align}
where we used \eqref{eq: intervention prodi}. On $C_{k+1}$, we have $\tilde N'|_{\tilde T_k'} = N'|_{ T_k'}$ by the established consistency condition. Therefore, on $C_{k+1}$, we have by Algorithm \ref{alg:scheme multi marked} that $\{T_{k+1}^{\setminus J} > t\} = \{U_t^{\setminus J} > \xi_{k+1}^{\setminus J}\}$, with $U_t^{\setminus J} \allowbreak := U^{\setminus J}(L', \tilde N'|_{\tilde T_k'}, \tilde T_k', t)$. $U_t^{\setminus J}$ is $\tilde \F_{\tilde T_k'}$-measurable, and therefore $\G$-measurable.

By Assumption \ref{enum: C_k independence multi marked}, $\xi_{k+1}^{\setminus J}$ is independent of 
$\sigma\big(L', \{\xi_z^j, \eta_z^j, \xi_z^{\setminus J}, \eta_z^{\setminus J}, \tilde \xi_z, \allowbreak \tilde \eta_z\}_{z \leq k, j \in J}\big)$. 
Since $\mathcal{G} = \tilde \F'_{\tilde T_k'} \vee \sigma(I_{C_{k+1}})$ and both $\tilde \F'_{\tilde T_k'}$ and $I_{C_{k+1}}$ are 
$\sigma\big(L', \{\xi_z^j, \eta_z^j, \xi_z^{\setminus J}, \eta_z^{\setminus J}, \tilde \xi_z, \allowbreak \tilde \eta_z\}_{z \leq k, j \in J}\big)$-measurable, we get 
$\xi_{k+1}^{\setminus J} \indep_{P'} \mathcal{G}$ by the decomposition property of (conditional) independence.

Applying Lemma \ref{lemma: williams independence} with $X = U_t^{\setminus J}$, $Z = \xi_{k+1}^{\setminus J}$, and $f(x,z) = I(x > z)$, we obtain
$$
P'(T_{k+1}^{\setminus J} > t |\mathcal{G}) = g(U_t^{\setminus J}),
$$
where $g(x) = E_{P'}[I(x > \xi_{k+1}^{\setminus J})] = x$, because $\xi_{k+1}^{\setminus J}$ is uniform on $[0,1]$ under $P'$. Thus, $P'(T_{k+1}^{\setminus J} > t |\mathcal{G}) = U_t^{\setminus J}$. Substituting this into \eqref{eq: hee fact} gives
\begin{align*}
    P'\Big( \wedge_{j \in J} \mathscr T_{k+1}^j \wedge T_{k+1}^{\setminus J} > t \Big| \mathcal{G} \Big)
&= \Big( \prod_{j \in J} \mathscr{U}^j(L', \tilde N'|_{\tilde T_k'}, \tilde T_k', t) \Big) 
   U^{\setminus J}(L', \tilde N'|_{\tilde T_k'}, \tilde T_k', t)
\\
&= \tilde U(L', \tilde N'|_{\tilde T_k'}, \tilde T_k', t). 
\end{align*}

For the second term in \eqref{eq: G sum multi marked} (on $C_{k+1}^c$), we similarly have $\{\tilde T_{k+1}^{*} > t\} = \{\tilde U_t > \tilde \xi_{k+1}\}$ with $\tilde U_t := \tilde U(L', \tilde N'|_{\tilde T_k'}, \tilde T_k', t)$, which is $\mathcal{G}$-measurable. Again by Assumption \ref{enum: C_k independence multi marked}, $\tilde \xi_{k+1}$ is independent of $\mathcal{G}$ under $P'$. Applying Lemma \ref{lemma: williams independence} as above gives $P'(\tilde T_{k+1}^{*} > t |\mathcal{G}) = \tilde U_t$. Inserting the preceding results into \eqref{eq: G sum multi marked} gives
$$
P'(\tilde T_{k+1}' > t |\mathcal{G}) 
= \tilde U_t I_{C_{k+1}} + \tilde U_t I_{C_{k+1}^c} = \tilde U_t.
$$
Since $\tilde U_t$ is $\tilde \F'_{\tilde T_k'}$-measurable, we obtain \eqref{eq: tilde time ker multi marked} by taking conditional expectations of the preceding equation with respect to $\tilde \F'_{\tilde T_k'}$.

The equality \eqref{eq: tilde mark ker multi marked} follows by analogous arguments, using that $\tilde \eta_{k+1}$ is independent of $\mathcal{G}$ by assumption \ref{enum: C_k independence multi marked}, and Lemma \ref{lemma: williams independence}. We omit the detailed computation.

Once \eqref{eq: tilde mark ker multi marked} is shown, it follows that $\tilde{\alpha}$ is the canonical compensator of $\tilde{N}'$ under $P'$. The same reasoning as for the observed data MPP gives that $\text{Law}_{P'}(L', \tilde N') = \text{Law}_P(L, \tilde N).$

\noindent \textbf{Verifying the exchangeability condition \ref{enum: prop exchangeability marked}.}
\ \\
We now verify that the claimed exchangeability condition \ref{enum: prop exchangeability marked} holds. We only establish the condition for one fixed component $j \in J$, as the condition for the remaining components is shown with identical reasoning.

We first verify that the following conditional independencies hold for each $k \geq 1$:
\begin{align}
    T_k^{j} &\indep_{P'} (T_k^{h}, T_k^{\setminus J}, \tilde N')_{h \in J \setminus \{j\}} \bigm| \F'_{T_{k-1}'}, \tau'^J > T_{k-1}', \label{eq: Tkj cond indep} \\
    V_k^{j} &\indep_{P'} (T_k^{h}, T_k^{\setminus J}, \tilde N')_{h \in J \setminus \{j\}} \bigm| \F'_{T_{k-1}'} \vee \sigma(T_k^{j}), \tau'^J > T_{k-1}'. \label{eq: Vkj Tkj cond indep}
\end{align}
\eqref{eq: Tkj cond indep}-\eqref{eq: Vkj Tkj cond indep} are independencies conditioning on a $\sigma$-algebra $\A$ mixed with an event $A \in \A$. We defined such mixed conditional independence in Appendix \ref{subsection: conditional independence and expectation} under the assumption that the event $A$ has positive probability. For verifying the exchangeability condition \ref{enum: prop exchangeability marked}, it suffices to establish \eqref{eq: Tkj cond indep}-\eqref{eq: Vkj Tkj cond indep} under the assumption $P'(\tau'^J > T_{k-1}') > 0$: if $\{ \tau'^J > T_{k-1}' \}$ is a $P'$-null set, nothing needs to be shown, as can be seen from \eqref{eq: split Gk hazard multi marked}. We therefore assume $P'(\tau'^J > T_{k-1}') > 0$ in the following derivations and understand mixed conditional independence as in \eqref{eq: conditional independence event probability measure} in Appendix \ref{subsection: conditional independence and expectation}. 

To show \eqref{eq: Tkj cond indep}, note that on $\{ \tau'^J > T_{k-1}' \}$, we have $~^{T_{k-1}'}\tilde N' = ~^{T_{k-1}'}N'$ by the previously established consistency condition. Because $~^{T_{k-1}'}N'$ is $\F'_{T_{k-1}'}$-measurable, \cite[Theorem 2.2.14]{LastBrandt1995marked}, \eqref{eq: Tkj cond indep} holds if we can establish the conditional independence
\begin{align}
    T_k^{j} \indep_{P'} (T_k^{h}, T_k^{\setminus J}, ~_{T_{k-1}'}\tilde N')_{h \in J \setminus \{j\}} |\F'_{T_{k-1}'}, \tau'^J > T_{k-1}', \label{eq: Tkj tilde N temp independence}
\end{align}
where we use the notation $~_tZ(ds \times dx) := I(s > t)Z(ds \times dx)$ for a random measure $Z$ 'started' at $t$.

To proceed we recall from \ref{enum: exch indep multi marked} that
\begin{align*}
    \xi_k^j \indep_{P'} \{ \xi_k^h \}_{h \neq j} \cup \{\xi_k^{\setminus J}, \eta_k^{\setminus J}\} \cup \{ \xi_m^h, \xi_m^{\setminus J}, \eta_m^h, \eta_m^{\setminus J}, \tilde \xi_m, \tilde \eta_m \}_{h \in J, m > k} \big| \F'_{T_{k-1}'}. 
\end{align*}
We now use that $\{\tau'^J > T_{k-1}'\} \in \F'_{T_{k-1}'}$ (apply \cite[Theorem 2.1.16 (i)]{LastBrandt1995marked} to $I(\tau'^J > \cdot)$). Since a conditional independence with respect to a $\sigma$-algebra $\A$ implies a corresponding conditional independence with respect to $\A,A$, the $\sigma$-algebra $\A$ mixed with an event $A \in \A$, the preceding independence implies, since $\F'_{T_{k-1}'} = \sigma(L', N'|_{T_{k-1}'})$,
\begin{equation}
\resizebox{0.95\textwidth}{!}{$\displaystyle
\begin{aligned}
 \xi_k^j \indep_{P'} \{ \xi_k^h \}_{h \neq j} \cup \{\xi_k^{\setminus J}, \eta_k^{\setminus J}\} \cup \{ \xi_m^h, \xi_m^{\setminus J}, \eta_m^h, \eta_m^{\setminus J}, \tilde \xi_m, \tilde \eta_m \}_{h \in J, m > k} \Big| \sigma(L', N'|_{T_{k-1}'}), \tau'^J > T_{k-1}'.
\end{aligned}
$}
\label{eq: xi indep k}
\end{equation}
The result \eqref{eq: Tkj tilde N temp independence} now follows from \eqref{eq: xi indep k}. This is because conditional independence statements are preserved under measurable transformations: The random variables in \eqref{eq: Tkj tilde N temp independence} are transformations of random variables in \eqref{eq: xi indep k}, as can be seen by inspecting Algorithm \ref{alg:scheme multi marked}. In particular, $T_k^j$ is a function of $(L', N'|_{T_{k-1}'}, \xi_k^j)$ while $(T_k^{h}, T_k^{\setminus J}, ~_{T_{k-1}'}\tilde N')_{h \neq j}$ are functions as described in \eqref{eq: variables functions of}:
\begin{align}
    \begin{split}
        &\text{On } \{\tau'^J > T_{k-1}'\}, \text{ and thus on }C_k  \text{ (by the established consistency }  \\
        &\text{condition } \ref{enum: prop consistency marked}), \text{ the variables } (T_k^{h}, T_k^{\setminus J}, ~_{T_{k-1}'}\tilde N')_{h \neq j} \text{ are functions of } \\
    &(L', N'|_{T_{k-1}'}, \{ \xi_k^h \}_{h \neq j} \cup \{\xi_k^{\setminus J}, \eta_k^{\setminus J}\} \cup 
\{ \xi_m^h, \xi_m^{\setminus J}, \eta_m^h, \eta_m^{\setminus J}, \tilde \xi_m, \tilde \eta_m \}_{h \in J, m > k}).
    \end{split}\label{eq: variables functions of}
\end{align}
Using these facts, we can deduce \eqref{eq: Tkj tilde N temp independence} from \eqref{eq: xi indep k} using a monotone class argument.

To show \eqref{eq: Vkj Tkj cond indep}, it suffices (analogously to \eqref{eq: Tkj tilde N temp independence}) to verify the independence
\begin{align}
    V_k^{j} \indep_{P'} (T_k^{h}, T_k^{\setminus J}, ~_{T_{k-1}'}\tilde N')_{h \neq j} |\F'_{T_{k-1}'} \vee \sigma(T_k^{j}), \tau'^J > T_{k-1}'. \label{eq: Vkj tilde N temp independence}
\end{align}
Because $\{\tau'^J > T_{k-1}'\} \in \F'_{T_{k-1}'}$, conditioning on this event in \ref{enum: exch indep multi marked} preserves the conditional independence. We get, by the same arguments as those used to deduce \eqref{eq: xi indep k} from \ref{enum: exch indep multi marked}, the independence
\begin{equation}
\resizebox{\textwidth}{!}{$\displaystyle
\begin{aligned}
\eta_k^j \indep_{P'} \{ \xi_k^h \}_{h \neq j} \cup \{\xi_k^{\setminus J}, \eta_k^{\setminus J}\} \cup \{ \xi_m^h, \xi_m^{\setminus J}, \eta_m^h, \eta_m^{\setminus J}, \tilde \xi_m, \tilde \eta_m \}_{h \in J, m > k} \big| \sigma(L', N'|_{T_{k-1}'}, T_k^j), \tau'^J > T_{k-1}'.
\end{aligned}
$}
\label{eq: temp indep}
\end{equation}
The result \eqref{eq: Vkj tilde N temp independence} follows from \eqref{eq: temp indep}, because conditional independence is preserved by measurable transformations. In particular, inspecting Algorithm \ref{alg:scheme multi marked}, we see that $V_k^j$ is a function of $(L', N'|_{T_{k-1}'}, T_k^j, \eta_k^j)$, and $(T_k^{h}, T_k^{\setminus J}, ~_{T_{k-1}'}\tilde N')_{h \neq j}$ are as described in \eqref{eq: variables functions of}. 
Using that $\F'_{T_{k-1}'} = \sigma(L', N'|_{T_{k-1}'})$, we can conclude that \eqref{eq: Vkj tilde N temp independence} holds, where this deduction can formally be made using a monotone class argument. We conclude that \eqref{eq: Vkj Tkj cond indep} holds.




We now show that the condition \ref{enum: prop exchangeability marked} for the $j$th component is ensured by \eqref{eq: Tkj cond indep}-\eqref{eq: Vkj Tkj cond indep}. Toward this end we write $\mu$ and $\nu$ for the $(P', \F_t')$ and $(P', \F_t' \vee \sigma(\tilde N'))$-compensators of $N'^j$, respectively. To establish the condition, it suffices to show that
$$
\mu(dt \times dx) I(\tau'^J \geq t) = \nu(dt \times dx) I(\tau'^J \geq t) \quad P'\text{-a.s.}
$$
The preceding equality implies that $~^{\tau'^J}\mu$ is a compensator of $~^{\tau'^J}N'^j$ with respect to both $\{\F_t'\}_t$ and $\{\F_t' \vee \sigma(\tilde N') \}_t$ under $P'$. Since $\alpha^j$ is the canonical compensator of $N'^j$ under $P'$, we have $\mu(dt \times dx) = \alpha^j(L', N', dt \times dx)$ $P'$-a.s., and the preceding equality implies that \ref{enum: prop exchangeability marked} holds for the $j$th component.

By \cite[Theorem 4.1.11]{LastBrandt1995marked} we have for each $k \geq 1$ that, $P'$-a.s. on $\{ T_{k-1}' < t \leq T_k' \}$, 
\begin{align}
    \mu(dt \times dx)  &= \frac{P'(T_k' \in dt, X_k' \in dx \cap X^j |  \F'_{T_{k-1}'})}{P'(T_k' \geq t |  \F'_{T_{k-1}'})}, \label{eq:LambdaRepTwo multi marked} \\
    \nu(dt \times dx) &= \frac{P'(T_k' \in dt, X_k' \in dx \cap X^j |  \F'_{T_{k-1}'} \vee \sigma(\tilde N'))}{P'(T_k' \geq t |  \F'_{T_{k-1}'} \vee \sigma(\tilde N'))}. \label{eq:GammaRepTwo multi marked}
\end{align}
Using that $I(T_k' \in dt, X_k' \in dx \cap X^j) = I(T_k^{j} \in dt, \wedge_{h \in J \setminus \{j\}} T_k^{h} \geq t, T_k^{\setminus J} \geq t, V_k^{j} \in dx \cap X^j)$ $P'$-a.s. (by Algorithm \ref{alg:scheme multi marked} and Lemma \ref{lemma: no shared mass II marked}), we have
\begin{align*}
P'(T_k' \in dt, & X_k' \in dx \cap X^j |  \F'_{T_{k-1}'}) \\
&= P'(T_k^{j} \in dt, \wedge_{h \in J \setminus \{j\}} T_k^{h} \geq t, T_k^{\setminus J} \geq t, V_k^{j} \in dx \cap X^j |  \F'_{T_{k-1}'}) \\
&= P'(V_k^{j} \in dx \cap X^j |  \F'_{T_{k-1}'}, T_k^{j} = t) P'(T_k^{j} \in dt |  \F'_{T_{k-1}'}) \\
&\quad \times P'(\wedge_{h \in J \setminus \{j\}} T_k^{h} \geq t |  \F'_{T_{k-1}'}) P'(T_k^{\setminus J} \geq t |  \F'_{T_{k-1}'}),
\end{align*}
where the last equality follows from \ref{enum: D1} and \ref{enum: D2}.

From \eqref{eq: Tkj cond indep} and the decomposition and weak union properties of conditional independence, we deduce the conditional independencies
\begin{align}
    T_k^{j} &\indep_{P'} (T_k^{h}, T_k^{\setminus J})_{h \in J \setminus \{j\}} |\F'_{T_{k-1}'} \vee \sigma(\tilde N'), \tau'^J > T_{k-1}', \label{eq: Tkj cond indep2} \\
    T_k^{j} &\indep_{P'} \tilde N' |\F'_{T_{k-1}'}, \tau'^J > T_{k-1}'. \label{eq: Tkj cond indep tilde N}
\end{align}
Using these independencies we get, $P'$-a.s. on $\{\tau'^J > T_{k-1}'\}$,
\begin{align*}
&P'(T_k' \in dt, X_k' \in dx \cap X^j |  \F'_{T_{k-1}'} \vee \sigma(\tilde N'))  \\
&= P'(T_k^{j} \in dt, \wedge_{h \in J \setminus \{j\}} T_k^{h} \geq t, T_k^{\setminus J} \geq t, V_k^{j} \in dx \cap X^j |  \F'_{T_{k-1}'} \vee \sigma(\tilde N'))  \\
&= P'(V_k^{j} \in dx \cap X^j |  \F'_{T_{k-1}'} \vee \sigma(\tilde N'), T_k^{j} = t, \wedge_{h \in J \setminus \{j\}} T_k^{h} \geq t, T_k^{\setminus J} \geq t)  \\
&\quad \times P'(T_k^{j} \in dt ,\wedge_{h \in J \setminus \{j\}} T_k^{h} \geq t, T_k^{\setminus J} \geq t |  \F'_{T_{k-1}'} \vee \sigma(\tilde N'))  \\
&= P'(V_k^{j} \in dx \cap X^j |  \F'_{T_{k-1}'}, T_k^{j} = t)  \\
&\quad \times P'(T_k^{j} \in dt |  \F'_{T_{k-1}'}) P'(\wedge_{h \in J \setminus \{j\}} T_k^{h} \geq t, T_k^{\setminus J} \geq t |  \F'_{T_{k-1}'} \vee \sigma(\tilde N')),
\end{align*}
where we used \eqref{eq: Vkj Tkj cond indep}, \eqref{eq: Tkj cond indep2}, and \eqref{eq: Tkj cond indep tilde N} in the last line. 

For the denominators on the right-hand side of \eqref{eq:LambdaRepTwo multi marked}-\eqref{eq:GammaRepTwo multi marked} we similarly get, by \ref{enum: D1} and \eqref{eq: Tkj cond indep2}-\eqref{eq: Tkj cond indep tilde N}, $P'$-a.s.,
\begin{align*}
    &P'(T_k' \geq t |  \F'_{T_{k-1}'}) = P'(T_k^{j} \geq t |  \F'_{T_{k-1}'}) P'(\wedge_{h \in J \setminus \{j\}} T_k^{h} \geq t |  \F'_{T_{k-1}'}) P'(T_k^{\setminus J} \geq t |  \F'_{T_{k-1}'}) \\ 
    &P'(T_k' \geq t |  \F'_{T_{k-1}'} \vee \sigma(\tilde N')) I(\tau'^J > T_{k-1}') \\
    &= P'(T_k^{j} \geq t |  \F'_{T_{k-1}'})  P'(\wedge_{h \in J \setminus \{j\}} T_k^{h} \geq t, T_k^{\setminus J} \geq t |  \F'_{T_{k-1}'} \vee \sigma(\tilde N') )I(\tau'^J > T_{k-1}').
\end{align*}
Substituting the preceding expressions into \eqref{eq:LambdaRepTwo multi marked} and \eqref{eq:GammaRepTwo multi marked}, we obtain for each $k \geq 1$ that 
\begin{align}
\begin{split}
    &\mu(dt \times dx) I(T_{k-1}' < t \leq T_k', \tau'^J > T_{k-1}') \\
&= \frac{P'(V_k^{j} \in dx \cap X^j |  \F'_{T_{k-1}'}, T_k^{j} = t) P'(T_k^{j} \in dt |  \F'_{T_{k-1}'})}{P'(T_k^{j} \geq t |  \F'_{T_{k-1}'})} I(T_{k-1}' < t \leq T_k', \tau'^J > T_{k-1}') \\
&= \nu(dt \times dx) I(T_{k-1}' < t \leq T_k', \tau'^J > T_{k-1}'),
\end{split}
\label{eq: split Gk hazard multi marked}
\end{align}
since the terms $P'(\wedge_{h \in J \setminus \{j\}} T_k^{h} \geq t |  \F'_{T_{k-1}'}) P'(T_k^{\setminus J} \geq t |  \F'_{T_{k-1}'})$ and $P'(\wedge_{h \in J \setminus \{j\}} T_k^{h} \geq t, T_k^{\setminus J} \geq t |  \F'_{T_{k-1}'} \vee \sigma(\tilde N') )$ cancel out from the respective numerator and denominator. We thus get from \eqref{eq: split Gk hazard multi marked} by summing over $k$ that, $P'$-a.s.,
\begin{align*}
    \mu(dt \times dx) I(\tau'^J \geq t)  &=  \sum_{k \geq 1}\mu(dt \times dx) I(T_{k-1}' < t \leq T_k') I(\tau'^J \geq t) \\
    &= \sum_{k \geq 1}\nu(dt \times dx) I(T_{k-1}' < t \leq T_k') I(\tau'^J \geq t) \\
    &= \nu(dt \times dx) I(\tau'^J \geq t).
\end{align*}
\end{proof}

\section{Characterizing potential outcomes in the general MPP setting, allowing for unmeasured components}
\label{appendix: general MPP potential outcomes}

This appendix characterizes potential outcome processes under multiple interventions in the general MPP setting. The construction tackles the situation with both observed and unmeasured components. We also allow here for interventions that are optional, but not predictable. 

We let $B$ be a set indexing all MPP components (both observed and unmeasured), and write $N = (N^j)_{j \in B}$ for the full MPP. Write $\mathscr{X}$ for the mark set of $N$. In addition to $N$ we also consider baseline random elements $L, U$. The observed components consist of a subset of these elements: the baseline variable $L$, and the processes $N^j$ for $j \in \mathcal{I}_d \subset B$. The unmeasured, or 'hidden', components are $U$ and components of $N$ are indexed by $B \setminus \mathcal{I}_d$. This set-up generalizes the one in Section \ref{subsec: setup multiple interventions}, which is recovered when $B = \mathcal{I}_d$ and $U$ is omitted (i.e., when there are no unmeasured components).

We consider interventions on a subset $J \subseteq \I_d$ of the observed components. Each intervention $\mathfrak{n}^j$ is a counting measure on the canonical space $(S \times \mathcal{U} \times \N_T^{\mathscr{X}}, \S \otimes \mathscr{U} \otimes \H_T^{\mathscr{X}})$ satisfying
$$
\mathfrak{n}^j(l, u, \varphi, dt \times dx) = \mathfrak{n}^j(l, u, \varphi|_{t}, dt \times dx)
$$
for all $(l, u, \varphi) \in S \times \mathcal{U} \times \N_T^{\mathscr{X}}$ and $t \in \T$. Following \cite{LastBrandt1995marked}, we call a kernel with this property \textit{optional}. We assume that $(\N_T^{\mathscr{X}}, \H_T^{\mathscr{X}})$ is Borel to ensure the existence of certain disintegrations. To reflect that an intervention can only depend on observed components, we take $\mathfrak{n}^j$ to be constant in $u$ and in any unobserved components of $\varphi$. We use the notation $\mathscr{N}^j := \mathfrak{n}^j(L, U, N, \cdot)$ for the intervention $\mathfrak n^j$ evaluated in the observed (and unmeasured) trajectories.

Let $\alpha = (\alpha^j)_{j \in B}$ denote the canonical compensator of $N = (N^j)_{j \in B}$, and $\dot{\mathscr{L}}^j$ the canonical compensator of $\mathscr{N}^j$ for $j \in J$ with respect to $P$ and the filtration generated by $(L, U, N)$. $\dot{\mathscr{L}}^j$ is a dual predictable projection on the canonical space, and exists by analogous arguments as those shown in the proof of \cite[Theorem 4.2.2]{LastBrandt1995marked}. These canonical compensators are functionals on the path space of $(L,U,N)$. In particular,  $\mathscr{L}^j := \dot{\mathscr{L}}^j(L,U,N)$ defines a compensator of $\mathscr{N}^j$ with respect to the filtration generated by $(L,U, N)$.

We characterize the potential outcome process via its canonical compensator. For the characterization to be meaningful, we adopt regularity conditions analogous to those in Section \ref{subsection: potential outcome processes} to ensure this canonical compensator is well-defined (see \cite[Theorem 4.2.2 and Section 4.3]{LastBrandt1995marked}):
\begin{align}
    & \Big\{ \sum_{j \in J} d\bar{\dot{\mathscr{L}}}^j_t(l,u,\varphi) + \sum_{j \in B \setminus J} d\bar{\alpha}^j_t(l,u,\varphi) \Big\} I(\pi_\infty'(l,u,\varphi) < t) = 0, \label{eq:app-explosion} \\
    & \sum_{j \in J} \Delta\bar{\dot{\mathscr{L}}}^j_t(l,u,\varphi) + \sum_{j \in B \setminus J} \Delta\bar{\alpha}^j_t(l,u,\varphi) \leq 1, \label{eq:app-jump}
\end{align}
where $\pi_\infty'(l,u,\varphi) = \inf\big\{ s > 0 | \sum_{j \in J} \bar{\dot{\mathscr{L}}}^j_s(l,u,\varphi) + \sum_{j \in B \setminus J} \bar{\alpha}^j_s(l,u,\varphi) = \infty \big\}$, and we recall the convention that for a random measure $\mu$ operating on a mark set $X^\mu$, $\bar{\mu}_t = \mu((0,t] \times X^\mu)$.

Write $\tilde{\F}_\T$ for the filtration generated by $(L,U,\tilde{N})$. The potential outcome process $\tilde{N}$ under the joint intervention on components indexed by $J \subseteq \I_d$ is characterized through its (canonical) compensator as follows:
\begin{align}
\begin{split}
    &\dot{\mathscr{L}}^j(L,U, \tilde{N}) \text{ defines a } (P, \tilde{\F}_\T)\text{-compensator of } \tilde{N}^j \text{ for } j \in J, \\
    &\alpha^j(L,U, \tilde{N}) \text{ defines a } (P, \tilde{\F}_T)\text{-compensator of } \tilde{N}^j \text{ for } j \in B \setminus J.
\end{split}
\label{eq: unmeasured pot outcomes compensator}
\end{align}
If $\mathfrak n^j$ is predictable, then, analogously to Definition \ref{definition: potential outcome process}, we have $\dot{\mathscr{L}}^j(L,U, \tilde{N}) = \mathfrak n^j(L,U, \tilde N)$ $P$-a.s. Furthermore, by similar arguments to those shown in Appendix \ref{subsection: potential outcome process for predictable regime}, it follows from \eqref{eq: unmeasured pot outcomes compensator} that $\tilde{N}^j = \mathfrak n^j(L,U, \tilde N)$ $P$-a.s., analogously to \eqref{eq: predictable intervention tilde a compensator}.


\section{Supporting lemmas and technical results}
\label{appendix: supporting calculations}


\begin{lemma}[Conditional expectation under independence]\label{lemma: williams independence}
(\cite[Theorem 9.10]{williams1991probability}) Let $\mathcal{G}$ be a $\sigma$-algebra, $X$ a $\mathcal{G}$-measurable random variable, and $Z$ independent of $\mathcal{G}$. Then, for any bounded measurable function $f$
$$
E\big[f(X,Z) |  \mathcal{G}\big] = g(X) \quad \text{a.s.},
$$
where $g(x) = E[f(x,Z)]$.
\end{lemma}

\begin{lemma}\label{lemma: mathbb representations}
    Recall the processes $\mathbb{N}^a$, $\mathbb{\Lambda}^a$, and $\mathbb{M}^a$ introduced in Section \ref{subsection: identifying conditions}. 
     We have 
        \begin{align}
            \mbbNa{} &= ~^{\Nastop}N^a + ~^{\Nastop} \na{}(N) - 2~^{\Nastop}[N^a, \na{}(N)], \label{eq: mathbb N a explicitly} \\
            \mathbb{\Lambda}^a &=  
            ~^{\tau^a}\Lambda^a  + ~^{\tau^a} \mathfrak {n}^a(N) - 2 \int_0^{\tau^a \wedge \cdot} \Delta \mathfrak {n}^a_s(N)  d\Lambda_s^a \quad P\text{-a.s.} \label{eq: mathbb Lambda a representation} \\
            \mathbb{M}^a &= \int_0^{\tau^a \wedge \cdot} \big(1 - 2 \Delta \mathfrak {n}^a_s(N) \big) dM_s^a  \quad P \text{-a.s.} \label{eq: mathbb M a representation}
        \end{align}
        Moreover, the exchangeability condition \eqref{eq: exchan disc} is equivalent to
        \begin{align}
            \begin{split}
            ~^{\tau^a}\Lambda^a  \text{ defines a compensator of }& ~^{\tau^a}N^a \\
            \text{ with respect to both } \F_\T \text{ and } \F_\T^{\tilde Y} & \text{ under } P.
            \end{split}\label{eq: stopped compensator exchangeability}
        \end{align}
\end{lemma}

\begin{proof}
We first highlight that $N^a - [N^a, \na{}(N)]$ counts the jumps of $N^a$ which are not shared with jumps of $\na{}(N)$, and similarly that $\na{}(N) - [N^a, \na{}(N)]$ counts the jumps of $\na{}(N)$ which are not shared with jumps of $N^a$. The identity \eqref{eq: mathbb N a explicitly} reflects that $\mbbNa{t}$ counts whether there are jumps of $N^a$ which are not shared with jumps of $\na{}(N)$, or jumps of $\na{}(N)$ which are not shared with jumps of $N^a$, up to time $t$. We also have
    \begin{align*}
        ~^{\Nastop}[N^a, \na{}(N)] = \int_0^{\Nastop \wedge \cdot} \Delta \na{s}(N) dN_s^a.
    \end{align*}
    Since $\na{}(N)$ is $\F_\T$-predictable and $\Lambda^a$ defines an $\F_\T$-compensator of $N^a$, it follows that $\int_0^{\Nastop \wedge \cdot} \allowbreak \Delta \na{s}(N) d\Lambda_s^a$ defines a compensator of $~^{\Nastop}[N^a, \na{}(N)]$ with respect to $P$ and $\F_\T$ (see, e.g., \cite[Theorem 8.2.9]{cohen2015stochastic}). Next, as $\na{}(N)$ is an $\F_\T$-predictable counting process, $~^{\Nastop}\na{}(N)$ defines a compensator of itself with respect to $P$ and $\F_\T$. By combining these facts with \eqref{eq: mathbb N a explicitly} we get the results \eqref{eq: mathbb Lambda a representation} and \eqref{eq: mathbb M a representation} by the uniqueness of compensators. 
 
Assume now the exchangeability condition \eqref{eq: exchan disc} holds. The right-hand side of \eqref{eq: mathbb M a representation} is a martingale with respect to the two filtrations $\F_\T$ and $\F_\T^{\tilde Y}$, because $\mathbb M^a$ is a martingale with respect to both these filtrations and $\mathfrak n^a(N)$ is $\F_\T$-predictable. By integrating with respect to bounded and predictable functions, it follows that $~^{\tau^a}M^a$ is a martingale with respect to both filtrations. Since $~^{\tau^a}\Lambda^a$ is predictable with respect to both filtrations, \eqref{eq: stopped compensator exchangeability} follows.

Conversely, if $~^{\tau^a}\Lambda^a$ is the compensator of $~^{\tau^a}N^a$ with respect to both $\F_\T$ and $\F_\T^{\tilde Y}$, then $~^{\tau^a}M^a$ is a local martingale with respect to both filtrations. Using the representation $\mathbb{M}^a = \int_0^{\tau^a \wedge \cdot} (1 - 2\Delta \mathfrak {n}^a_s(N)) dM_s^a$ from \eqref{eq: mathbb M a representation}, we get by integrating with respect to bounded and predictable functions that $\mathbb{M}^a$ is a martingale with respect to both filtrations. Since $\mathbb{\Lambda}^a$ is $\F_\T$-predictable, the condition \eqref{eq: exchan disc} follows.
\end{proof}

\begin{lemma}\label{lemma: tau a is infty under Q}
    If $W$ in \eqref{eq: W dolean} is a likelihood ratio process, then $\Nastop = \infty$ $Q$-a.s., where $Q$ is the induced measure $dQ = W_T dP$ on $\F_T$.
\end{lemma}

\begin{proof}
 The event $\{ \Nastop < \infty \}$ is contained in $\F_\Nastop$. Since $W$ is a uniformly integrable martingale we get
    \begin{align*}
        Q(\Nastop < \infty) = E_P[ W_T I(\Nastop < \infty) ] =  E_P[ W_\Nastop I(\Nastop < \infty) ] = 0,
    \end{align*}
    since $W_\Nastop = 0$ $P$-a.s. on $\{ \Nastop < \infty \}$ by Lemma \ref{lemma: W finite variation}.
\end{proof}


\begin{lemma}\label{lemma: W finite variation}
    Condition \eqref{eq: K jump to zero condition} (Positivity) implies the following:
\begin{enumerate}[label=\roman*)]
    \item $\{ \frac{1}{1 - \Delta \mbbLa{s}} \}_s$ is locally bounded. \label{item: mbbLa recipocal is locally bounded}
    \item $\KK$ in \eqref{eq: K semimart} is a local  square integrable martingale of finite variation with respect to $P$ and $\F_\T$. \label{enum: K local martingale}
    \item $W = \E(\KK)$ in \eqref{eq: W dolean} is of finite variation, and is a nonnegative local martingale with respect to $P$ and $\F_\T$. \label{enum: W1}
    \item $\Nastop$ in \eqref{eq: tau a} coincides with the optional times $\tau^{a\prime} = \inf\{ s > 0 | \Delta \mathbb K^a_s = -1 \}$ and $\tau^{a\prime\prime} = \inf \{ s > 0 | W_s = 0 \}$. \label{enum: W2}
    \item $I(W_\cdot > 0) = I(\cdot < \Nastop)$ $P$-a.s.  \label{enum: W jump to zero identity}
    \item $ \langle M^j , W \rangle^P$ and $ \langle M^j , \KK \rangle^P$ exist with respect to $\F_\T$ for each $j \in \I_d$. 
    Furthermore $\langle M^j , \KK \rangle^P \allowbreak = - \int_0^\cdot \allowbreak \frac{1}{1 - \Delta \mbbLa{s}} 
 \allowbreak d\langle M^j, \mbbMa{} \rangle_s^P$ and $\langle M^j , W \rangle^P = \allowbreak\int_0^\cdot W_{s-} \allowbreak d\langle M^j, \KK \rangle^P_s$.
    \label{enum: W3}
    \medskip
    \item[] If the exchangeability condition in Definition \ref{def:strong prime} \ref{item: strong cf exch} also holds, then:
    \item \ref{enum: K local martingale} and \ref{enum: W1} also holds with respect to $P$ and $\F_\T^{\tilde Y}$. \label{enum: K exch martingale}
\end{enumerate}
\end{lemma}
\begin{proof}

    \ref{item: mbbLa recipocal is locally bounded}:
    
    The process $\xi := \int_0^\cdot \frac{d\mbbLa{s}}{1 - \Delta \mbbLa{s}}$ is cadlag and predictable. This can be seen by decomposing the integral as
    $$ \xi = \mathbb \Lambda^{a, c} + \sum_{0 < s \leq \cdot} \frac{\Delta \mbbLa{s}}{1 - \Delta \mbbLa{s}}. $$
    where $\mathbb \Lambda^{a,c}(dt) = I(\Delta \mbbLa{s} = 0) \mbbLa{}(dt)$ is the continuous part of $\mbbLa{}$. The process on the right-hand side of the previous line is cadlag, and predictable as it is a sum of predictable terms. 
    
    Under \eqref{eq: K jump to zero condition}, $\xi$ takes values in $\mathbb{R}$ on $\T$, i.e. $\xi$ does not explode to  $\infty$ on $\T$. If the condition were not satisfied, $\xi$ would evidently fail to be locally bounded. However, since $\xi$ is cadlag, predictable, and real-valued, it is locally bounded by \cite[Lemma 7.3.20]{cohen2015stochastic}.
    
    We can thus pick a localizing sequence $\{ \sigma_n \}_n$, i.e. $\sigma_n$ are optional times increasing up to $\infty$ so that $I(\sigma_n > 0)\xi_{\sigma_n \wedge \cdot} \leq c_n$ for some constant $c_n$ for each $n$. It certainly follows that
    $$ I(\sigma_n > 0)\Delta \xi_{\sigma_n \wedge \cdot} = I(\sigma_n > 0) \frac{\Delta \mbbLa{\sigma_n \wedge \cdot}}{1 - \Delta \mbbLa{\sigma_n \wedge \cdot}} \leq c_n. $$
  for each $n$. We can conclude from this that 
\begin{align*}
	I(\sigma_n > 0) \frac{1}{1 - \Delta \mbbLa{\sigma_n \wedge \cdot}} \leq c_n + 1,
\end{align*}
for each $n \geq 1$. That is, $\{ \frac{1}{1 - \Delta \mbbLa{s}} \}_s$ is locally bounded.

    \ref{enum: K local martingale}:

Since $\mbbMa{}$ is a martingale, we have by \ref{item: mbbLa recipocal is locally bounded} and \cite[IV 2 Theorem 29]{protter} that $\KK = - \int_0^\cdot \frac{d\mbbMa{s}}{1 - \Delta \mbbLa{s}}$ is a local martingale. It is of finite variation since the integrator is of finite variation and the integrand is locally bounded.

We show that it is also a local square integrable martingale. $\mbbMa{}$ is the difference between a one-jump counting process $\mbbNa{}$ and its compensator $\mbbLa{}$, a cadlag and predictable process. $\mbbNa{}$ is automatically locally bounded, while $\mbbLa{}$ is locally bounded by \cite[Lemma 7.3.20]{cohen2015stochastic}, making $\mbbMa{}$ a locally bounded martingale. It is therefore in particular a local square integrable martingale. Moreover, by \cite[II 6 Corollary 2]{protter}, there is a local martingale $\tilde M$ such that $(\mbbMa{})^2 = \tilde M + [\mbbMa{}]$. Since $\mbbMa{}$ is a local square integrable martingale, it follows that $[\mbbMa{}]$ is locally integrable.

By definition of the quadratic variation (see Appendix \ref{appendix: Basic definitions}) and \cite[IV 2 Theorem 22]{protter}, we get 
\begin{align}
    \KK^2 = 2 \int_0^\cdot \mathbb K^a_{s-} d\mathbb K^a_s + \int_0^\cdot \frac{1}{(1 - \Delta \mbbLa{s})^2} d[\mbbMa{}]_s. \label{eq: KK local sq int}
\end{align}
The integrand in the first term of the right-hand side of \eqref{eq: KK local sq int} is the left-limit of a cadlag and adapted process. Thus, it is predictable, and $S_n = \inf\{ s > 0 | |\mathbb K^a_s| \geq n \}$ defines a sequence of optional times increasing up to $\infty$ making $\mathbb K^a_-$ locally bounded. Moreover, the integrator is a local martingale. Combining localizing sequences, the first term on the right-hand side of \eqref{eq: KK local sq int} is a local martingale by \cite[IV 2 Theorem 29]{protter}. The integrand in the second term on the right-hand side of \eqref{eq: KK local sq int} is locally bounded by \ref{item: mbbLa recipocal is locally bounded}, and the integrator is locally integrable. Combining localizing sequences, stopping, taking expectations, and suprema, we conclude that $\KK$ is a local square integrable martingale. 
    
\ref{enum: W1}:

     Equation \eqref{eq: W dolean} implies that
    \begin{align*}
        V_{[0,t]}(W) \leq 1 + \int_0^t V_{[0,s)}(W) d V_{[0,s]}(\KK),
    \end{align*}
    where $V_{[0,t]}(Z)$ is the variation of a process $Z$ up to $t$. Since $\KK$ is of finite variation, it follows from Grönwall's inequality, e.g. \cite[Lemma 15.1.6]{cohen2015stochastic}, that $W$ is of finite variation.

    Next, since $0 \leq \Delta \mbbLa{} \leq 1$, $\Delta \KK = -\frac{\Delta \mbbNa{} - \Delta \mbbLa{}}{1 - \Delta \mbbLa{}} \geq -1$ is immediate, where the only ambiguity is the case "$\frac{0}{0}$". The proof of \ref{enum: W2} shows that this ambiguity occurs on an exceptional set; thus $\Delta \KK \geq -1$ $P$-a.s. Since $\KK$ is a local martingale and $W = \E(\KK)$, we get from \cite[Lemma 15.3.2.]{cohen2015stochastic} that $W$ is a nonnegative local martingale with respect to $P$ and $\F_\T$.

\ref{enum: W2}:

    The equality $\Delta \KK = -\frac{\Delta \mbbNa{} - \Delta \mbbLa{}}{1 - \Delta \mbbLa{}} = -1$ holds only when $\Delta \mathbb N^a = 1$, provided that $\Delta \mbbLa{} \neq 1$. Under condition \eqref{eq: K jump to zero condition}, we have  in particular have that $\Delta \mbbLa{\Nastop} \neq 1$ $P$-a.s. on $\{ \Nastop < \infty \}$. It follows that $\Delta \mathbb K^a_{\Nastop} = -1$ $P$-a.s. on $ \{ \Nastop < \infty \}$, and therefore that $\Nastop = \inf\{ s > 0 | \Delta \mathbb K^a_s = -1 \}$ $P$-a.s. The other identity can be seen from the solution of the stochastic exponential; see \cite[II 8 Theorem 37]{protter}.

\ref{enum: W jump to zero identity}:

Inspecting the solution of the stochastic exponential, \cite[II 8 Theorem 37]{protter}, this follows from \ref{enum: W2}.

    \ref{enum: W3}:

    $\KK$ is a local square integrable martingale by \ref{enum: K local martingale}. Moreover, under Assumption \ref{assumption: no explosions}, $M^j$, is a finite variation local square integrable martingale; see \cite[Chapter II]{Andersen}. Thus, $[\KK, M^j]$ is via the Kunita-Watanabe inequality, \cite[II 6 Theorem 25]{protter}, locally integrable, and its compensator with respect to $P$ and $\F_\T$, $\langle \KK, M^j \rangle^P$, exists. Moreover, by \cite[IV 2 Theorem 22]{protter} we have $[\KK, M^j] = -\int_0^\cdot \frac{d[\mbbMa{},M^j]_s}{1 - \Delta \mbbLa{s}}$. Since the integrand of this equation is locally bounded (by \ref{item: mbbLa recipocal is locally bounded}) and predictable, and the compensator of $[\mbbMa{}, M^j]$ is $\langle \mbbMa{}, M^j \rangle^P$, it follows that $\langle \KK, M^j \rangle^P = -\int_0^\cdot \frac{d\langle \mbbMa{},M^j \rangle_s^P}{1 - \Delta \mbbLa{s}}$ by e.g. \cite[Theorem 8.2.9]{cohen2015stochastic}. Next,
    \begin{align}
        [W, M^j] = \int_0^\cdot W_{s-} d[\KK, M^j]_s \label{eq: W Mj bracket}
    \end{align}
    by \eqref{eq: W dolean} and \cite[IV 2 Theorem 22]{protter}. Since $W$ is a cadlag and adapted process of finite variation, $W_-$ is predictable, and locally bounded (see the proof of \ref{enum: K local martingale}). Since $\KK$ and $M^j$ are local square integrable martingales and $W_-$ is locally bounded, it follows from \eqref{eq: W Mj bracket} that $[W, M^j]$ is of locally integrable variation by combining localizing sequences. This means that $\langle W, M^j \rangle^P$ exists. Since $W_-$ is predictable and locally bounded, it follows from \eqref{eq: W Mj bracket} that $\langle W, M^j \rangle^P = \int_0^\cdot W_{s-} d\langle \KK, M^j \rangle_s^P$.

    \ref{enum: K exch martingale}:

    Since $\mbbLa{}$ also defines a $(P, \F_\T^{\tilde Y})$-compensator of $\mbbNa{}$ under exchangeability, the result follows by applying the arguments shown in the proofs of \ref{enum: K local martingale} and \ref{enum: W1} to $\F_\T^{\tilde Y}$ instead of $\F_\T$.
\end{proof}

By Lemma \ref{lemma: W finite variation}, $W$ in \eqref{eq: W dolean} is a nonnegative local $P$-martingale with respect to $\F_\T$, and if exchangeability holds, this is also true with respect to the filtration $\F_\T^{\tilde Y}$. In the following lemma we use the general notation $\Z_\T$ for a filtration on $\T$, where we take  $\Z_\T = \F_\T^{\tilde Y}$ if exchangeability holds and $\Z_\T = \F_\T$ otherwise.

\begin{lemma}\label{lemma: positivity and likelihood ratio equivalence}
    Under \eqref{eq: K jump to zero condition}, $W$ in \eqref{eq: W dolean} is a nonnegative mean one uniformly integrable martingale with respect to $P$ and $\Z_\T$ on $\T$ if and only if Definition \ref{def:strong prime} \ref{item: likelihood ratio regularity} holds. 
\end{lemma}
\begin{proof}
    $W$ is a nonnegative local martingale under $P$ and $\Z_\T$ by Lemma \ref{lemma: W finite variation}, where we can take $\Z_\T = \F_\T^{\tilde Y}$ under exchangeability and $\Z_\T = \F_\T$ in general. \cite[Lemma 15.3.2.]{cohen2015stochastic} says that $W = \E(\KK)$ is a martingale on $\T$ with respect to $P$ and $\Z_\T$ if and only if \eqref{eq: W likelihood ratio regularity} holds. Thus, under \eqref{eq: W likelihood ratio regularity}, $W$ is a nonnegative mean one martingale. It is uniformly integrable if and only if 
    $$ \lim_{b \rightarrow \infty} \sup_{t \in \T} E_P[I(W_t > b)W_t] = 0.$$
    A uniformly integrable martingale is certainly a martingale. Moreover, a martingale on the compact time domain $\T$ is uniformly integrable. This can be seen, because $W_T$ is (by \eqref{eq: W likelihood ratio regularity}) integrable, and $W_t = E_P[W_T | \F_t]$ by the martingale property. The collection $\{ E_P[Z | \G_t] \}_t$ is always uniformly integrable for an integrable random variable $Z$ and a filtration $\{\G_t\}_t$; see e.g. \cite[Chapter 13.4]{williams1991probability}.
\end{proof}

\begin{lemma}\label{lemma: MPP exchangeability 2026}
With the set-up and notation in Section \ref{section: identification with multiple interventions in the general MPP setting}, suppose that for each $j \in J$,
\begin{align}
    \text{the } (P,\F_\T) \text{-compensator of } ~^{\tau^J}N^j \text{ is also a } (P,\F_\T^{\tilde Y})\text{-compensator of } ~^{\tau^J}N^j. \label{eq: j comp measure}
\end{align}
Then the $(P,\F_\T)$-compensator of $\mathbb{N}^J$ is also a $(P,\F_\T^{\tilde Y})$-compensator of $\mathbb{N}^J$; i.e., the exchangeability condition Theorem \ref{theorem: multiple marked interventions} \ref{enum: mpp exchangeability} holds.
\end{lemma}

\begin{proof}
We prove the result assuming one intervention, i.e., $J = \{ j \}$. The case when $|J| > 1$ is shown with analogous arguments at the expense of a higher notational burden. 

Recall that we in Section \ref{section: identification with multiple interventions in the general MPP setting} assume that the mark set $X^j$ is finite for each $j \in J$. Let $\{x_k^j\}_k$ be an enumeration of $X^j$ which is also a partition of $X^j$ (i.e. $x_k^j = x_\ell^j \implies k = \ell$). Define, for each $k \geq 1$,
\begin{align*}
    \tau^{j,k} &:= \inf\{s > 0 | N^{j,k}_s \neq \mathscr{N}^{j,k}_s\}, \\
    \mathbb{N}^{j,k}_t &:= I(\tau^{j,k} \leq t),
\end{align*}
where $N^{j,k}_t := N^j((0,t]\times\{x_k^j\})$ and $\mathscr{N}^{j,k}_t := \mathfrak n^j(L, N, (0,t]\times\{x_k^j\})$. Each process $\mathscr{N}^{j,k}$ is $\F_\T$-predictable because the intervention $\mathfrak n^j$ is predictable. 

Define $\tau^{(1)} := \tau^{j,1}$, $\mathbb{N}^{(1)}_t := \mathbb{N}^{j,1}_t$, and for $n \geq 1$
$$
    \tau^{(n+1)} := \tau^{(n)} \wedge \tau^{j,n+1}, \qquad \mathbb{N}^{(n+1)}_t := I(\tau^{(n+1)} \leq t).
$$
We make use of the basic identity
\begin{align}
\begin{split}
    \mathbb N_t^{(n+1)} &= I(\tau^{(n+1)} \leq t) = I(\tau^{(n)} \wedge \tau^{j,n+1} \leq t) \\
    &= ~^{\tau^{(n+1)}} \Big( \mathbb N_t^{(n)} + \mathbb N_t^{j,n+1} - [\mathbb N^{(n)}, \mathbb N^{j,n+1}]_t \Big), 
\end{split}\label{eq: basic obs}
\end{align}
and the representation
 \begin{align}
     \mathbb{N}^{j,k} &= ~^{\tau^{j,k}} N^{j,k} + ~^{\tau^{j,k}}\mathscr{N}^{j,k} - 2\int_0^{\tau^{j,k} \wedge \cdot} \Delta \mathscr{N}^{j,k}_s dN^{j,k}_s. \label{eq: Njk representation}
 \end{align}
\eqref{eq: Njk representation} can be verified using arguments similar to those used in Lemma \ref{lemma: mathbb representations} (compare \eqref{eq: Njk representation} with \eqref{eq: mathbb N a explicitly}). Combining \eqref{eq: basic obs}-\eqref{eq: Njk representation}, and using that $[\mathbb N^{j,n+1}, \mathbb N^{(n)}] = \int_0^{t }\Delta \mathbb N_s^{j, n+1} d \mathbb N_s^{(n)}$ and $\tau^{j,n+1} \geq \tau^{(n+1)}$, we obtain the recursive system
\begin{align}
    \begin{split}
        \mathbb N_t^{(n+1)} &= ~^{\tau^{(n+1)}} \mathbb N_t^{j,n+1} \\
        &\quad + \int_0^{t \wedge \tau^{(n+1)} } \bigl( 1 - \Delta \mathscr N_s^{j,n+1}  - \Delta N_s^{j,n+1} + 2 \Delta \mathscr N_s^{j,n+1}\Delta  N_s^{j,n+1} \bigr) d \mathbb N_s^{(n)},
    \end{split} \label{eq: recursive main} \\[4pt]
    \mathbb N^{(1)}_t &= \int_0^{t \wedge \tau^{j,1}} ( 1 - 2 \Delta \mathscr N_s^{j,1} ) d N^{j,1}_s + ~^{\tau^{j,1}} \mathscr N_t^{j,1}. \label{eq: recursive base}
\end{align}
Using \eqref{eq: recursive main} and \eqref{eq: recursive base}, we will show that, for each $n \geq 1$, there are bounded and $\F_\T$-predictable processes $\{ A^{n,k} \}_{k \leq n}$ such that
\begin{align}
    ~^{\tau^J}\mathbb{N}^{(n)}_t = \sum_{k=1}^n \Big( \int_0^{t \wedge \tau^J} A_s^{n,k} dN_s^{j,k} + ~^{\tau^J}\mathscr{N}_t^{j,k} \Big). \label{eq: induction sum}    
\end{align}
\eqref{eq: induction sum} expresses $~^{\tau^J}\mathbb{N}^{(n)}$ as a sum of terms that are either $\F_\T$-predictable, or integrals of bounded and $\F_\T$-predictable processes against $~^{\tau^J}N^{j,k}$. By Assumption \eqref{eq: j comp measure}, each $~^{\tau^J}N^{j,k}$ has the same compensator with respect to both $\F_\T$ and $\F_\T^{\tilde Y}$ under $P$. A standard result on stochastic integrals (see, e.g., \cite[Theorem 8.2.9]{cohen2015stochastic}) then implies that, if a process $Z$ is the integral of a bounded and $\F_\T$-predictable process against $~^{\tau^J}N^{j,k}$, then $Z$ has the same compensator with respect to the two filtrations under $P$. Moreover, each $\mathscr{N}^{j,k}$ is $\F_\T$-predictable, and therefore its own compensator with respect to both filtrations. It follows that, if \eqref{eq: induction sum} holds, then $~^{\tau^J}\mathbb{N}^{(n)}$ has the same compensator with respect to both filtrations for each $n$. Since $\tau^{(n)} \downarrow \tau^J$ and $I(\tau^J \leq t) = \mathbb{N}^J_t$, it follows that $~^{\tau^J}\mathbb{N}^J = \mathbb{N}^J$ has the same compensator with respect to $P$ and both filtrations. Thus, proving the lemma reduces to establishing \eqref{eq: induction sum}.

We proceed by induction. First, we see that the base case $n=1$ in \eqref{eq: recursive base} stopped at $\tau^J$ has the desired representation with $A^{1,1}_s = 1 - 2\Delta\mathscr{N}_s^{j,1}$. That is, \eqref{eq: induction sum} holds for $n = 1$.

Assume now that \eqref{eq: induction sum} holds for some $n > 1$, and consider \eqref{eq: recursive main}. First, by stopping \eqref{eq: Njk representation} at $\tau^J$ we get
\begin{align*}
    ~^{\tau^J}\mathbb{N}_t^{j,n+1} = \int_0^{t \wedge \tau^J} (1 - 2\Delta\mathscr{N}_s^{j,n+1}) dN_s^{j,n+1} + ~^{\tau^J} \mathscr{N}_t^{j,n+1}, 
\end{align*}
which equals the first term on the right-hand side of \eqref{eq: recursive main} stopped at $\tau^J$. Next, inserting the induction hypothesis into the second term on the right-hand side of \eqref{eq: recursive main}, and stopping at $\tau^J$, we get
\begin{align*}
    &\int_0^{t \wedge \tau^J}( 1 - \Delta \mathscr N_s^{j,n+1} - \Delta N_s^{j,n+1} + 2 \Delta \mathscr N_s^{j,n+1}\Delta  N_s^{j,n+1} ) d \mathbb N_s^{(n)} \\
    &= \sum_{k=1}^n  \int_0^{t \wedge \tau^J} ( 1 - \Delta \mathscr N_s^{j,n+1} - \Delta N_s^{j,n+1} + 2 \Delta \mathscr N_s^{j,n+1}\Delta  N_s^{j,n+1} ) A_s^{n,k} dN_s^{j,k} \\
    &+ \sum_{k=1}^n \int_0^{t \wedge \tau^J} ( 1 - \Delta \mathscr N_s^{j,n+1} - \Delta N_s^{j,n+1} + 2 \Delta \mathscr N_s^{j,n+1}\Delta  N_s^{j,n+1} ) d\mathscr{N}_s^{j,k} \\
    &\stackrel{(\star)}{=} \sum_{k=1}^n  \int_0^{t \wedge \tau^J} ( 1 - \Delta \mathscr N_s^{j,n+1} ) A_s^{n,k} dN_s^{j,k} \\
    &+ \sum_{k=1}^n \int_0^{t \wedge \tau^J} ( 1 - \Delta N_s^{j,n+1} )  d\mathscr{N}_s^{j,k} \\
    &= \sum_{k=1}^n  \int_0^{t \wedge \tau^J} ( 1 - \Delta \mathscr N_s^{j,n+1} ) A_s^{n,k} dN_s^{j,k} \\
    &+ \sum_{k=1}^n  ~^{\tau^J} \mathscr{N}_t^{j,k} - \int_0^{t \wedge \tau^J} \sum_{k=1}^n   \Delta \mathscr{N}_s^{j,k} dN_s^{j,n+1}.
\end{align*}
In $(\star)$ we used that $\Delta N_s^{j,k} \Delta N_s^{j,\ell} = \Delta \mathscr N_s^{j,k} \Delta \mathscr N_s^{j,\ell} = 0$ whenever $k \neq \ell$. This holds because $N^{j,k}$ and $N^{j,\ell}$ (resp. $\mathscr N^{j,k}$ and $\mathscr N^{j,\ell}$) count jumps of distinct marks of the same MPP $N^j$ (resp. $\mathscr N^j := \mathfrak n^j(L, N)$) and they therefore have no common jumps.

Combining the previous two equations, we obtain
\begin{align*}
    ~^{\tau^J}\mathbb{N}^{(n+1)}_t &= \int_0^{t \wedge \tau^J} (1 - 2\Delta\mathscr{N}_s^{j,n+1})  dN_s^{j,n+1} + ~^{\tau^J} \mathscr{N}_t^{j,n+1} \\
    &+ \sum_{k=1}^n \int_0^{t \wedge \tau^J} (1 - \Delta \mathscr N_s^{j,n+1}) A_s^{n,k}  dN_s^{j,k} \\
    &+ \sum_{k=1}^n  ~^{\tau^J} \mathscr{N}_t^{j,k}  - \int_0^{t \wedge \tau^J} \sum_{k=1}^n \Delta \mathscr{N}_s^{j,k}  dN_s^{j, n+1} \\
    &= \sum_{k=1}^{n+1} \Big( \int_0^{t \wedge \tau^J} A_s^{n+1,k}  dN_s^{j,k} + ~^{\tau^J}  \mathscr{N}_t^{j,k} \Big),
\end{align*}
where we in the last line defined
\begin{align}
    \begin{split}
        A_s^{n+1,k} &:= (1 - \Delta \mathscr N_s^{j,n+1}) A_s^{n,k} \quad \text{for } k \leq n, \\
        A_s^{n+1,n+1} &:= 1 - 2\Delta\mathscr{N}_s^{j,n+1} - \sum_{k=1}^n \Delta \mathscr N_s^{j,k},
    \end{split} \label{eq: A functions}
\end{align}
with the base case $A_s^{1,1} = 1 - 2\Delta\mathscr{N}_s^{j,1}$. Each $A^{n+1,k}$ in \eqref{eq: A functions} has an explicit solution involving products and sums of processes in $\{ \Delta \mathscr N^{j,\ell} \}_{1 \leq \ell \leq n+1}$; consequently, each $A^{n+1,k}$ is $\F_\T$-predictable. Moreover, since at most one of the processes in $\{ \mathscr N^{j,\ell} \}_{1 \leq \ell \leq n+1}$ can jump at any given time, each $A^{n+1,k}$ is also bounded. We conclude that \eqref{eq: induction sum} holds for $n+1$.
\end{proof}

\subsection{The invariance property (IP)}
\label{appendix: the likelihood ratio process which gives invariance}

Recall that we denote the $(P, \F_\T)$- compensator of $N^j$ by $\Lambda^j$. Consider another probability measure $Q^* \ll P$ on $\F_T$ that satisfies (IP) with respect to some $\Lambda^{a,*}$. That is:
\begin{itemize}
    \item For $j \neq a$, the $(Q^*, \F_\T)$-compensator of $N^j$ is $\Lambda^j$ (i.e., is unchanged from $P$);
    \item The $(Q^*, \F_\T)$-compensator of $N^a$ is $\Lambda^{a,*}$ (i.e. is modified from $P$).
\end{itemize}
We here characterize the associated likelihood ratio process $W^*_t := \frac{dQ^*}{dP}\big|_{\F_t}$.

By \cite[Theorem 10.2.2]{LastBrandt1995marked}, $W^*$ solves the SDE
\begin{align*}
    W^*_t &= W^*_0 + \sum_{i=1}^d\int_0^t W^*_{s-} V_s^i dM_s^i,  
\end{align*}
where $V_t^i = \xi(t,i) - 1 - \frac{ \Delta \Lambda_t^a - \Delta \Lambda^{a,*}_t }{1 - \Delta \bar {\Lambda}_t}$ and $W^*_0 = 1$. With the given distributions $P$ and $Q^*$ we have $\xi(t,a) = \frac{d \Lambda^{a,*}_t}{ d\Lambda_t^a }$ while $\xi(t,i)=1$ for $i \neq a$, so
\begin{align*}
    V_t^i &= \Big(\frac{d \Lambda^{a,*}_t}{d\Lambda^a_t} - 1\Big)I(i=a) - \frac{ \Delta  \Lambda_t^a - \Delta  \Lambda^{a,*}_t }{1 - \Delta \bar {\Lambda}_t}.
\end{align*}
Since $W^*_0=1$, the likelihood-ratio process solves
\begin{align}
\begin{split}
    W^*_t &= 1 - \sum_{i \neq a}\int_0^t W^*_{s-} \frac{\Delta \Lambda^a_s - \Delta \Lambda^{a,*}_s}{1 - \Delta \bar \Lambda_s} dM_s^i \\
    &+ \int_0^t W^*_{s-} \Big( \frac{d\Lambda^{a,*}_s}{d\Lambda_s^a} - 1 -  \frac{\Delta \Lambda^a_s - \Delta \Lambda^{a,*}_s}{1 - \Delta \bar \Lambda_s} \Big) dM_s^a.
\end{split}\label{eq: (IP) lik rat proc}
\end{align}

\subsubsection{An example of a distribution satisfying (IP) which is not a $g$-formula distribution}
\label{appendix: non g formula distribution}

Suppose the $(P, \F_\T)$-compensator of $N^j_t$ is $\Lambda_t^j = t$ for each $j$. Consider the process in \eqref{eq: (IP) lik rat proc} with $\Lambda^{*,a}_t = c t$ for some deterministic $c > 0$. In this case, \eqref{eq: (IP) lik rat proc} reduces to
$$ W_t^* = 1 - \int_0^t(c - 1) W_{s-}^* dM_s^a. $$
The preceding SDE has the explicit solution $W_t^* = c^{N_t^a}e^{-(c-1)t}$. $N_t^a$ is a Poisson process with rate 1 under $P$, and we therefore get by explicit expressions of moment generating functions that $ E_P[c^{N_t^a}] = E_P[e^{N_t^a ln(c)}] = e^{t(e^{ln(c)} - 1)} = e^{t(c - 1)}$. Thus $E_P[W_t^*] = 1$ for each $t \in \T$, and $W^*$ is a mean one nonnegative uniformly integrable martingale with respect to $P$ and $\F_\T$. We can therefore define the measure $dQ^* := W_T^* dP$ on $\F_T$, which is then going to be a probability measure. $Q^*$ satisfies (IP) with respect to $\Lambda^{a,*}_t = c t$. However, there is no (deterministic) intervention $\mathfrak n^a$ as in Definition \ref{definition: intervention} such that $\mathfrak n^a_t(N) = c t$.

\subsection{Definition \ref{definition: potential outcome process} \ref{item: invariance a} for a predictable regime}
\label{subsection: potential outcome process for predictable regime}
We show first that if \eqref{eq: predictable intervention tilde a compensator} holds, then Definition \ref{definition: potential outcome process} \ref{item: invariance a} holds. Because $\mathfrak n^a$ is predictable we have that $\mathfrak{n}^a(\tilde N)$ is $\tilde \F_\T$-predictable by \cite[Corollary 2.2.8]{LastBrandt1995marked}, and by \eqref{eq: predictable intervention tilde a compensator} we get in particular that
\begin{align*}
    E_P[\tilde N^a_{\tilde \tau}] = E_P[\mathfrak n^a_{\tilde \tau}(\tilde N)]  \quad \text{ for each } \tilde \F_\T\text{-optional time }  \tilde \tau.
\end{align*}
 The preceding equality and the predictability of $\mathfrak{n}^a(\tilde N)$ implies, by the definition of compensators, that Definition \ref{definition: potential outcome process} \ref{item: invariance a} holds.

Assume now that Definition \ref{definition: potential outcome process} \ref{item: invariance a} holds, i.e., the $(P,\tilde \F_\T)$-compensator of $\cfproc{a}$ is $\na{}(\cfproc{})$. Because $\mathfrak n^a$ is predictable, $\na{}(\cfproc{})$ is a counting process that is predictable with respect to $\tilde \F_\T$. Defining $\tilde T_0^a := 0$ and
\begin{align}
    \tilde T_n^a := \inf \{ s > 0 | \na{s}(\cfproc{}) \geq n \}, \label{eq: tilde T n a definition}
\end{align}
it follows by \cite[Theorem 2.1.30]{LastBrandt1995marked} that each $\tilde T_n^a$ is a predictable time with respect to $\tilde \F_\T$. Consequently, $\tilde h_t^n := I(\tilde T_{n}^a < t < \tilde T_{n+1}^a)$ defines an $\tilde \F_\T$-predictable process. This can be seen by noting that $I(\tilde T_{n}^a < t < \tilde T_{n+1}^a) = I(\tilde T_{n}^a < t ) - I(\tilde T_{n+1}^a \leq t)$, where $I(\tilde T_{n}^a < t )$ is left-continuous and adapted (and thus predictable), and $I(\tilde T_{n+1}^a \leq t)$ is predictable by \cite[Theorem 2.1.35]{LastBrandt1995marked} because $\tilde T_{n+1}^a$ is a predictable time. 
By the definition of compensators, it follows that
$$ E_P\Big[\int_0^T \tilde h_t^n d\tilde N^a_t\Big] = E_P\Big[\int_0^T \tilde h_t^n d \mathfrak n^a_t(\tilde N)\Big] = 0, $$
where we used that $\mathfrak n^a(\tilde N)$ is piecewise constant with jumps at $\{ \tilde T_n^a \}_n$. Since the terms in the preceding expectations are non-negative, it follows that $\int_0^T \tilde h_t^n d\tilde N^a_t = 0$ $P$-a.s. for each $n$. Because $\tilde h^n \geq 0$ and $\tilde N^a$ is non-decreasing, we obtain $\tilde h^n_t \tilde N^a(dt) = 0$ $P$-a.s. for each $n$ ---  $\tilde N^a_t$ is constant on the set $\{ \tilde T_n^a < t < \tilde T_{n+1}^a \}$ for each $n$. That is, $\tilde N^a$ can only jump at the times $\{ \tilde T_n^a \}_n$.

Since $\tilde T_n^a$ is a $\tilde \F_\T$-predictable time and $\mathfrak n^a(\tilde N)$ is a compensator of $\tilde N^a$ with respect to $\tilde \F_{\T}$, we have by e.g. \cite[Equation (4.1.16)]{LastBrandt1995marked} that
\begin{align*}
    E_P[ \Delta \tilde N_{\tilde T_1^a}^a | \tilde \F_{\tilde T_1^a-} ]I(\tilde T_1^a < \infty) &=  \Delta \na{\tilde T_1^a}(\tilde N)I(\tilde T_1^a < \infty) = I(\tilde T_1^a < \infty) \quad P\text{-a.s.},
\end{align*}
where the last equality follows from \eqref{eq: tilde T n a definition}. Using that $\tilde T_1^a \in \tilde \F_{\tilde T_1^a-}$ (by \cite[Theorem 2.2.15]{LastBrandt1995marked}) and the tower property, we get from the preceding equality that $E_P[ \Delta \tilde N_{\tilde T_1^a}^a I(\tilde T_1^a < \infty) ] = E_P[ I(\tilde T_1^a < \infty) ]$.  Because $\Delta \tilde N^a \leq 1$ (since $\tilde N^a$ is a counting process), it follows that $\Delta \tilde N_{\tilde T_1^a} = 1$ $P$-a.s. on $\{ \tilde T_1^a < \infty \}$. Combining this with the observation that $\tilde N_t^a = 0$ on $\{  t < \tilde T_1^a \}$, we conclude that $\tilde N^a_{\tilde T_1^a \wedge \cdot} = \na{\tilde T_1^a \wedge \cdot}(\tilde N)$ $P$-a.s. The identity $\tilde N^a_{\tilde T_n^a \wedge \cdot} = \na{\tilde T_n^a \wedge \cdot}(\tilde N)$ $P$-a.s. for each $n$ follows by induction. 

\subsection{The equivalence of exchangeability conditions in Example \ref{example: delayed allocation}}
\label{appendix: example: delayed allocation}

We show the equivalence of the exchangeability condition \ref{enum: mpp exchangeability} of Theorem \ref{theorem: multiple marked interventions} and \eqref{eq: delayed allocation seq ex} when there is no censoring. We use that $\bar N^a$ is $\F_\T$-predictable, which holds due to the delay assumption \eqref{eq: delay assumption}. It follows from \cite[Theorem 2.1.30]{LastBrandt1995marked} that, for each $k \geq 1$, $T_k^a$ is an $\F_\T$-predictable time, where $\{ T_k^a \}_k$ are the ordered jump times of $\bar N^a = N^a((0,\cdot] \times X^a)$. Because these times are predictable, 
$$ h^k_s := I(T_k^a < s < T_{k+1}^a)  $$
is an $\F_\T$-predictable process for each $k$. This can be seen, as $I(T_k^a < s < T_{k+1}^a) = I(T_k^a < s) - I(T_{k+1}^a \leq s)$, and $I(T_k^a < s)$ is left-continuous and adapted (because $T_k^a$ is an optional time) and thus predictable, and $I(T_{k+1}^a \leq t)$ is predictable by \cite[Theorem 2.1.35]{LastBrandt1995marked}.

For the intervention \eqref{eq: allocation assignment} under study, (assuming no right-censoring), regime deviation can only occur at the treatment allocation times $\{ T_k^a \}_k$, i.e. $\tau^a$ takes values in the range of $\{ T_k^a \}_k$; consequently, $\int_0^T h_s^k d\mathbb N^a_s = 0$ $P$-a.s. for each $k$, where $\mathbb N^a = I(\tau^a \leq \cdot)$. Since $\mathbb \Lambda^a$ is a compensator of $\mathbb N^a$, it thus follows from the definition of compensators and the predictability of each $h^k$ that
\begin{align*}
    0 = E_P\Big[ \int_0^T h_s^k d\mbbNa{s} \Big] = E_P\Big[ \int_0^T h_s^k d\mathbb \Lambda_s^a \Big].
\end{align*}
Since $\mathbb \Lambda^a$ is nondecreasing and $h^k \geq 0$, it follows from the preceding line that $h_t^k d\mathbb \Lambda^a_t = 0$ $P$-a.s. for each $k\geq 1$. Hence, $\mathbb \Lambda^a$ is constant on each open interval $(T_k^a, T_{k+1}^a)$. That is, $\mathbb \Lambda^a$ can only change at the times $\{ T_k^a \}_k$, $P$-a.s. Since these times are predictable, it follows from \cite[Equation (4.1.16)]{LastBrandt1995marked} that the exchangeability condition \ref{enum: mpp exchangeability} in Theorem \ref{theorem: multiple marked interventions} holds if and only if
    \begin{align}
    \begin{split}
        E_P\big[\Delta \mathbb N_{T_k^a}^a | \F_{T_k^a-}\big] I(T_k^a < \infty) &= \Delta \mathbb \Lambda_{T_k^a}^a I(T_k^a < \infty) \\
        &= E_P\big[\Delta \mathbb N_{T_k^a}^a | \F_{T_k^a-} \vee \sigma(\tilde Y)\big] I(T_k^a < \infty)
    \end{split}\label{eq: ex jump equivalence}
    \end{align}
$P$-a.s. for each $k \geq 1$. We now use the basic identity $\Delta \mathbb N_{T_k^a}^a = \big (1 - I(\tau^a > T_k^a) \big) I(\tau^a \geq T_k^a)$. Inserting this identity in \eqref{eq: ex jump equivalence} we get, after simplifying, that this equality holds if and only if
\begin{align*}
    &E_P\big[I(\tau^a > T_k^a) | \F_{T_k^a-}\big] I(\tau^a \geq T_k^a, T_k^a < \infty)\\
    &= E_P\big[I(\tau^a > T_k^a) | \F_{T_k^a-} \vee \sigma(\tilde Y)\big] I(\tau^a \geq T_k^a, T_k^a < \infty)
\end{align*}
$P$-a.s. for each $k\geq 1$. Here we used that $I( \tau^a \geq T_k^a ) \in \F_{T_k^a-}$ (apply \cite[Theorem 2.1.16 (ii)]{LastBrandt1995marked} to $I(\tau^a \geq \cdot)$) to move the indicator outside the conditional expectations. Because $T_k^a$ is $\F_{T_k^a-}$-measurable (see, e.g., \cite[Theorem 3.4 1)]{He1992Semimartingale}), and $I(\tau^a > T_k^a) = 0$ on $\{ T_k^a = \infty \}$, it follows that the last equality is equivalent to
\begin{align*}
    E_P\big[I(\tau^a > T_k^a) | \F_{T_k^a-}\big] I(\tau^a \geq T_k^a) = E_P\big[I(\tau^a > T_k^a) | \F_{T_k^a-} \vee \sigma(\tilde Y)\big] I(\tau^a \geq T_k^a).
\end{align*}
Because $I(\tau^a > T_k^a) = I(\tau^a > T_k^a)  I(\tau^a \geq T_k^a)$ and $I( \tau^a \geq T_k^a ) \in \F_{T_k^a-}$, we can move the latter indicator inside the conditional expectation again and absorb it into $I(\tau^a > T_k^a)$. It follows that the exchangeability condition Theorem \ref{theorem: multiple marked interventions} \ref{enum: mpp exchangeability} in this example is equivalent to any of the two following conditional independence statements:
\begin{align*}
    I(\tau^a > T_k^a) &\indep \tilde Y | \F_{T_k^a-}, \tau^a \geq T_k^a \text{ for each } k \geq 1, \\
    I(\tau^a > T_k^a) &\indep \tilde Y | \F_{T_k^a-} \text{ for each } k \geq 1.
\end{align*}



\subsection{Incompatibility of the exchangeability condition \eqref{eq: exchan disc} with the strong consistency condition \eqref{eq: stronger consistency} when \eqref{eq: mcp can comp orthogonality} is violated}
\label{appendix: incompatibility of exchangeability and stronger consistency}
Consider a bivariate observed data counting process $N = (N^a, N^y)$ on $\T$ with $T > 1$. Assume the compensator of $N$ under $P$ is
\begin{align}
\begin{split}
    \Lambda^a(dt) &= \lambda^a \delta_1(dt), \\
    \Lambda^y(dt) &= \lambda^y \delta_1(dt),
\end{split}\label{eq: toy obs data compensator}
\end{align}
with real numbers $\lambda^a, \lambda^y$ satisfying $0 < \lambda^a, \lambda^y < 1$ and $\lambda^a + \lambda^y \leq 1$. The compensator in \eqref{eq: toy obs data compensator} are constant as a functional of $N = (N^a, N^y)$. Consequently, we  may take the compensator to be equal to the canonical compensator; i.e. $\alpha^j = \Lambda^j$ for $j \in \{ a,y \}$, where $\alpha^j$ is as in \eqref{eq: alpha j compensator}. In particular, from \eqref{eq: toy obs data compensator} it follows that
\begin{align*}
    \Delta  {\alpha}_1^a \Delta \alpha_1^y = \lambda^a \lambda^y > 0,
\end{align*}
which shows that the orthogonality assumption \eqref{eq: mcp can comp orthogonality} is violated.

Now, consider the intervention $\na{} = 0$ which prevents treatment. By Definition \ref{definition: intervention}, the respective compensators of the potential outcome processes $\tilde N^a$ and $\tilde N^y$ arising under this intervention are
\begin{align}
    \begin{split}
        \mathfrak n^a(\tilde N)(dt) &= 0, \\
        \alpha^y(\tilde N)(dt) &= \lambda^y \delta_1(dt).
    \end{split}
    \label{eq: toy pot outcomes compensator}
\end{align}

Suppose the potential outcome of interest is $\tilde Y = \tilde N^y$, with the associated observed outcome process $Y = N^y$.

For the intervention of interest ($\mathfrak{n}^a = 0$), the regime deviation time is $\tau^a = \inf\{ s > 0 | N_s^a \neq \mathfrak n_s^a(N) \} = \inf\{ s > 0 | N_s^a \neq 0 \}$; i.e. $\tau^a$ and coincides with the first jump time of $N^a$. 

Now observe from \eqref{eq: toy obs data compensator} that $\tau^a \geq 1$. Combining this observation with the stronger consistency assumption \eqref{eq: stronger consistency}, it follows that
\begin{align*}
\Delta N_1^y = \Delta \tilde{N}_1^y \quad P\text{-a.s.}
\end{align*}
Moreover, the exchangeability condition \eqref{eq: exchan disc} reduces to
\begin{align*}
    E_P[\Delta \tilde{N}_1^a] = E_P[\Delta \tilde{N}_1^a | \tilde{N}^y |_T] \quad P \text{-a.s.}
\end{align*}
Combining the last two equalities, and using that $\tilde N^y|_T = \Delta \tilde N^y_1$ (which follows from \eqref{eq: toy pot outcomes compensator}), we get, $P$-a.s.,
\begin{align}
\begin{split}
    E_P[\Delta N_1^a] &= E_P[\Delta N_1^a | \tilde N^y |_T]
    = E_P[\Delta N_1^a | \Delta \tilde N_1^y] \\
    &= E_P[\Delta N_1^a | \Delta N_1^y] = E_P[\Delta N_1^a | \Delta N_1^y = 0](1 - \Delta N_1^y),
\end{split} \label{eq: exchangeability requirement single}
\end{align}
where the last equality follows due to the no-common-jumps property $\Delta N_1^a\Delta N_1^y = 0$ of multivariate counting processes. Using that the compensator in \eqref{eq: toy obs data compensator} specifies the joint distribution
\begin{align*}
P(\Delta N_1^a = 1, \Delta N_1^y = 1) &= 0, \\
P(\Delta N_1^a = 1, \Delta N_1^y = 0) &= \lambda^a, \\
P(\Delta N_1^a = 0, \Delta N_1^y = 1) &= \lambda^y, \\
P(\Delta N_1^a = 0, \Delta N_1^y = 0) &= 1 - \lambda^a - \lambda^y,
\end{align*}
we get
\begin{align*}
E_P[\Delta N_1^a | \Delta N_1^y = 0] = \frac{P(\Delta N_1^a = 1, \Delta N_1^y = 0)}{P(\Delta N_1^y = 0)} = \frac{\lambda^a}{1 - \lambda^y}.
\end{align*}
Substituting the preceding equality into \eqref{eq: exchangeability requirement single}, our exchangeability condition requires
\begin{align}
\lambda^a = \frac{\lambda^a}{1 - \lambda^y}(1 - \Delta N_1^y). \label{eq: la ort violation}
\end{align}
However, the equality \eqref{eq: la ort violation} can not hold in general. For example, with $\lambda^a = 0.3$ and $\lambda^y = 0.4$, the preceding statement reduces to
\begin{align*}
0.3 = \frac{0.3}{0.6}(1 - \Delta N_1^y) = 0.5(1 - \Delta N_1^y),
\end{align*}
which can not hold, because $\Delta N^y_1 \in \{ 0, 1 \}$.

\subsection{Testability of exchangeability assumptions}
\label{subsection: testability of exchangeability}
The exchangeability conditions presented here can be seen to be single world irrelevance conditions. Using that compensators of $\mbbNa{}$ are constant on $\{ \cdot > \Nastop \}$, the statement \eqref{eq: exchan disc} is under the consistency condition $N|_\cdot I(\cdot  < \Nastop) = \cfproc{}|_\cdot I(\cdot  < \Nastop)$ $P$-a.s. equivalent to \footnote{Here, $\Upsilon_{\cdot } = \dot \Upsilon_{\cdot }(N|_{\cdot -}, \tilde Y )$ is an $\F_\T^{\tilde Y}$-compensator of $\mbbNa{}$.} 
 \begin{align*}
     \dot {\mathbb \Lambda}^a_{\cdot }(\cfproc{}|_{\cdot -}) = \dot \Upsilon_{\cdot }(\cfproc{}|_{\cdot -},\tilde Y|_T ) \quad P\text{-a.s.}, 
 \end{align*}
a statement involving functionals of realizations coming from the same intervention. As such, our conditions differ from cross-world assumptions such as those present in the NPSEM-IE model (aka structural causal model with independent errors) \cite{pearl},
which contain independence statements about a given unit between different counterfactual worlds. Such cross-world independence assumptions, generally, cannot be experimentally tested. For further discussion and some implications we refer to \cite{richardson_single_2013,robins2011alternative,Andrews2020}. 

  The exchangeability conditions we formulate here can in principle always be assessed indirectly using the following recipe: first, assume that exchangeability holds in a given study, and calculate effects of interest according to \eqref{eq: exchangeability ipw formula} (this means at least also assuming that positivity and consistency holds). Then, compare the estimated effects with estimates coming from an ideal experiment with perfect compliance conducted in the same population, where the intervention  is implemented. If the difference between these estimates are statistically significant, and all other assumptions hold (which they likely do not), this would indicate that exchangeability is violated. 
       
  Following  \cite{Young2014natural,richardson_single_2013}, tests of exchangeability assumptions like ours can be envisaged if we have access to the process recording the natural value of treatment. A version of our condition involving the natural value process can be directly tested on the basis of observations coming from an ideal experiment where the regime is implemented and the natural value process is recorded.  \footnote{For instance, we can then introduce the optional time "the first time the regime deviates from the natural value process." An exchangeability condition for the identification of regimes as presented here can be phrased in terms of the compensator of this optional time. This optional time, and potential outcomes of interest, are observable in experiments where the intervention is implemented with perfect compliance. When these criteria are met, one can check whether the exchangeability assumption holds by e.g. fitting regression models to estimate its compensator with respect to different filtrations, and inspecting the resulting martingale residuals. } 
To formally employ such a condition in this work we would need to alter our definitions so as to include the natural value process. Such an altering would lead to changes throughout and is a separate work in its own right.

\subsection{Orthogonality of counting process martingales}
\label{appendix: orthogonal counting process martingales}
    If mild integrability conditions are satisfied, a counting process minus its compensator becomes a locally square integrable martingale. This means that predictable variation processes exist. Consider such a multivariate counting process $N$ with compensator $\Lambda$, so that $M = N - \Lambda$ is a locally square integrable martingale. Two martingale components $M^i, M^j$, $i \neq j$, are \emph{orthogonal} if their predictable variation process is zero, i.e. if $\langle M^i, M^j \rangle = 0$ a.s. Elementary calculations with counting process martingales then gives that  \footnote{see e.g. \cite[p. 75]{Andersen}.}
      $$ \langle M^i, M^j \rangle = \int_0^\cdot \Delta \Lambda_s^i d\Lambda_s^j = \int_0^\cdot \Delta \Lambda_s^j d\Lambda_s^i = 0 \quad P\text{-a.s.}  $$
      In particular, orthogonality of the martingales holds if almost every path of each compensator process is continuous. Of particular relevance for the statistical literature on counting processes is the situation when intensities exist, i.e. when there exists a non-negative and adapted process $\lambda=(\lambda^1, \dots, \lambda^d)$ such that $\Lambda = \int_0^\cdot \lambda_s ds$. A significant proportion of the existing statistical methods based on multivariate counting processes assumes that counting process intensities exist \cite{Andersen,aalen2008survivalandevent,FlemingHarrington2005,martinussen2006dynamic,CookLawless2007}. 

\subsubsection{Practical violation of orthogonal martingales}
\label{appendix: failure of orthogonality assumption}
Consider a situation where the survival status of possibly right-censored subjects only are monitored at a collection of predictable times $\{ \sigma_k \}_k$. This could correspond to a follow-up schedule which is decided in advance, or adaptively determined based on earlier recorded events, where information on the subjects is only collected at the scheduled times. The observed censoring and survival status is unknown except at the follow-up times, but each of the events generally has a non-zero chance at happening (i.e. being observed) at follow-up. Formally, there are non-negative constants $\{ \theta_k, \beta_k \}_k$ such that the compensators of the observed censoring and death processes are given by
    \begin{align*}
        \Lambda^c(ds) &= \sum_{k} J_{\sigma_k} \theta_k \delta_{\sigma_k}(ds), \\
        \Lambda^d(ds) &= \sum_{k} J_{\sigma_k} \beta_k \delta_{\sigma_k}(ds),
    \end{align*}
where $J_t = I(N_{t-}^c = N_{t-}^d = 0)$ is the at-risk process. We have that $\Delta \Lambda_{\sigma_k}^c \Delta \Lambda^d_{\sigma_k} = J_{\sigma_k} \theta_k \beta_k$, which is positive whenever $\theta_k\beta_k > 0$ and $J_{\sigma_k} = 1$ (i.e. when individuals are at risk), and the 'observed' martingales are not orthogonal.

\subsection{A proof sketch of \eqref{eq: exchangeability ipw formula} under a particular sequential exchangeability condition}
\label{appendix: seq exchangeability proof sketch}

We here outline an argument which gives the IPW formula \eqref{eq: exchangeability ipw formula} when the exchangeability condition \eqref{eq: exchan disc} is swapped with the sequential condition \eqref{eq: exch strong predict}, given in Proposition \ref{proposition: sequential exchangeability predictable times}. We adopt in Proposition \ref{proposition: sequential exchangeability predictable times} the assumption that the event times $\{ T_k \}_k$ are predictable. To carry out the argument we rely on alternative representations of the likelihood ratio $W$, expressed in terms of sums over the event times, which is established in the following lemmas. 

\begin{lemma}\label{lemma: dolean dade identity}
     $W = \E(\KK)$ satisfies, $P$-a.s.
    \begin{align}
        W_t &= 1 + \sum_{m=0}^\infty \int W_{s-} I(T_m < s \leq t \wedge T_{m+1} ) d\mathbb K^a_s \label{eq: dolean primitive sum} \\
        W_t &= 1 - \sum_{m=0}^\infty W_{T_m} I(t > T_m) \Big(1 - \frac{I(\Nastop > t \wedge T_{m+1} > T_m  )}{\E\left( - \int I(T_m < s ) d \mbbLa{s} \right)_{t \wedge T_{m+1}}}\Big)
        \label{eq: dolean sum}
    \end{align}
\end{lemma}
\begin{proof}
    Using that $W = 1 + \int_0^\cdot W_{s-}d\mathbb K^a_s$ we get
    \begin{align*}
        W_t = 1 + \int_0^t W_{s-} d\mathbb K^a_s
        = 1 + \sum_{m=0}^\infty \int W_{s-} I(T_m <  s \leq t \wedge T_{m+1} ) d\mathbb K^a_s
    \end{align*}
    by the linearity of the integral. This shows \eqref{eq: dolean primitive sum}. Next, we show that the sum in the previous equation coincides with the sum in \eqref{eq: dolean sum}.  We get that
    \begin{align*}
        \int_0^t W_{s-} I(T_m < s \leq T_{m+1} ) d\mathbb K^a_s & = W_{T_m}\int_0^t \E\Big( \int I(T_m < u \leq T_{m+1}) dK_u \Big)_{s-} \\
        &I(T_m < s \leq  T_{m+1} ) d\mathbb K^a_s \\ 
        &= W_{T_m}\int_0^t \E(  \mathbb H^m )_{s-} d \mathbb H_s^m \\
        &= W_{T_m} \big( 1 - \E(  \mathbb H^m )_{t} \big)
    \end{align*}
    where $\mathbb H_t^m = \int_0^t I(T_m < s \leq T_{m+1} ) d\mathbb K^a_s$. Here we used that  $$\E(\mathbb H^m) = \allowbreak 1 + \allowbreak \int \E(\mathbb H^m)_{s-}d\mathbb H^m_s,$$ and that 
    \begin{align*}
        W_{s-}I(T_m < s \leq T_{m+1} ) &= W_{T_m} \frac{ \prodi_{T_m < u < s}(1 - d \mathbb N_u^a) }{\prodi_{T_m < u < s}(1 - d \mbbLa{u}) }  I(T_m < s \leq T_{m+1} ) \\
        &=   W_{T_m} \E(H^m)_{s-} I(T_m < s \leq T_{m+1} ),
    \end{align*} 
    which can be seen by inspecting the calculations in the proof of Proposition \ref{prop: characterizations of W}. We can conclude that 
    \begin{align*}
        W_t = 1 - \sum_{m=0}^\infty W_{T_m} \big(1 - \E( \mathbb H^m)_t \big).
    \end{align*}
    Next, we note that since $\mathbb H_t^m  = 0$ on $ \{ t \leq T_m\}$, we have 
    $ W_{T_m} \big( 1 - \E\left(  \mathbb H^m \right)_{t} \big) = 0 $ on this set and we may multiply with  $I(t > T_m)$. Since we are also multiplying with $W_{T_m}$ we are on the set $\{ \Nastop > T_m \}$, and we can consequently assume that $t \wedge \Nastop > T_m$.  This gives that, arguing as in the proof of Proposition \ref{prop: characterizations of W}, 
    \begin{align}
        \begin{split}
        &\frac{I(\Nastop > t \wedge T_{m+1} > T_m )}{\E\left( - \int I(T_m < s \leq T_{m+1}) d \mbbLa{s} \right)_{t}} =  \frac{\prodi_{T_m < s \leq T_{m+1} \wedge t } (1 - d \mbbNa{s}) }{\prodi_{T_m < s \leq T_{m+1} \wedge t } (1 - d \mbbLa{s})} \\
        &= 1 - \int I(T_m < s \leq T_{m+1} \wedge t) \frac{\prodi_{T_m < u < s } (1 - d \mbbNa{u}) }{\prodi_{T_m < u < s } (1 - d \mbbLa{u})} \frac{ d\mbbNa{s} - d \mbbLa{s} }{1 - \Delta \mbbLa{s}}.
        \end{split} \label{eq: prop ratio martingale}
    \end{align}
    The given ratio on the left-hand side of \eqref{eq: prop ratio martingale}  coincides with $\E(\mathbb H^m)_t$, as it solves the same SDE. This shows \eqref{eq: dolean sum}.
\end{proof}

\begin{lemma}\label{lemma: ipw equivalence}
    Let bounded variables $H, \tilde H$ be such that $E_P[ W_t \tilde H ] = E_P[ W_t H ]$. Then, the following are equivalent:  
    \begin{enumerate}[label=(\roman*)]
        \item $E_P[\tilde  H] = E_P[ W_t H]$
        \item $\sum_{m=0}^\infty E_P \Big[ W_{T_m} I(t > T_m )\Big(1 - \frac{I(\Nastop > t \wedge T_{m+1} > T_m )}{\E\left( - \int I(T_m < s ) d \mbbLa{s} \right)_{t \wedge T_{m+1}}}\Big) \tilde  H \Big] = 0$
        \item $\sum_{m=0}^\infty E_P \Big[ \int W_{s-} I(T_m < s \leq t \wedge T_{m+1} ) d\mathbb K^a_s \tilde  H \Big] = 0.$
\end{enumerate}

\end{lemma}
\begin{proof}
   Using the representation in \eqref{eq: dolean sum} we get
     \begin{align*}
         E_P[W_t\tilde  H] &= E_P[\tilde  H] \\
         &- \sum_{m=0}^\infty E_P\left[   W_{T_m} I(T_m < t)\left(1 - \frac{I(\Nastop > t \wedge T_{m+1} > T_m )}{\E\left( - \int I(T_m < s ) d \mbbLa{s} \right)_{t \wedge T_{m+1}}}\right) \tilde   H\right].
     \end{align*}
    By the fact that $E_P[W_t H] = E_P[W_t \tilde  H]$ 
    we get the first equivalence. A similar argument using the representation in \eqref{eq: dolean primitive sum} gives the other equivalence.
\end{proof}

\begin{proposition}
\label{proposition: sequential exchangeability predictable times}
    Suppose positivity and consistency as in Definition \ref{def:strong prime} holds 
    and that 
    \begin{align}
        \tilde Y_t \indep \frac{ I(\Nastop > T_{n+1} \wedge t > T_n) }{ \E( - \int I(T_n < s) d\mbbLa{s} )_{T_{n+1} \wedge t} }  \Big| \F_{T_n}, \Nastop > T_n \label{eq: exch strong}
    \end{align}
    for each $n \geq 1, t \in \T$. 
    Then 
    \begin{align}
        E_P[\tilde Y_t] = E_P[ W_t Y_t]. \label{eq: cfopipw}
    \end{align}
    If the random times $\{ T_k \}_{k=1}^\infty$ are predictable, then \eqref{eq: cfopipw} is true under the independence condition
    \begin{align}
        \tilde Y_t \indep I(\Nastop > T_{n+1} \wedge t > T_n) | \F_{T_n}, \Nastop > T_n \label{eq: exch strong predict}
    \end{align}
    for each $n \geq 1, t \in \T$.
\end{proposition}
\begin{proof}
    We show the result under \eqref{eq: exch strong predict} when each random time $T_k$ is predictable. By consistency we have that $E_P[W_t Y_t] = E_P[W_t \tilde  Y_t]$. Lemma \ref{lemma: ipw equivalence} then says that $$ E_P[\tilde Y_t] = E_P[W_t Y_t] $$
    if and only if
    \begin{align}
        \sum_{m=0}^\infty E_P \Big[ W_{T_m} I(t > T_m )\left(1 - Z_{t}^m\right) \tilde Y_t \Big] = 0, \label{eq: ipw summand}
    \end{align}
    where we have defined
    \begin{align*}
        Z_t^m = \frac{I(\Nastop > t \wedge T_{m+1} > T_m )}{\E\left( - \int I(T_m < s ) d \mbbLa{s} \right)_{t \wedge T_{m+1}}}.
    \end{align*}
    Since $\E\big( - \int I(T_m < s ) d \mbbLa{s} \big)$ is $\F_\T$-predictable we have that $\E\big( - \int I(T_m < s ) d \mbbLa{s} \big)_{T_{m+1}}$ is measurable with respect to $\F_{T_{m+1}-}$. Since each $T_m$ is predictable it follows that $\E\big( - \int I(T_m < s ) d \mbbLa{s} \big)_{T_{m+1}}$ is $\F_{T_m}$-measurable.   We can thus conclude that $\E\big( - \int I(T_m < s ) d \mbbLa{s} \big)_{T_{m+1}\wedge t}$ is $\F_{T_m}$-measurable. A monotone class argument then gives the implication 
      \begin{align}
          \tilde Y_t \indep_{P_{B_m}} I(\Nastop > T_{m+1} \wedge t > T_m) | \F_{T_m} \quad \implies \quad \tilde Y_t \indep_{P_{B_m}} Z_t^m | \F_{T_m}, \label{eq: id under predictable times and sequential exchangeability}
      \end{align}
      using the notation in Appendix \ref{subsection: conditional independence and expectation}. Fix $m$ and consider
    \begin{align}
        E_P\bigg[ W_{T_m} I(t > T_m ) \big( 1 - Z_t^m \big) \tilde Y_t \Big| \F_{T_m} \bigg] = W_{T_m} I(t > T_m ) E_P \left[  \left(1 - Z_t^m \right) \tilde Y_t \Big| \F_{T_m} \right]. \label{eq: rhs mean zero}
    \end{align}
    We show that the right-hand side of \eqref{eq: rhs mean zero} is zero almost surely. If this is true, then \eqref{eq: ipw summand} holds. By definition of $W$, the right hand size of \eqref{eq: rhs mean zero} is zero except on $\{ \Nastop > T_m \}$. We thus only need to show that $E_P[ Z_t^m \tilde Y_t | \F_{T_m}] = E_P[ \tilde Y_t | \F_{T_m}]$ almost surely on the set $\{\Nastop > T_m\}$. We have
    \begin{align*}
        I(\Nastop \wedge t > T_m ) E_P \big[  Z_t^m \tilde Y_t \big| \F_{T_m}\big] &= I(\Nastop \wedge t > T_m ) E_P \big[  Z_t^m \big| \F_{T_m} \big] E_P \big[   \tilde Y_t \big| \F_{T_m} \big].
    \end{align*}
    This equality can be seen to follow from \eqref{eq: id under predictable times and sequential exchangeability}. Using the martingale representation of $Z_t^m$ in \eqref{eq: prop ratio martingale},
    $$ Z_t^m = 1 + \int I(T_m < s \leq T_{m+1} \wedge t) Z_{s-}^m d\mathbb K^a_s, $$
    the optional sampling theorem, \cite[I 2 Theorem 16]{protter}, gives that $E_P[Z_t^m | \F_{T_m}] = 1$ $P$-a.s. for each $t$. 
\end{proof}

\section{Proofs of main results}
\label{appendix: proofs}

\subsection{Proof of Proposition \ref{proposition: mcp construction of observed and potential outcomes}}
\label{appendix: proof of proposition 1}

\ \\
The result follows by applying Proposition \ref{proposition: mcp construction of observed and potential outcomes marked} to the setting of a $d$-dimensional counting process. To formally do so, we must align the notation of Sections \ref{section: set-up and notation}–\ref{section: relation to discrete-time theories} with that of Proposition \ref{proposition: mcp construction of observed and potential outcomes marked}. Specifically, we take $X = \mathcal I_d$, $J = \{a\}$, and for each $j \in \mathcal I_d$ set the mark set to the singleton $X^j = \{j\}$. We then identify the counting process component $N^j$ with the MPP $\delta_j N^j$, and the canonical compensator $\alpha^j$ in \eqref{eq: alpha j compensator} with  $\delta_j \alpha^j$. We similarly identify the intervention $\mathfrak n^a$ with $\delta_a \mathfrak n^a$, and take $L' = L = 1$.

The condition \eqref{eq: compensator explosion regularity} is equivalent to the condition \eqref{eq: tilde alpha explosion regularity condition}. The condition \eqref{eq: compensator jump regularity} combined with the fact that $\Delta \mathfrak n_t^a(\varphi) \leq 1$ implies
\begin{align}
    \Delta \mathfrak n_t^a(\varphi) \sum_{j \neq a} \Delta \alpha_t^j(\varphi) = 0. \label{eq: p1 alpha cond}
\end{align}
Together with the premise \eqref{eq: mcp can comp orthogonality}, this gives the orthogonality conditions \eqref{eq: ort obs multi marked} assumed in Proposition \ref{proposition: mcp construction of observed and potential outcomes marked}. The independence assumptions \ref{enum: mutual indep multi marked}-\ref{enum: exch indep multi marked} in that proposition can be satisfied, for instance, by taking all randomizers in \eqref{eq: uniform random variables marked} to be mutually independent. Hence, all conditions of Proposition \ref{proposition: mcp construction of observed and potential outcomes marked} are met, and \ref{enum: mcp cons} and \ref{enum: mcp law} in Proposition \ref{proposition: mcp construction of observed and potential outcomes} follow directly from \ref{enum: prop consistency marked} and \ref{enum: prop correct law marked} in Proposition \ref{proposition: mcp construction of observed and potential outcomes marked}. Finally, \ref{enum: prop exchangeability marked} in Proposition \ref{proposition: mcp construction of observed and potential outcomes marked} says that $~^{\tau'^a}\alpha^a(N')$ is the compensator of $N'^a$ with respect to both $\{\F'_t \}_t$ and $\{\F_t' \vee \sigma(\tilde N')\}_t$ under $P'$. By the identity \eqref{eq: stopped compensator exchangeability} in Lemma \ref{lemma: mathbb representations}, this is equivalent to \ref{enum: mcp exch} in Proposition \ref{proposition: mcp construction of observed and potential outcomes}.

\subsection{Proof of Lemma \ref{lemma: P fidi dist}}

\ \\
    Let $X_k$ denote the accompanying mark of $T_k$, and consider a generic MPP trajectory $\varphi = (t_k, x_k)_{k \geq 1}$. By \cite[Theorem 8.1.2]{LastBrandt1995marked} there is a kernel $\alpha^{(n)}$ from $\N^{p}$ to $(\Rplus \times \I_p)^n$,
    and a function $U$ from $\Rplus \times \Rplus \times \N^{p}$ to $\Rplus$ such that the density of $T_1, X_1, \dots, T_n, X_n$ at $t_1, x_1,  \dots, t_n, x_n$ under $P$ is given by
    \begin{align}
         I(t_1<\dots<t_n) \cdot U(0,t_n-,\varphi)\cdot \alpha^{(n)}(0,d(t_1, x_1, \dots, t_n, x_n)). \label{eq: density solution}
    \end{align}
    This expression can be formulated in terms of the canonical compensator of $N$ (see Appendix \ref{appendix: the canonical space of point process realizations}), $\alpha$, a 
    kernel from $\N^{p}$ to $\Rplus \times \I_p$, as follows:
    \begin{align}
        U(0,t_n-,\varphi) &= 
        \prod_{ \stackrel{0 < u \leq t}{\bar \varphi_{u-} = \bar \varphi_u} } \big(1 - \alpha(\varphi,\{ u\}\times \I_p) \big) e^{- \alpha^c(\varphi, (0,t] \times \I_p)} \label{eq: U equation} \\
        \alpha^{(n)}(0,d(t_1, x_1, \dots, t_n, x_n)) &= \alpha\big( 0, dt_1 \times  dx_1\big) \cdot \alpha\big( \varphi|_{t_1}, dt_2 \times dx_2\big) \\
        &\cdots \alpha\big( \varphi|_{t_{n-1}}, dt_n \times dx_n\big) \label{eq:alpha n kernel}  
    \end{align}
    where $\bar \varphi_u := \varphi((0,u]\times \I_p)$, and
    $$ \alpha^c(\varphi, dt \times \I_p) = I\big( \alpha(\varphi, \{ t \} \times \I_p ) = 0\big) \alpha(\varphi, dt \times \I_p) $$
     is the continuous part of $\alpha$.

     Next, we note from \eqref{eq: U equation} that
     \begin{align*}
         U(0,t_n-,\varphi) = U(0,t_1-,\varphi) \cdot U(t_1,t_2-,\varphi)\cdots U(t_{n-1},t_n-,\varphi), 
     \end{align*}
     and that each $U(t_i, t_{i+1}-,\varphi)$ has the representation
     \begin{align*}
     \begin{split}
         U(t_i, t_{i+1}-,\varphi) &= \prod_{t_i < u < t_{i+1}} \big(1 - \alpha (\varphi, \{ u \} \times \I_p) \big) e^{-\alpha^c( \varphi, (t_i, t_{i+1}) \times \I_p )} \\
         &= \prodi_{t_i < u < t_{i+1}} \big(1 - \alpha (\varphi, du \times \I_p) \big)
     \end{split}
     \end{align*}
by definition of the product-integral \cite{gill1990survey}. 
    
    Next, by virtue of $N$ being the identity (we are in the canonical setting) and $\alpha$ is its canonical compensator (actually, it is then just its compensator), we have that
    \begin{align}
        \alpha(\varphi, dt \times dx) = \alpha(N(\varphi), dt \times dx) = \Gamma(\varphi, dt \times dx), \label{eq: canonical compensator distribution}
    \end{align}
    and in particular $\alpha(\varphi, dt \times A) = \alpha(N(\varphi), dt \times A) = \Gamma(\varphi,dt \times A) = \sum_{j\in A} d\Gamma_t^j(\varphi)$ for $P$-a.e. $\varphi$.  
    Combining the equations \eqref{eq:alpha n kernel}-\eqref{eq: canonical compensator distribution} with \eqref{eq: density solution} we 
    get the result \eqref{eq: P regular conditional distribution new} and hence the finite-dimensional distribution shown in \eqref{eq: P density n new}. 

    For integrable $H \in \H^p := \bigvee_{t \geq 0} \H_t^p$, we define $H_n = E_P[H | \H_{T_n}^p]$, a uniformly integrable martingale with respect to $P$ and $\{ \H_{T_n}^p \}_{n\geq 1}$. By Levy's upward theorem, \cite[Theorem 14.2]{williams1991probability}, there is an integrable random variable $H_\infty$ such that $H_n \rightarrow H_\infty$ $P$-a.s. and in $L^1(P)$, and $H_\infty = E_P[H | \H^p_\infty]$ where $\mathcal{H}^p_\infty = \bigvee_{n \geq 1} \mathcal{H}_{T_n}^p$. Since each $\H_{T_n}^p$ is contained in $\H^p$, it follows that $H_\infty$ is $\H^p$-measurable.
    
Next, by the no explosion assumption we have that $\{T_n = \infty \} \uparrow \Omega$, and it follows that
\begin{align}
    H_n I(T_n = \infty) \rightarrow H_\infty \quad P\text{-a.s.} \text{ and in }L^1(P) \label{eq: Hn equality}
\end{align}
    
By Theorem 2.2.14 in \cite{LastBrandt1995marked} and its proof, we have the equality $\mathcal{H}_{T_n}^p \cap \{T_n = \infty\} = \sigma(N|_{T_n}) \cap \{T_n = \infty\} = \mathcal{H}^p \cap \{T_n = \infty\}.$ Thus, by \cite[Lemma 8.3]{Kallenberg2021foundations}, it follows that $E_P[H | \H_{T_n}^p]$ and $E_P[H | \H^p]$ coincide on $\{T_n = \infty\}$. Since $H$ is $\H^p$-measurable, we get, $P$-a.s.,
\begin{align*}
    &H_n I(T_n  = \infty) = E_P[H | \H_{T_n}^p]I(T_n  = \infty) \\
    &= E_P[H | \H^p]I(T_n = \infty) = H I(T_n = \infty).
\end{align*}
It follows from the preceding equality and \eqref{eq: Hn equality} that $H_\infty = H$ $P$-a.s. and 
\begin{align}
    H_n \rightarrow H \quad P\text{-a.s. and in }L^1(P). \label{eq: H conv}
\end{align}

The validity of the identity $F_n = F_\infty \circ \pi_n^{-1}$ is a consequence of Kolmogorov's extension theorem. Next, since $H_n$ is $\H_{T_n}^p$-measurable we have that $H_n \circ \pi_n = H_n$, and we may thus regard $H_n$ as a function on either 
$(\mathbb{R}_+ \times \I_p)^n$ or $(\mathbb{R}_+ \times \I_p)^{\mathbb{N}}$. 
With this identification, change of variables gives
    \begin{align*}
        E_P[ H_n ] &= \int_{(\mathbb R \times \I_p)^n} H_n dF_n \\
                    &= \int_{(\mathbb R \times \I_p)^{n}} H_n d(F_\infty \circ \pi^{-1}_n) \\
                    &= \int_{(\mathbb R \times \I_p)^{\mathbb N}} (H_n \circ \pi_n) dF_\infty  \\
                    &= \int_{(\mathbb R \times \I_p)^{\mathbb N}} H_n dF_\infty. 
    \end{align*}
    The desired result then follows from \eqref{eq: H conv}.
    $\hfill\square$

\subsection{Proof of Lemma \ref{lemma: Q compensator}}
\label{appendix: proof of lemma: Q compensator}
\ \\
Before proving Lemma \ref{lemma: Q compensator}, we express first in Lemma \ref{lemma: Q compensator general} the $(Q, \F_\T)$-compensator of $N$ under the general conditions in Definition \ref{def:strong prime}. After having established the general result, we show how the compensators shown in Lemma \ref{lemma: Q compensator} follow under the regularity conditions \eqref{eq: compensator jump regularity} and \eqref{eq: mcp can comp orthogonality}.

To express the compensator in the general case it is convenient to introduce the notation
\begin{itemize}
    \item $\mathbb{N}^{j,a} := I(\tau^a \leq \cdot , \Delta N_{\tau^a}^j = 1)$,
\end{itemize}
and we let $\mathbb{\Lambda}^{j,a}$ 
denote the ($P,\F_\T$)-compensator of $\mathbb{N}^{j,a}$ 
respectively.
\begin{lemma}[The $(Q, \F_\T)$-compensator of $N$] \label{lemma: Q compensator general}
    If Definition \ref{def:strong prime} \ref{item: likelihood ratio regularity} holds, then 
    \begin{enumerate}[label=\textnormal{(\alph*)}]
        \item \label{item: Q a compensator} 
        $\mathfrak n^a(N)$ defines a $(Q, \F_\T)$-compensator of $N^a$, 
        \item \label{item: Q j compensator}  $\mathcal L^j := \int_0^\cdot \frac{d\Lambda_s^j - d\mathbb \Lambda^{j, a}_s}{1 - \Delta \mathbb \Lambda^a_s} $ defines a $(Q,\F_\T)$-compensator of $N^j$ for $j \neq a$.
    \end{enumerate} 
\end{lemma}

\begin{proof}
Because the intervention $\mathfrak n^a$ is predictable we have that $\mathfrak {n}^a(N)$ is an $\F_\T$-predictable counting process. Thus, the following direct argument shows Lemma \ref{lemma: Q compensator general} \ref{item: Q a compensator}: By definition of $Q$, $W$, and $\tau^a$, and Lemma \ref{lemma: W finite variation} \ref{enum: W jump to zero identity}, we get
\begin{align}
    E_Q[N^a_\tau] = E_P[W_\tau N^a_\tau] = E_P[W_\tau \mathfrak {n}^a_\tau(N)] = E_Q[ \mathfrak {n}^a_\tau(N)], \label{eq: Q a comp}
\end{align}
for each $\F_\T$-optional time $\tau$. Since $\mathfrak {n}^a(N)$ is $\F_\T$-predictable, the preceding equality shows that $\mathfrak {n}^a(N)$ defines a $(Q, \F_\T)$-compensator of $N^a$.

We show in the following that the $(Q, \F_\T)$-compensator of $N^j$ for $j \neq a$ is
    \begin{align}
       \langle N^j \rangle^{Q} = \langle N^j \rangle^{P} - \int_0^{\cdot} \frac{d\mathbb \Lambda^{j,a}_s - \Delta \langle N^j \rangle^{P}_s d\mbbLa{s} }{1 - \Delta \mbbLa{s}}.  \label{eq: Q comp of Nj}
    \end{align}
Here, we use the notation $\langle \cdot \rangle^{P'}$ for the predictable variation with respect to a probability measure $P'$. The filtration it is defined with respect to is fixed in any given derivation and should be clear from the context.


By Lemma \ref{lemma: W finite variation}, the predictable variation process $\langle W, M^j \rangle^P$ exists under $P$ and $\F_\T$, and coincides with $\int_0^\cdot W_{s-} d\langle W, M^j\rangle^P$ for each $j \in \I_d$. 

From \cite[III 8 Theorem 41]{protter} it follows that $Z^j = \int_0^\cdot \frac{1}{W_{s-}}d \langle W, M^j \rangle_s^P$ exists under $Q$, and the 
compensator of $N^j$  under $Q$ and $\F_\T$ is given by $\langle N^j \rangle^P + Z^{j}$. We get, using that $W = \E(\KK)$, $\KK = \int_0^{\cdot} G_s d\mathbb M_s^{a}$, and $G = -\frac{1}{1 - \Delta \mbbLa{}}$,
\begin{align*}
    \int_0^\cdot\frac{1}{W_{s-}}d \langle W, M^j \rangle_s^P = \int_0^{\cdot}\frac{W_{s-}}{W_{s-}}d\langle \KK, M^j \rangle_s^P = \langle \KK, M^j \rangle^P = \int_0^\cdot G_s \langle \mbbMa{}, M^j \rangle_s^P.
\end{align*}
The integrator in the second equality is zero whenever the integrand equals $"\frac{0}{0}"$, which justifies the middle equality. We conclude that
\begin{align}
	\langle N^j \rangle^Q = \langle N^j \rangle^P  + \int_0^\cdot G_s d \langle \mbbMa{}, M^{j} \rangle_s^P.  \label{eq: Na girsanov}
\end{align}
To express the integral on the right-hand side of \eqref{eq: Na girsanov} it is useful to find $[ M^{j}, \mbbMa{} ]$:
\begin{align}
\begin{split}
     [ M^{j}, \mbbMa{} ] &= [ N^j - \langle N^j \rangle^P , \mbbNa{} - \mbbLa{} ] \\
				&= [ N^j , \mbbNa{} ] - [ \langle N^j \rangle^P , \mbbNa{} ]  - [ N^j ,  \mbbLa{} ] + [  \langle N^j \rangle^P , \mbbLa{} ] \\
				&= \mathbb N^{j, a} - \int_0^\cdot \Delta \langle N^j \rangle^P_s d \mbbNa{s} 
				- \int_0^\cdot \Delta \mbbLa{s}  d N^j + \int_0^\cdot \Delta \langle N^j \rangle^P_s d\mbbLa{s}.
\end{split} \label{eq: optional bracket}
\end{align}
This gives, since the predictable bracket defines a compensator of the optional bracket,
\begin{align*}
	 \langle M^{j}, \mbbMa{} \rangle^P &=  \mathbb \Lambda^{j, a} - \int_0^\cdot \Delta \langle N^j \rangle^P_s d\mbbLa{s}. 
\end{align*}
Since $G = - \frac{1}{1 - \Delta \mbbLa{}}$ we get the result
\begin{align*}
	\langle N^j \rangle^Q &= \langle N^j \rangle^P  + \int_0^\cdot G_s d \langle M^{j}, \mbbMa{} \rangle_s^P \\
					&= \langle N^j \rangle^P - \int_0^\cdot \frac{1}{1 - \Delta \mbbLa{s}} (d\mathbb \Lambda^{j, a}_s - \Delta \langle N^j \rangle^P_s d\mbbLa{s}  ).
\end{align*}
Upon arranging terms, this gives \ref{item: Q j compensator}. 
\end{proof}

\subsubsection*{Proof of Lemma \ref{lemma: Q compensator}.}

We verify the simplifications for the non-treatment component, 
i.e., that Lemma \ref{lemma: Q compensator general} \ref{item: Q j compensator} reduces. First, we show that $\mathbb{N}^{j,a} = \mathbb \Lambda^{j,a} = 0$. Condition \eqref{eq: compensator jump regularity} implies that $\mathfrak {n}^a(N)$ shares no jumps with $N^j$ for $j \neq a$; that is, we have $[\mathfrak {n}^a(N), N^j] = 0$ $P$-a.s. for $j \neq a$. By assumption, $N = (N^1, \dots, N^d)$ is a multivariate counting process, so distinct components do not jump simultaneously, i.e. $[N^a, N^j] = 0$ $P$-a.s. for each $j \neq a$. Since $\tau^a$ is an optional time supported on jumps of $N^a$ and $\mathfrak {n}^a(N)$ (by \eqref{eq: tau a}), we obtain $\Delta N^j_{\tau^a} I(\tau^a < \infty) \leq \sum_{i \neq a} [N^a, N^i]_T + [\mathfrak {n}^a(N), N^i]_T = 0$ $P$-a.s., and thus $\mathbb{N}^{j,a} = 0$ $P$-a.s. for each $j \neq a$. It follows that $\mathbb \Lambda^{j,a}=0$ $P$-a.s.

From Lemma \ref{lemma: Q compensator general} we therefore have
$$
\mathcal L^j = \int_0^\cdot \frac{d\Lambda_s^j}{1 - \Delta \mathbb \Lambda^a_s}.
$$
The desired result $\mathcal L^j = \Lambda^j$ follows if we can show that 
\begin{align}
    \Delta \mathbb \Lambda^a_s \Delta \Lambda_s^j = 0 \quad P\text{-a.s.} \label{eq: mbbla lj ort}
\end{align}
for each $s$ and $j \neq a$.

Recall the representation \eqref{eq: mathbb Lambda a representation}:
$$
\mathbb \Lambda^a = ~^{\tau^a}\Lambda^a + ~^{\tau^a} \mathfrak {n}^a(N) - 2 \int_0^{\tau^a \wedge \cdot} \Delta \mathfrak {n}^a_s(N) d\Lambda_s^a.
$$
Since $\mathfrak n^a(N)$ is a predictable counting process satisfying \eqref{eq: compensator jump regularity}, we have $\Delta \mathfrak {n}^a(N) \Delta \Lambda^j = 0$, and by assumption \eqref{eq: mcp can comp orthogonality} we have $\Delta \Lambda^a \Delta \Lambda^j = 0$ for $j \neq a$. Combining this with \eqref{eq: mathbb Lambda a representation}, it follows that \eqref{eq: mbbla lj ort} holds. Therefore,
$$
\int_0^\cdot \frac{d\Lambda_s^j}{1 - \Delta \mathbb \Lambda^a_s} = \Lambda^j.
$$
$\hfill\square$

\subsection{Proof of Theorem \ref{theorem: representation of functionals}}
\ \\ 
    Because $W$ is a mean one nonnegative uniformly integrable $(P, \F_\T)$-martingale (shown in Lemma \ref{lemma: positivity and likelihood ratio equivalence}), it follows that $dQ = W_T dP$ defines a probability measure on $\F_T$. Lemma \ref{lemma: P fidi dist} states that the distribution under $Q$ is determined by the compensator in the canonical setting. This compensator was found in Lemma \ref{lemma: Q compensator} and is given in \eqref{eq: kappa P compensating measure}. The representation \eqref{eq: g-formula new} then follows by combining  Theorem \ref{theorem: ipw},  Lemma \ref{lemma: P fidi dist} and Lemma \ref{lemma: Q compensator}. 
$\hfill\square$

\subsection{Proof of Proposition \ref{prop: characterizations of W}.}

\ \\
\textbf{Proof of the characterization of $W$ though \ref{enum: W dolean unique} and \ref{enum: W greater than 0}.}

\noindent We establish the alternative characterization of $W$ involving the items \ref{enum: W dolean unique}-\ref{enum: W greater than 0}.  We denote by $W^h$ the process which solves
 \begin{align*}
     W^h_t &= 1 + \int_0^t W_{s-}^h d\mathbb K^h_s \\
     \mathbb K^h_t &= \int_0^t h_s d\mbbMa{s}
 \end{align*}
 subject to \ref{enum: W dolean unique} and \ref{enum: W greater than 0}. As $W^h$ is non-negative, the item \ref{enum: W greater than 0} gives that $W^h_{\Nastop} = 0$ and $W^h_{\Nastop-} > 0$ on $\{\Nastop < \infty \}$. Using that $W^h$ solves the given SDE we get that
    \begin{align*}
        W^h_\Nastop &= W^h_{\Nastop-} + W^h_{\Nastop-} \Delta \mathbb K^h_\Nastop \\
        &= W^h_{\Nastop-} + W^h_{\Nastop-} h_\Nastop \Delta \mbbMa{\Nastop} \\
        &= W^h_{\Nastop-} + W^h_{\Nastop-} h_\Nastop ( \Delta \mbbNa{\Nastop} - \Delta \mbbLa{\Nastop} ) \\
        &= W^h_{\Nastop-} + W^h_{\Nastop-} h_\Nastop ( 1 - \Delta \mbbLa{\Nastop} )
    \end{align*}
    on $\{ \Nastop < \infty \}$, where we used the definition of $\Nastop$ and $\mbbNa{}$ in the last line. By rearranging terms we conclude that $W^h$ satisfies \ref{enum: W greater than 0}  only if 
    \begin{align*}
        h_\Nastop = -\frac{1}{1 - \Delta \mathbb \Lambda_\Nastop^a} \quad \text{ on } \{ \Nastop < \infty \}.
    \end{align*}
  Using this fact and the definition of $\mbbNa{}$ we get that
    $$ \int_0^\cdot h_s d\mbbNa{s} = h_{\Nastop} \mbbNa{} =   -\frac{1}{1 - \Delta \mathbb \Lambda_\Nastop^a} \mbbNa{} =  - \int_0^\cdot \frac{1}{1 - \Delta \mbbLa{s}}  d\mbbNa{s}. $$ 
    This gives
    \begin{align*}
        \int_0^\cdot h_s d\mbbMa{s} &= \int_0^\cdot h_s d\mbbNa{s} - \int_0^\cdot h_s d\mbbLa{s} \\
        &= - \int_0^\cdot \frac{1}{1 - \Delta \mbbLa{s}} d\mbbNa{s} - \int_0^\cdot h_s d\mbbLa{s} \\
        &= - \int_0^\cdot \frac{1}{1 - \Delta \mbbLa{s}} d\mbbMa{s} - 
\underbrace{\int_0^\cdot \Big(h_s + \frac{1}{1 - \Delta \mbbLa{s}} \Big) d\mbbLa{s}}_{=:\Upsilon}
    \end{align*}
    where we obtained the last line by adding and subtracting terms. Since $\int_0^\cdot h_s d\mbbMa{s}$ (by assumption) and $\int_0^\cdot \frac{1}{1 - \Delta \mbbLa{s}} d\mbbMa{s}$ (by Lemma \ref{lemma: W finite variation}
    ) are local martingales of finite variation, it follows from the last equality that $\Upsilon$ is also a finite variation local martingale. It is furthermore a predictable process. 
    Since $\Upsilon_0 = 0$, and $\Upsilon$ is a predictable finite variation local martingale, it follows from \cite[III Theorem 15]{protter} that $\Upsilon = 0$ $P$-a.s., and thus $ \int_0^\cdot h_s d\mbbMa{s} = - \int_0^\cdot \frac{1}{1 - \Delta \mbbLa{s}} d\mbbMa{s} = \KK$ $P$-a.s. By uniqueness of solutions of the stochastic exponential, e.g. \cite[II Theorem 37]{protter}, it follows that $W^h = \E(\int_0^\cdot h_s d\mbbMa{s}) = \E(\KK) = W$ $P$-a.s., where $W$ is the process given in \eqref{eq: W unstabilized weight}. 

\ \\
\textbf{Proof that \eqref{eq: W unstabilized weight} solves \eqref{eq: W dolean}.}

In our case $\E(-\mbbLa{})$ has the unique solution $\prodi\limits_{0 < s \leq \cdot} \big(1 - d \mbbLa{s} \big) $; see \cite{Gill1994lecturenotes}. A calculation with indicator processes shows that $I(\Nastop > \cdot) = \E(-\mbbNa{})$, which validates the rightmost equality in \eqref{eq: W unstabilized weight}. We use properties of stochastic exponentials to show that $\E(-\mbbNa{})/\E(-\mbbLa{})$ coincides with the stochastic exponential of $\KK$ in \eqref{eq: K semimart}. 



Since $0 \leq \Delta \mbbLa{} \leq 1$, we have by \eqref{eq: K jump to zero condition} that the process $\sum_{0 < s \leq \cdot} \allowbreak \frac{(\Delta \mbbLa{s})^2}{1 - \Delta \mbbLa{s}} = \allowbreak\int_0^\cdot \allowbreak \frac{\Delta \mbbLa{s}}{1 - \Delta \mbbLa{s}} \allowbreak d\mbbLa{s}$ is of finite variation. Thus, the process 
\footnote{Here, $[Z]^c$ is the path-by-path continuous part of $[Z]$ for a semimartingale $Z$, which is defined by $[Z]^c = [Z] - \sum_{0 < s \leq \cdot} (\Delta Z_s)^2$.}
        \begin{align}
            J = \mbbLa{} + [\mbbLa{}]^c + \sum_{0 < s \leq \cdot} \frac{(\Delta \mbbLa{s})^2}{1 - \Delta \mbbLa{s}}. \label{eq: J process}
            \end{align}
is a well-defined semimartingale, and \cite[Section V, Theorem 63, p 342]{protter} then says that $\E(-\mbbLa{})^{-1} = \E(J)$. By \cite[Section II, Theorem 38, p 86]{protter}, we get 
\begin{align*}
    W &= \E(J) \E( -\mbbNa{} ) = \E(J -\mbbNa{} - [J, \mbbNa{}] ).
\end{align*}
 In our case we have that $[\mbbLa{}]^c = 0$, since $\mbbLa{}$ is a quadratic pure jump semimartingale, e.g. by \cite[II 6 Theorem 26]{protter}. Thus, by the bilinearity of the quadratic variation process we get that
\begin{align*}
    [J,\mbbNa{}] &= \Big[ \mbbLa{} + [\mbbLa{}]^c + \sum_{0 < s\leq \cdot} \frac{(\Delta \mbbLa{s})^2}{1 - \Delta \mbbLa{s}} , \mbbNa{}\Big] \\
    &= \big[ \mbbLa{}, \mbbNa{}\big]  + \big[[\mbbLa{}]^c, \mbbNa{}\big]  +  \Big[\sum_{0 < s\leq \cdot} \frac{(\Delta \mbbLa{s})^2}{1 - \Delta \mbbLa{s}} , \mbbNa{}\Big] \\
    &= \big[ \mbbLa{}, \mbbNa{}\big]  +  \Big[\sum_{0 < s\leq \cdot} \frac{(\Delta \mbbLa{s})^2}{1 - \Delta \mbbLa{s}} , \mbbNa{}\Big].
\end{align*}
This gives
\begin{align*}
    J -\mbbNa{} - [J, \mbbNa{}] &= - \mbbMa{} + [\mbbLa{}]^c 
    + \sum_{0 < s\leq \cdot} \frac{(\Delta \mbbLa{s})^2}{1 - \Delta \mbbLa{s}}  - \big[ \mbbLa{}, \mbbNa{}\big] \\&  -  \Big[\sum_{0 < s\leq \cdot} \frac{(\Delta \mbbLa{s})^2}{1 - \Delta \mbbLa{s}} , \mbbNa{}\Big] \\
    &= - \mbbMa{} + [\mbbLa{}] - \sum_{0 < s \leq \cdot}(\Delta \mbbLa{s})^2 + \sum_{0 < s \leq \cdot} \frac{(\Delta \mbbLa{s})^2}{1 - \Delta \mbbLa{s}}  - \big[ \mbbLa{}, \mbbNa{}\big] \\
    &-  \Big[\sum_{0 < s\leq \cdot} \frac{(\Delta \mbbLa{s})^2}{1 - \Delta \mbbLa{s}} , \mbbNa{}\Big]. 
\end{align*}
The identities
\begin{align*}
    - \sum_{0 < s \leq \cdot}(\Delta \mbbLa{s})^2 + \sum_{0 < s \leq \cdot} \frac{(\Delta \mbbLa{s})^2}{1 - \Delta \mbbLa{s}} &= \sum_{0 < s \leq \cdot} \frac{(\Delta \mbbLa{s})^3}{1 - \Delta \mbbLa{s}}, \\
    \Big[\sum_{0 < s\leq \cdot} \frac{(\Delta \mbbLa{s})^2}{1 - \Delta \mbbLa{s}} , \mbbNa{}\Big] &= \sum_{0< u \leq \cdot} \Delta \sum_{0 < s\leq u} \frac{(\Delta \mbbLa{s})^2}{1 - \Delta \mbbLa{s}}  \Delta \mbbNa{u} \\
    &=  \sum_{0< u \leq \cdot}  \frac{(\Delta \mbbLa{u})^2}{1 - \Delta \mbbLa{u}}  \Delta \mbbNa{u},
\end{align*}
together give that
\begin{align*}
    J -\mbbNa{} - [J, \mbbNa{}] &= - \mbbMa{} + [\mbbLa{}] + \sum_{0 < s \leq \cdot} \frac{(\Delta \mbbLa{s})^3}{1 - \Delta \mbbLa{s}}  - \big[ \mbbLa{}, \mbbNa{}\big]  -  \Big[\sum_{0 < s\leq \cdot} \frac{(\Delta \mbbLa{s})^2}{1 - \Delta \mbbLa{s}} , \mbbNa{}\Big] \\
    &=  - \mbbMa{} + [\mbbLa{}] - \sum_{0 < s \leq \cdot} \frac{(\Delta \mbbLa{s})^2}{1 - \Delta \mbbLa{s}} \Delta \mbbMa{s}  - \big[ \mbbLa{}, \mbbNa{}\big] \\
    &=  - \mbbMa{} - [\mbbLa{}, \mbbMa{}] - \sum_{0 < s \leq \cdot} \frac{(\Delta \mbbLa{s})^2}{1 - \Delta \mbbLa{s}} \Delta \mbbMa{s}.
\end{align*}
Next, using that $\mbbLa{}$ is a quadratic pure jump semimartingale, we get
\begin{align*}
   [\mbbLa{}, \mbbMa{}] + \sum_{0 < s \leq \cdot} \frac{(\Delta \mbbLa{s})^2}{1 - \Delta \mbbLa{s}} \Delta \mbbMa{s} &= \sum_{0 < s \leq \cdot}\Delta \mbbLa{s} \Delta \mbbMa{s} + \sum_{0 < s \leq \cdot} \frac{(\Delta \mbbLa{s})^2}{1 - \Delta \mbbLa{s}} \Delta \mbbMa{s} \\
   &= \sum_{0 < s \leq \cdot} \frac{\Delta \mbbLa{s}}{1 - \Delta \mbbLa{s}} \Delta \mbbMa{s} = \int_0^\cdot \frac{\Delta \mbbLa{s}}{1 - \Delta \mbbLa{s}} d\mbbMa{s}.
\end{align*}
Finally, since 
$$\mbbMa{} + \int_0^\cdot \frac{\Delta \mbbLa{s}}{1 - \Delta \mbbLa{s}} d\mbbMa{s} = \int_0^\cdot \Big( 1 + \frac{\Delta \mbbLa{s}}{1 - \Delta \mbbLa{s}} \Big) d\mbbMa{s} = \int_0^\cdot \frac{1}{1 - \Delta \mbbLa{s}} d\mbbMa{s}, $$
we get the advertised result.



\ \\
\textbf{Derivation of \eqref{eq: W discrete and cont expression}.}

    We establish the expression \eqref{eq: W discrete and cont expression}. Since $\mbbLa{}$ is cadlag and predictable it follows that the jump times $\{ \boldsymbol \sigma_k\}_k$ of $\mbbLa{}$ are predictable as well. \footnote{This follows by application of \cite[Theorem 2.1.30]{LastBrandt1995marked}. This theorem says that if $Y$ is a nonnegative, nondecreasing, right-continuous, predictable process, then for any $c \in \mathbb R$, $\inf\{t>0 | Y_t \geq c\}$ is a predictable time. Applying this result to the discontinuous part $Z = \mbbLa{} - \mathbb \Lambda^{a,c}$ where $\mathbb \Lambda^{a,c}$ is the continuous part of $\mathbb \Lambda^a$ and a sequence of constants $\{ c_n \}_n$, $c_n > 0$ with $c_n \downarrow 0$, we get a sequence of predictable  times
    \begin{align*}
        \tau_n^1 := \inf\{ t>0 | Z_t \geq c_n \}.
    \end{align*}
    By construction, $\wedge_n \tau_n^1 = \boldsymbol\sigma_1$, where $\boldsymbol\sigma_1$ is the first jump time of $Z$ (and hence of $\mathbb \Lambda^a$). Since each $\tau_n^1$ is predictable, $\wedge_n \tau_n^1 = \boldsymbol\sigma_1$ is also a predictable time; see, e.g., \cite[I 2 Proposition 2.9]{JacodShiryaev}. Because $\boldsymbol\sigma_1$ is a predictable time, it follows that the process started at $\boldsymbol\sigma_1$, $~_{\boldsymbol\sigma_1}Z_t := I(t > \boldsymbol\sigma_1) (Z_t - Z_{\boldsymbol\sigma_1})$, satisfies the conditions of \cite[Theorem 2.1.30]{LastBrandt1995marked}. We can thus apply the previous argument to $~_{\boldsymbol\sigma_1}Z$ to show that $\boldsymbol\sigma_2$ is a predictable time, and by induction that $\boldsymbol\sigma_k$ is a predictable time for each $k$.} 
    We thus get, by the definition of the product-integral,
    \begin{align*}
        \prodi_{s \leq t}(1 - d\mbbLa{s}) &= e^{- \mathbb \Lambda^{a, c}_t}\prod_{s \leq t}(1 - \Delta \mbbLa{s}) = e^{- \mathbb \Lambda^{a, c}_t}\prod_{\boldsymbol \sigma_k \leq t}(1 - \Delta \mbbLa{\boldsymbol \sigma_k}) \\
        &= e^{- \mathbb \Lambda^{a, c}_t}\prod_{\boldsymbol \sigma_k \leq t}P(\mbbNa{\boldsymbol \sigma_k}=0 | \F_{\boldsymbol \sigma_k-}),
    \end{align*}
    where we used that each $\boldsymbol \sigma_k$ is predictable and \cite[Equation (4.1.16)]{LastBrandt1995marked}  in the last equality, which gives the result.

\ \\
\textbf{Proof of \eqref{eq: Gill likelihood} and the SDE for predictable interventions.}

The representation of $\mbbLa{}$ in terms of $\Lambda^a$ and $\na{}(N)$ was shown in Lemma \ref{lemma: mathbb representations}. Inserting this representation into \eqref{eq: W unstabilized weight} leads to \eqref{eq: Gill likelihood}. 

We establish the alternative representation $\KK = -\int_0^{\Nastop \wedge \cdot} \frac{1}{1- \Delta \na{s}(N) - \Delta \Lambda_s^a} dM_s^a$. By combining Lemma \ref{lemma: mathbb representations} with the expression for $\KK$ we get that $\KK \allowbreak = \allowbreak - \int_0^{\Nastop \wedge \cdot} \allowbreak \frac{1 - 2 \Delta \na{s}(N)}{ 1 - \Delta \Lambda_s^a  - \Delta \na{s}(N) + 2  \Delta \na{s}(N)  \Delta\Lambda_s^a}  \allowbreak dM_s^a$. Algebraic manipulation of the integrand in this equation gives
\begin{align}
    \begin{split}
        \frac{1 - 2 \Delta \na{s}(N)}{ 1 - \Delta \na{s}(N) - (1 - 2  \Delta \na{s}(N) ) \Delta\Lambda_s^a} &= \frac{1}{ \frac{1 - \Delta \na{s}(N)}{1 - 2 \Delta \na{s}(N)}  - \Delta\Lambda_s^a} \\
        &=  \frac{1}{ 1 - \Delta \na{s}(N)  - \Delta\Lambda_s^a},
    \end{split} \label{eq: split algebraic compensator}
\end{align}
 where the final equality can be seen by checking the two cases $\Delta \na{s}(N) = 1$ and $\Delta \na{s}(N) = 0$.  $\hfill\square$

\subsection{Proof of Proposition \ref{proposition: relation to existing assumptions}}
\ \\
\textbf{Equivalence of exchangeability conditions.}

    Since the compensators are piecewise constant, the only possible jump times of the counting processes are at the jump times of their compensators. Since the regime $\na{}$ is predictable, we can use the representation in Lemma \ref{lemma: mathbb representations} and the a.s. representation
    $$ \mbbMa{} = \int_0^{\Nastop \wedge \cdot} (1 - 2 \Delta \na{s}(N)) dM_s^a. $$
    It follows from this that $\mbbMa{}$ is a martingale with respect to a superfiltration of $\F_\T$ if and only if $~^{\Nastop}M^a$ is a martingale with respect to that filtration. 
    A key fact in what follows is that jumps of compensators at predictable times (and thus in particular deterministic times) admit explicit representations, see e.g. \cite[Theorem 4.1.7]{LastBrandt1995marked};  we have that $~^{\Nastop}\Lambda^a$ is a $(P,\F_\T)$-compensator of $~^{\Nastop}N^a$ if and only if 
    \begin{align}
        \Delta ~^{\Nastop}\Lambda^a_{\theta_{k}} = E_P[\Delta ~^{\Nastop}N^a_{\theta_{k}}  | \F_{\theta_{k}-} ] \quad P\text{-a.s.} \label{eq: piecewise continuous compensator jumps}
    \end{align}
    for each $k$. We similarly have that the exchangeability condition \eqref{eq: exchan disc} in this case is equivalent to $\Delta ~^{\Nastop}\Lambda^a_{\theta_{k}} = E_P[\Delta ~^{\Nastop}N^a_{\theta_{k}} | \F_{\theta_{k}-} \vee \tilde \F^{\Idmina}_T ] $ $P\text{-a.s.}$ for each $k$, where $\tilde \F_T^{\Idmina} = \sigma(\tilde N^1|_T, \dots, \tilde N^{a-1}|_T, \tilde N^{a+1}|_T, \dots, \tilde N^d|_T)$. It follows that \eqref{eq: exchan disc} is equivalent to 
    \begin{align}
        E_P[\Delta ~^{\Nastop}N^a_{\theta_{k}}  | \F_{\theta_{k}-} ] = E_P[\Delta ~^{\Nastop}N^a_{\theta_{k}}  | \F_{\theta_{k}-} \vee \tilde \F_T^{\Idmina} ] \quad P\text{-a.s. for } k=1,\dots,K. \label{eq: piecewise continuous compensator jumps hard}
    \end{align} 
    
    Define $C_k = \{ \bar A_k = \bar a_k \}$ for $k \geq 1$ and $C_0 = \Omega$. We show that the identity \eqref{eq: piecewise continuous compensator jumps hard} is identical to 
    \begin{align}
        P_{C_{k-1}}( A_k = 1 | \bar L_{k}) &= P_{C_{k-1}}( A_k = 1 | \bar L_k,  \underline{\tilde L}_{k+1} ) \quad P_{C_{k-1}} \text{-a.s.}  
        \label{eq: PAC equality}
    \end{align}
    for each $k$, where $P_{C_k}(\cdot) = P(C_{k} \cap \cdot)/P(C_{k})$ is the $C_k$-conditional-probability measure, which is defined since $P(C_k)>0$.     It is sufficient to show this equivalence because, since $A_k$ is binary, the statement \eqref{eq: PAC equality} is equivalent to
    $$ A_k \indep_{P_{C_{k-1}}} \underline{\tilde L}_{k+1} | \bar L_k, $$
    which is equivalent to \eqref{eq: RCISTG indep}.

    By well known identities of point process filtrations,
    \cite[Theorem 2.2.15]{LastBrandt1995marked}, we have that 
    \begin{align*}
        \F_{\theta_{k}-} &= \sigma( N |_{\theta_k-}) \vee \sigma ( \theta_k) \\
        &=  \sigma( \Delta N_{\ell_1^1}^1, \dots, \Delta N_{\theta_{1}}^a, \Delta N_{\ell_2^1}^1, \dots, \Delta N_{\theta_2}^a, \Delta N_{\ell_3^1}^1, \dots, \Delta N_{\ell_k^{d-1}}^{k} ). 
    \end{align*}
    The last equality follows because the trajectories of $N|_{\theta_k-}$ are fully determined by the variables $\Delta N_{\ell_1^1}^1, \dots, \Delta N_{\theta_{1}}^a, \dots, \Delta N_{\ell_k^{d-1}}^{k}$, as the times $\{ \theta_k \}_k$, $\{\ell_k^j \}_{k,j}$ are deterministic. 
    Recalling the definitions of the variables $\{A_k\}$, $\{L_k^j\}$, and $\{L_k\}$, we thus get 
    \begin{align}
        \begin{split}
            \F_{\theta_{k}-} &= \sigma( \Delta N_{\ell_1^1}^1, \dots, \Delta N_{\theta_{1}}^a, \Delta N_{\ell_2^1}^1, \dots, \Delta N_{\theta_2}^a,\Delta N_{\ell_3^1}^1, \dots, \Delta N_{\ell_k^{d-1}}^{k} ) \\
        &= \sigma( L_1, A_1, \dots,L_{k-1},A_{k-1}, L_{k} ) = \sigma( \bar L_k, \bar A_{k-1} ).
        \end{split}
    \label{eq: F alpha minus filtration representation}
    \end{align}
    Next, since the jump times of $\na{}(N)$ are contained in $\{ \theta_k \}_k$, we have that $\{ \Nastop \geq \theta_k \} = \{ \bar A_{k-1} = \bar a_{k-1} \} = C_{k-1}$. (In words, these equalities say that if the observed regime is followed through treatment visit $k-1$, it is at least followed until the next treatment decision is being recorded).  This leads to the identities
    \begin{align}
    \begin{split}
        I_{C_{k-1}} \Delta ~^{\Nastop}N^a_{\theta_{k}} &= I_{C_{k-1}}\Delta N^a_{\theta_{k}} = I_{C_{k-1}}I(A_k=1), \\
        I_{C_{k-1}^c} \Delta ~^{\Nastop}N^a_{\theta_{k}} &= 0, \\
        \bar L_k &= \bar{\tilde L}_k \quad P_{C_{k-1}} \text{-a.s.,}
    \end{split} \label{eq: Na jump identities}
    \end{align}
    where the last line holds due to consistency. We get, using the identities in \eqref{eq: Na jump identities} and that $C_{k-1}$ is measurable with respect to $\sigma(\bar A_{k-1})$,  the $P$-a.s. equalities
\begin{align}
\begin{split}
    E_P\big[ \Delta ~^{\Nastop}N^a_{\theta_{k}}  | \F_{\theta_{k}-} \big] &= E_P\big[ \Delta ~^{\Nastop}N^a_{\theta_{k}}  | \bar L_k, \bar A_{k-1} \big] \\
    &= E_P\big[ \Delta ~^{\Nastop}N^a_{\theta_{k}}  | \bar L_k, \bar A_{k-1} \big] \big( I_{C_{k-1}} + I_{C_{k-1}^c} \big) \\
    &= E_P\big[ \Delta ~^{\Nastop}N^a_{\theta_{k}}  | \bar L_k, \bar A_{k-1} \big]I_{C_{k-1}} \\
    &= P( A_k=1  | \bar L_k, \bar A_{k-1} )I_{C_{k-1}} \\
    &= P( A_k=1  | \bar L_k, \bar A_{k-1} = \bar a_{k-1} )I_{C_{k-1}} \\
    &= P_{C_{k-1}}( A_k=1  | \bar L_k )I_{C_{k-1}}.
\end{split} \label{eq: disc comp simple}
\end{align} 
    In the two last lines we used well-known identities of partial conditional expectation and conditional expectation with respect to event-conditional-probability measures such as $P_{C_{k-1}}$; see e.g. \cite[Remark 14.31 and Theorem 14.33]{nagel2017probability}. 
    By similar arguments as those leading to \eqref{eq: F alpha minus filtration representation} we can identify $\tilde \F^{\Idmina}_T$ with the $\sigma$-algebra generated by the variables $\{\tilde L_k\}_k$;
    \begin{align*}
        \tilde \F^{\Idmina}_T &= \tilde \F_{\ell_{{K+1},{d-1}}}^{\Idmina} \\
        &= \sigma( \tilde L_1^1, \dots, \tilde L_1^{d-1} ,  \tilde L_2^1, \dots, \tilde L_2^{d-1} ,\dots , \tilde L_{K+1}^{d-1} ) \\
        &= \sigma(\tilde L_1, \dots, \tilde L_{K+1}) = \sigma(\bar {\tilde L}_{K+1})
    \end{align*}
     Using similar arguments as in \eqref{eq: disc comp simple} we therefore get
     \begin{align}
     \begin{split}
         E_P\big[\Delta ~^{\Nastop}N^a_{\theta_{k}} &| \F_{\theta_{k}-} \vee \tilde \F^{\Idmina}_T \big] \\
        &=  E_P\big[ \Delta ~^{\Nastop}N^a_{\theta_{k}}  | \bar L_k, \bar A_{k-1}, \bar{\tilde { L}}_{K+1} \big] \big( I_{C_{k-1}} + I_{C_{k-1}^c} \big)\\
        &= E_P\big[ \Delta ~^{\Nastop}N^a_{\theta_{k}}  | \bar L_k, \bar A_{k-1}, \bar{\tilde { L}}_{K+1} \big]I_{C_{k-1}} \\
        &= P_{C_{k-1}}( A_k=1 | \bar L_k, \bar{\tilde { L}}_{K+1}) I_{C_{k-1}}.
     \end{split} \label{eq: partial condexp hard}
    \end{align}
    Since $\bar{ \tilde L}_k = \bar L_k$ $P_{C_{k-1}}$-a.s. (see \eqref{eq: Na jump identities}), we have in particular that
    \begin{align*}
        \begin{pmatrix}
            \bar L_k \\ \bar {\tilde L}_k \\ \underline{ \tilde L}_{k+1}
        \end{pmatrix} = \begin{pmatrix}
            \bar L_k \\ \bar {L}_k \\ \underline{ \tilde L}_{k+1}
        \end{pmatrix} \quad P_{C_{k-1}} \text{-a.s.}
    \end{align*}
    From simple checks of the definition of conditional expectation we can deduce from this that
    \begin{align*}
        P_{C_{k-1}}( A_k=1 | \bar L_k, \bar{\tilde { L}}_{K+1} ) = P_{C_{k-1}}( A_k=1 | \bar L_k, \underline{\tilde { L}}_{k+1} ) \quad P_{C_{k-1}} \text{-a.s.}
    \end{align*}
    Combining this result with \eqref{eq: partial condexp hard}, using     \eqref{eq: disc comp simple}, we see that \eqref{eq: piecewise continuous compensator jumps hard} holds if and only if 
    \begin{align*}
    \begin{split}
    P_{C_{k-1}}( A_k=1  | \bar L_k )I_{C_{k-1}} 
        &= P_{C_{k-1}}(  A_k=1 | \bar L_k, \underline{\tilde { L}}_{k+1}) I_{C_{k-1}}, \quad P\text{-a.s.}
    \end{split}
    \end{align*}
    which is equivalent to the desired equality \eqref{eq: PAC equality}. 
  
\
\\ 
\textbf{Equivalence of positivity conditions.}

     Under the premises of the proposition, the paths of $\mbbLa{}$ are piecewise constant with jumps in $\{ \theta_k\}_k$. Thus, $\int_0^\cdot \frac{d\mbbLa{s}}{1 - \Delta \mbbLa{s}} = \sum_{k: \theta_k \leq \cdot} \frac{\Delta \mbbLa{\theta_k}}{1 - \Delta \mbbLa{\theta_k}}$. Since the number of terms in the sum is finite, \eqref{eq: K jump to zero condition} holds if and only if
    \begin{align}
        P( \Delta \mbbLa{ \theta_k } < 1) = 1 \text{ for } k=1, \dots, K+1. \label{eq: pos2 equivalence}
    \end{align}
    From the observation that $\{\Delta \mbbNa{\theta_k} = 1 \} = \{ \Nastop = \theta_k \}  = \{ A_k \neq a_k, \bar A_{k-1} = \bar a_{k-1} \}$, using that $\Delta \mbbLa{\theta_k} = E_P[ \Delta \mbbNa{\theta_k} | \F_{\theta_k-} ]$ (using again that the $\theta_k$'s are deterministic, and thus predictable, and e.g. \cite[Theorem 4.1.7]{LastBrandt1995marked}), we get that
    \begin{align}
        \Delta \mbbLa{\theta_k} = P( \Delta\mbbNa{\theta_k} = 1 | \F_{\theta_k-} ) = P( A_k \neq a_k | \bar L_k, \bar A_{k-1} )I(\bar A_{k-1} = \bar a_{k-1}), \label{eq: pos discrete}
    \end{align}
    where we used \eqref{eq: F alpha minus filtration representation}. It follows that \eqref{eq: pos2 equivalence} holds if and only if  
    \begin{align*}
        P(A_k = a_k | \bar L_k, A_{k-1}) > 0 \quad P\text{-a.s. on } \{ \bar A_{k-1} = \bar a_{k-1} \} 
    \end{align*}
    for $k = 1, \dots, K+1$, which is the condition \eqref{eq: hr condition}.

    We show that Definition \ref{def:strong prime} \ref{item: likelihood ratio regularity} holds. If \eqref{eq: hr condition} holds (i.e. that \eqref{eq: K jump to zero condition} holds, under the assumptions of this proposition), then $\KK$ in \eqref{eq: W dolean} is a local martingale by Lemma \ref{lemma: W finite variation}. In particular, it is a well-defined semimartingale, and \eqref{eq: W dolean} is well-defined, with the solution given in Proposition \ref{prop: characterizations of W}, (e.g. by \cite[II 8 Theorem 37]{protter}),
    $$ W_t = \frac{I(\Nastop > t)}{ \prodi_{0< s \leq t} (1 - \mbbLa{s}) }. $$
    This is also a local martingale, as highlighted in Lemma \ref{lemma: W finite variation}. Since the only contribution to the product are at the times $\{ \theta_{k} \}$, which are discontinuity points of $\Lambda^a$, we get by \eqref{eq: pos discrete}
    \begin{align}
        W_t = \frac{I(\bar A_{k_t^*} = \bar a_{k_t^*})}{ \prod_{j=1}^{k_t^*} P(A_j = a_j | \bar L_j, \bar A_{j-1}) },  \label{eq: W star}
    \end{align}
    where $k_t^* = \sup \{ k | \theta_k \leq t \}$, and where we used \eqref{eq: F alpha minus filtration representation} once more. $W_t$ is almost surely finite under \eqref{eq: hr condition}. Writing out $E_P[W_t]$ in terms of the law of the observed variables, this is a finite sum of terms which are finite (otherwise $W_t$ would not have been almost surely finite), and thus $W_t$ is integrable for each $t$. In particular, $W_{\theta_k}$ is integrable for each $k$, which implies that 
    $\sup_{s \leq t}W_s = \max\{W_{\theta_k} | \theta_k \leq t \}$ is integrable for each $t$. Since $W$ is a local martingale, it follows from this observation and \cite[I 6 Theorem 51]{protter} that it is a true martingale. Thus $E_P[W_t] = E_P[W_0] = 1$ for each $t \in \T$, and \eqref{eq: W likelihood ratio regularity} holds.

\ \\
    \textbf{Reduction of the data-generating law.}

    Since $\na{}$ is predictable, we have from \eqref{eq: predictable intervention tilde a compensator} that $\cfproc{a} = \na{}(\cfproc{})$ $P$-a.s. With the notation and regime stated in the proposition, we get that $\tilde A_k = a_k$ $P$-a.s. for each $k$.
    
    Next, we use that $\tilde  \Lambda^j$ defines a compensator of $\cfproc{j}$ with respect to $\tilde \F_\T$, and that $\Lambda^j$ defines a compensator of $N^j$ with respect to $\F_\T$. Since each $\ell_k^j$ is deterministic (and thus predictable), we have that, $P$-a.s.,
    \begin{align}
    \begin{split}
        \Delta \Lambda_{\ell_k^j}^j &= E_P[ \Delta N_{\ell_k^j}^j | \F_{\ell_k^j-} ] = P( L_k=1 | \F_{\ell_k^j-} ), \\
        \Delta \tilde  \Lambda_{\ell_k^j}^j &= E_P[ \Delta \cfproc{j}_{\ell_k^j} | \tilde{ \F}_{\ell_k^j-} ] = P( \tilde L_k=1 | \tilde{ \F}_{\ell_k^j-} ),
    \end{split}\label{eq: ell k j compensator}
    \end{align}
    where the rightmost equalities follow since we have defined $L_k^j = \Delta N_{\ell_k^j}^j$ and $\tilde L_k^j = \Delta \tilde N_{\ell_k^j}^j$. 
    Arguing as in \eqref{eq: F alpha minus filtration representation} we find that $\F_{\ell_k^j-} = \sigma( L_{k,<j}, \bar A_{k-1})$, and similarly $\tilde \F_{\ell_k^j-} = \sigma( \tilde L_{k,<j}, \bar {\tilde A}_{k-1})$. It follows from \eqref{eq: ell k j compensator} and the definition of $f_k^j$ that 
    \begin{align*}
        P(\tilde L_k^j = 1 | \tilde L_{k,<j}, \bar {\tilde A}_{k-1}) &= P(\tilde L_k^j = 1 | \tilde \F_{\ell_k^j-}) =  \Delta \tilde  \Lambda_{\ell_k^j}^j \\
        &= f_k^j( x,y )|_{x=\tilde L_{k,<j}, y=\bar {\tilde A}_{k-1}} = f_k^j( \tilde L_{k,<j},\bar {\tilde A}_{k-1} ),
    \end{align*}
    where the equalities in the first line hold almost surely and the equalities in the last line hold by definition. 
$\hfill\square$

\subsection{Proof of Theorem \ref{theorem: multiple marked interventions}}
\ \\
 
By the regularity conditions \eqref{eq: can comp j h}-\eqref{eq: interv interv}, Proposition \ref{proposition: mcp construction of observed and potential outcomes marked} in Appendix \ref{appendix: construction of potential outcomes} guarantees that there is a probability space supporting observed and potential outcome processes with the respective correct laws such that the consistency condition \ref{enum: mpp consistency} in Theorem \ref{theorem: multiple marked interventions} holds, and for each $j\in J$, the condition \eqref{eq: j comp measure} of Lemma \ref{lemma: MPP exchangeability 2026} is satisfied. \footnote{Choose, for instance the randomizers in \eqref{eq: uniform random variables marked} in Proposition \ref{proposition: mcp construction of observed and potential outcomes marked} to be mutually independent, so that the independence assumptions \ref{enum: mutual indep multi marked} - \ref{enum: exch indep multi marked} in Proposition \ref{proposition: mcp construction of observed and potential outcomes marked} hold.} Next, since \eqref{eq: j comp measure} holds for each $j \in J$, the exchangeability condition \ref{enum: mpp exchangeability} of Theorem \ref{theorem: multiple marked interventions} also holds, as shown in Lemma \ref{lemma: MPP exchangeability 2026}. Thus, under the theorem's regularity conditions (assumptions \eqref{eq: can comp j h}-\eqref{eq: interv interv}), 
the consistency and exchangeability conditions in Theorem \ref{theorem: multiple marked interventions} \ref{enum: mpp consistency}-\ref{enum: mpp exchangeability} can be made to hold simultaneously; that is, these conditions do not impose hidden restrictions on the observed data law.

The formula \eqref{eq: g-formula multiple interventions}, i.e.
\begin{align*}
    E_P[\tilde{Y}_t] = E_P[\mathcal{E}(\mathbb{K}^J)_t Y_t] = E_Q[Y_t],
\end{align*}
with $dQ = \mathcal{E}(\mathbb{K}^J) dP$, is derived using the same arguments as used in Theorem \ref{theorem: ipw}. Specifically, under the conditions \eqref{eq: J int positivity}, and \ref{enum: mpp consistency}, \ref{enum: mpp exchangeability} and \ref{enum: mpp likelihood ratio regularity} of Theorem \ref{theorem: multiple marked interventions}, the process $\mathcal{E}(\mathbb{K}^J)$ satisfies properties parallel to those established for $W = \mathcal{E}(\mathbb{K})$ in Lemma \ref{lemma: W finite variation} and Lemma \ref{lemma: positivity and likelihood ratio equivalence}, with $\mathbb{K}^J$, $\mathbb{N}^J$, $\mathbb{\Lambda}^J$, and $\tau^J$ playing the roles of $\KK$, $\mathbb{N}^a$, $\mathbb{\Lambda}^a$, and $\tau^a$, respectively. Using these properties, the proof of \eqref{eq: g-formula multiple interventions} follow the lines of the proof of Theorem \ref{theorem: ipw}. 

As in Lemma \ref{lemma: P fidi dist}, to express the distribution under $Q$, we first need to find the $(Q, \F_\T)$-compensator of $N$. We make use of similar techniques as in Lemma \ref{lemma: Q compensator} to derive the compensator.

First, for $j \in J$, 
we have by definition of $dQ = \E(\mathbb K^J) dP, \mathbb K^J$, and $\tau^J$, that $\mathfrak n^j(L, N) = N^j$ $Q$-a.s. This can be shown by analogous techniques as those in Lemma \ref{lemma: tau a is infty under Q}. It follows from predictability that $\mathfrak n^j(L, N)$ defines a $(Q, \F_\T)$-compensating measure of $N^j$. This can formally be established by arguing as in \eqref{eq: Q a comp}, with $dQ = \E(\mathbb K^J)_T dP$, and $N^j((0,\cdot] \times D)$ in place of $N^a$ and $\mathfrak n^j(L, N,(0,\cdot] \times D)$ in place of $\mathfrak n^a(N)$ in \eqref{eq: Q a comp}, for each $D \in \X^j$.

It remains to find the $(Q, \F_\T)$-compensating measure of $N^i$ for $i \not\in J$. Analogous to the notation used to state Lemma \ref{lemma: Q compensator general}, we define
\begin{align*}
    \mathbb{N}^{i,J}((0,t] \times \cdot) = I(\tau^J \leq t , N^i( \{ \tau^J \} \times \cdot) = 1),
\end{align*}
and denote by $\mathbb{\Lambda}^{i,J}$ the compensating measure of $\mathbb{N}^{i,J}$ with respect to $P$ and $\F_\T$.

For $i \not\in J$ and $D \in \mathcal{X}^i$, we define $N^{i,D}_t := N^i((0,t]\times D)$. Repeating the Girsanov-type calculations that led to \eqref{eq: Na girsanov} (see Lemma \ref{lemma: Q compensator general} for details), we obtain
\begin{align*}
    \langle N^{i,D} \rangle^Q &= \langle N^{i,D} \rangle^P + \int_0^\cdot G_s d \langle M^{i,D} , \mathbb{M}^J \rangle_s^P,
\end{align*}
where $M^{i,D} := N^{i,D} - \langle N^{i,D} \rangle^P$, $G = -\frac{1}{1-\Delta \mathbb{\Lambda}^J}$, $\mathbb{M}^J := \mathbb{N}^J - \mathbb{\Lambda}^J$, and the predictable brackets are defined with respect to the filtration $\F_\T$. We furthermore get that 
\begin{align*}
    \langle M^{i,D} , \mathbb{M}^J \rangle^P &= \langle \mathbb{N}^{i,J}((0,\cdot] \times D) \rangle^P  - \int_0^\cdot  \Delta \mathbb{\Lambda}_s^J d\langle N^{i,D} \rangle^P_s \\
    &= \mathbb{\Lambda}^{i,J}((0,\cdot] \times D)  - \int_0^\cdot  \Delta \mathbb{\Lambda}_s^J d \Lambda^i((0,s] \times D).
\end{align*}
Using the preceding line and the definition of $G$, we conclude that
\begin{align*}
    \langle N^{i,D} \rangle^Q &= \langle N^{i,D} \rangle^P - \int_0^\cdot \frac{  d\mathbb{\Lambda}^{i,J}((0,s] \times D)  -   \Delta \mathbb{\Lambda}_s^J d \Lambda^i((0,s] \times D) }{1 - \Delta \mathbb{\Lambda}_s^J},
\end{align*}
which determines the $(Q,\F_\T)$-compensating measure of $N^i$. As $\L^i( (0,\cdot ] \times D ) = \langle N^{i,D} \rangle^P = \Lambda^i((0,\cdot] \times D)$ we get by algebraic manipulation that
\begin{align*}
    \L^i(dt \times dx) = \frac{\Lambda^i(dt \times dx) - \mathbb{\Lambda}^{i,J}(dt \times dx)}{1 - \Delta \mathbb{\Lambda}_t^J},
\end{align*}
where $\L^i$ is the $(Q, \F_\T)$-compensator of $N^i$. 

We show that $\mathbb{\Lambda}^{i,J}=0$ under \eqref{eq: intervention can comp}. For $j \in J$, define $A^j_t := \mathfrak n^j(L, N, (0,t] \times X^j)$. Condition \eqref{eq: intervention can comp} implies that $A^j$ shares no jumps with $\bar N^i$ for $i \not\in J$, i.e. $[A^j, \bar N^i] = 0$ $P$-a.s. for each $j \in J$. Since $N$ is an MPP, $(\bar N^1, \dots,\bar N^d)$ is a multivariate counting process, so distinct components do not jump simultaneously, i.e. $[\bar N^j, \bar N^i] = 0$ $P$-a.s. for $j \neq i$. Because $\tau^J$ is an optional time supported on the jump times of $\bar N^j$ and $A^j$ for $j \in J$, we obtain the bound $\Delta \bar N^i_{\tau^J} I(\tau^J < \infty) \leq \sum_{j \in J} \sum_{z \not\in J} [\bar N^j, \bar N^z]_T + [A^j, \bar N^z]_T = 0$ $P$-a.s., and thus $\mathbb{N}^{i,J} = 0$ $P$-a.s. for each $i \not\in J$. It follows that $\mathbb \Lambda^{i,J}=0$ $P$-a.s., where this equality also holds $Q$-a.s. because $Q \ll P$. We are left with
\begin{align}
    \L^i(dt \times dx) &= \frac{\Lambda^i(dt \times dx) }{1 - \Delta \mathbb{\Lambda}_t^J}. \label{eq: N Q i comp}
\end{align}
Finally, we use that $\mathbb N^J$ satisfies the bound $\mathbb N^J \leq \sum_{j \in J} \sum_{k} \mathbb N^{j,k}$ where $\mathbb N^{j,k}$ are as in \eqref{eq: Njk representation}. The compensator $\mathbb{\Lambda}^J$ thus satisfies
$$
\Delta\mathbb{\Lambda}^J_t \leq \sum_{j \in J}  \sum_{k} \Delta\mathbb \Lambda^{j,k}_t \quad P\text{-a.s.},
$$
where, using the decomposition in \eqref{eq: Njk representation}, $\Delta \mathbb \Lambda^{j,k}_t = ~^{\tau^{j,k}}\Lambda^{j}(\{ t \} \times \{ x_k^j \}) + ~^{\tau^{j,k}}\mathfrak{n}^{j}(L, N,\{ t \} \times \{ x_k^j \}) - 2\int_0^{\tau^{j,k} \wedge t} \mathfrak{n}^{j}(L, N,\{ s \} \times \{ x_k^j \}) \Lambda^{j}( ds \times \{ x_k^j \})$ for some enumeration $\{ x_k^j \}_k$ of $X^j$. Since $\Lambda^{j} = \alpha^j(L, N)$, the regularity conditions \eqref{eq: can comp j setminus J}-\eqref{eq: intervention can comp} imply that $\mathbb \Lambda^{j,k}$ does not share jumps with $\Lambda^i$ for $i \not \in J$. Consequently,
$$
\Lambda^i(\{t\} \times dx) \Delta\mathbb{\Lambda}^J_t = 0 \quad P\text{-a.s.}
$$
which combined with \eqref{eq: N Q i comp} gives the desired result
\begin{align*}
    \L^i(dt \times dx) &= \Lambda^i(dt \times dx).
\end{align*}
We now make sense of the integral in \eqref{eq: thm 3 g formula}. Since $X$ is Borel it follows that $\mathcal{N}_T^X$ is Borel, and there consequently is a kernel $K^Q$ from $S$ to $\mathcal{N}_T^X$ so that $Q\big( (L,N) \in dl \times d\varphi \big) = K^Q(l, d\varphi) F_L(dl)$. $K^Q(l, d\varphi)$ is a regular version of the conditional distribution of $N$ given $L=l$ under $Q$, and $F_L(dl) = P(L \in dl) = Q(L \in dl)$ because $Q|_{\F_0} = P|_{\F_0}$.

Write $\alpha^Q$ is the canonical compensator of $N$ with respect to $Q$ and $\F_\T$. The conditional distribution of $N$ given $L=l$ under $Q$, $K^Q(l, d\varphi)$, can be derived by similar arguments to those shown in the proof of Lemma \ref{lemma: P fidi dist}, with the kernel (for fixed $l$) $\alpha^Q(l,\cdot)$ in place of $\alpha$ in that proof. 
Since $(\Omega, \F) = (S \times \N_T^X, \S \otimes \H_T^X)$ and $(L,N) = Id$ (we assume that we are in the canonical setting), we have, by the relationship between the canonical compensator and the compensator, that
\begin{align*}
    \alpha^Q(l,\varphi, dt \times dx) = \alpha^Q\big( (L,N)(l,\varphi), dt \times dx\big) = \L(l,\varphi,dt \times dx)
\end{align*}
for $P$-a.e. $(l,\varphi)$, where $\L$ is the $(Q, \F_\T)$-compensator of $N$, displayed in \eqref{eq: Q compensating measure}. This represents the distribution in terms of the compensator, analogously to Lemma \ref{lemma: P fidi dist}. The derivation of the distributional representation on the right-hand side of \eqref{eq: thm 3 g formula} follows step for step the proof of Lemma \ref{lemma: P fidi dist} and is therefore omitted.  

$\hfill\square$

\section{Notation}
\label{appendix: notation}

\subsection*{General mathematical notation}
\setlength{\tabcolsep}{0pt}
\begin{longtable}{@{}p{0.25\linewidth} p{.75\linewidth}}
    $(\Omega, \F, P)$ & The underlying abstract probability space where $P$ is a fixed, but arbitrary, probability measure \\
    $x \wedge y$, $x \vee y$ & The minimum of $x$ and $y$, and the maximum of $x$ and $y$ \\
    $\A \vee \B$ & The join of two $\sigma$-algebras $\A$ and $\B$; the smallest $\sigma$-algebra which contains both $\A$ and $\B$ \\
    $\bar{ x }_i, \underline{ x }_i$ & The first $i$ components $(x_1, \dots, x_i)$ of a vector $x=(x_1, \dots)$, and the vector $(x_i, \dots)$ subsequent to the $i$'th component of a vector $x=(x_1, \dots)$ \\
    $\mathbb{R}, \Rplus, \mathbb{N}$ & The set of real numbers, non-negative real numbers, and natural numbers excluding zero
\end{longtable}

\subsection*{Point processes}
\begin{longtable}{@{}p{0.25\linewidth} p{.75\linewidth}}
    $\mathcal{I}_d$ & The index set $\{ 1, \dots, d \}$ \\
    $N$, $N^j$ & $N = (N^j)_{j \in \I_d}$ is a multivariate counting process or an MPP; $N^j$ denotes its $j$th component \\
    $T_n, X_n$ & The $n$'th jump time and mark of $N$ \\
    $\F_T,\F_\T$ & The observed $\sigma$-algebra and filtration on $\T$; generated by $N$ before Section \ref{section: identification with multiple interventions in the general MPP setting}, and generated by $L,N$ starting from Section \ref{section: identification with multiple interventions in the general MPP setting} \\
    $\varphi$ & A generic point process trajectory. It may be represented via a double sequence of times and marks and as a counting measure. If the mark set is $\I_d$ it may be represented as a multivariate counting process trajectory (see Appendix \ref{appendix: the canonical space of point process realizations})  \\
    $\varphi|_t, \varphi|_{t-}$ & The restriction of a point process trajectory $\varphi$ to the interval $[0,t]$, resp. $[0,t)$ (see Appendix \ref{appendix: the canonical space of point process realizations}) \\
    $(\N_T^X,\H_T^X)$  & The canonical space of realizations of an MPP on $\T$ with mark set $X$ (see Appendix \ref{appendix: the canonical space of point process realizations}). We use the shorthand notation $\N_T^d := \N_T^{\I_d}$ and $\H_T^d := \H_T^{\I_d}$ \\
    $(S \times \N_T^X, \S \otimes \H_T^X)$  & The canonical space of realizations of $(L,N)$, where $L$ takes values in $(S, \S)$ and $N$ is an MPP on $\T$ with mark set $X$ (see Appendix \ref{appendix: the canonical space of point process realizations}) \\
    $\alpha$, $\alpha^j$ & The canonical compensator of $N$ ($N^j$) with respect to $P$ and $\F_\T$ \\
    $\Lambda$, $\Lambda^j$ & The $(P,\F_\T)$-compensator of $N$ ($N^j$) \\
    $\tilde{\alpha}$ & The canonical compensator of the potential outcome process $\tilde{N}$ (see Section \ref{subsec: setup multiple interventions} and Appendix \ref{appendix: construction of potential outcomes}) \\
    $M$ & The local $(P,\F_\T)$-martingale $N - \Lambda$  \\
\end{longtable}

\subsection*{General process definitions}
\begin{longtable}{@{}p{0.25\linewidth} p{.75\linewidth}}
    $Z_{-}$ & The process defined by $Z_{t-} = \lim_{\epsilon \rightarrow 0^+} Z_{t - \epsilon}$ (assuming left-limits exist) \\
    $\Delta Z$ & The process $Z - Z_{-}$, when $Z$ is a process having left limits \\
    $~^{\tau}Z$ & The process $Z$ stopped at a random time $\tau$, defined by $\{ Z_{\tau \wedge t} \}_{t}$ \\
    $~_{\tau}Z$ & The process $Z$ started at a random time $\tau$, defined by $\{ I(t > \tau)(Z_t -  Z_{\tau}) \}_{t}$\\
    $\bar{ Z }$ & The total process $\bar Z = \sum_{i=1}^k Z^i$ for a process $Z=(Z^1, \dots, Z^k)$, or $\bar Z_t = Z((0,t] \times S)$ for a random measure $Z$ operating on $\T \times S$ \\
    $\E(\cdot)$ & The stochastic exponential; $\E(Z)$ is the unique solution of the SDE $U = 1 + \int_0^{\cdot} U_{s-} dZ_s$ for a semimartingale $Z$ \\
    $[ \cdot, \cdot ]$ & The optional variation process (see Appendix \ref{appendix: Basic definitions}). $[Z] := [Z, Z]$ \\
    $\langle \cdot, \cdot \rangle^{P}$ & The predictable variation under the probability measure $P$ and filtration under study (see Appendix \ref{appendix: Basic definitions}). $\langle Z \rangle^P := \langle Z,  Z \rangle^P$ 
\end{longtable}

\subsection*{Interventions and identifying likelihood ratio processes}
\begin{longtable}{@{}p{0.25\linewidth} p{.75\linewidth}}
    $J$ & Subset of $\I_d$ indexing intervention components (used in Section \ref{section: identification with multiple interventions in the general MPP setting} and later) \\
    $\mathfrak{n}^a$, $\mathfrak{n}^j$ & An intervention on component $a$ (single intervention case) or component $j$ (multiple interventions, $j \in J$) \\
    $\mathfrak{n}^a(N)$, $\mathfrak{n}^j(L,N)$ & The intervention evaluated at the observed trajectories; $\mathfrak{n}^a(N)$ in the single intervention case before Section \ref{section: identification with multiple interventions in the general MPP setting}, $\mathfrak{n}^j(L,N)$ for multiple interventions (starting from Section \ref{section: identification with multiple interventions in the general MPP setting}) \\
    $\tau^a$, $\tau^j$, $\tau^J$ & $\tau^a$ is the first time the observed treatment process deviates from the intervention evaluated in the observed trajectories (see \eqref{eq: tau a} for the definition before Section \ref{section: identification with multiple interventions in the general MPP setting}), and \eqref{eq: tau j deviation time} for the definition starting from Section \ref{section: identification with multiple interventions in the general MPP setting}); $\tau^J := \wedge_{j \in J} \tau^j$ \\
    $\mathbb N^a,\mathbb \Lambda^a, \mathbb K^a$ & $\mathbb N^a = I(\tau^a \leq \cdot)$, $\mathbb \Lambda^a$ is the compensator of $\mathbb N^a$ with respect to $P$ and $\F_\T$, and $\mathbb K^a = -\int_0^\cdot \frac{d\mathbb{N}^a_s - d\mathbb{\Lambda}^a_s}{1 - \Delta\mathbb{\Lambda}^a_s}$ \\
    $\mathbb N^J,\mathbb \Lambda^J, \mathbb K^J$ & For multiple interventions: $\mathbb N^J = I(\tau^J \leq \cdot)$, $\mathbb \Lambda^J$ is the compensator of $\mathbb N^J$ with respect to $P$ and $\F_\T$, and $\mathbb K^J = -\int_0^\cdot \frac{d\mathbb{N}^J_s - d\mathbb{\Lambda}^J_s}{1 - \Delta\mathbb{\Lambda}^J_s}$ \\
    $\E(\mathbb K^a),\E(\mathbb K^J)$ & The identifying likelihood ratio process; $\E(\mathbb K^a) (=W)$ in the single intervention case before Section \ref{section: identification with multiple interventions in the general MPP setting}, and $\E(\mathbb K^J)$ for multiple interventions, starting from Section \ref{section: identification with multiple interventions in the general MPP setting} \\
\end{longtable}

\subsection*{Potential outcomes and outcomes of interest}
\begin{longtable}{@{}p{0.25\linewidth} p{.75\linewidth}}
    $\cfproc{}$ & The potential outcome process (characterized in Appendix \ref{appendix: general MPP potential outcomes}; see also Definition \ref{definition: potential outcome process}) \\ 
    $\tilde \F_T, \tilde \F_\T$ & The $\sigma$-algebra and filtration generated by $\tilde N$ on $\T$ before Section \ref{section: identification with multiple interventions in the general MPP setting}; generated by $(L,\tilde N)$ on $\T$ starting from Section \ref{section: identification with multiple interventions in the general MPP setting}, and generated by $(U,L,\tilde N)$ on $\T$ in Appendix \ref{appendix: general MPP potential outcomes} \\
    $\dot{Y}$ & An optional outcome functional on a given canonical space $(\N_T^X, \H_T^X)$ under study; a stochastic process on $(\N_T^X, \H_T^X)$ satisfying $\dot Y_t(\varphi) = \dot Y_t(\varphi|_t)$ for each $t \in \T$ and $\varphi \in \mathcal{N}_T^X$ (see Section \ref{subsection: outcomes of interest} and \ref{subsec: setup multiple interventions}) \\
    $Y$ & The observed outcome of interest; $Y := \dot{Y}(N)$ \\
    $\tilde{Y}$ & The potential outcome of interest; $\tilde{Y} := \dot{Y}(\tilde{N})$ \\
    $\F_\T^{\tilde{Y}}$ & The filtration $\{\F_t \vee \sigma(\tilde{Y})\}_{t \in \T}$
\end{longtable}

\end{document}